\documentclass[12pt]{article}

\usepackage{amsthm,amsmath,mathrsfs,amsfonts,amssymb,bbm}
\usepackage{url}       %
\usepackage{xcolor}
\usepackage[margin=1in]{geometry}
\usepackage{graphicx}
\usepackage{setspace}
\usepackage{enumitem}
\usepackage{natbib}
\usepackage{todonotes}

\newtheorem{theorem}{Theorem}

\newtheorem{assumption}{Assumption}

\newtheorem{corollary}{Corollary}

\newtheorem{example}{Example}

\newtheorem{lemma}{Lemma}

\newtheorem{proposition}{Proposition}
\newtheorem{remark}{Remark}

\DeclareMathOperator*{\argmin}{argmin}
\DeclareMathOperator{\Tr}{Tr}

\DeclareMathOperator*{\plim}{plim}
\newcommand{\E}{\mathbb{E}}
\newcommand{\C}{\mathbb{C}}

\newcommand{\M}{\mathbf{M}}
\renewcommand{\P}{\mathbf{P}}

\newcommand{\R}{\mathbb{R}}
\newcommand{\Op}{O_P}

\begin{document}

  \title{\bf Nuclear Norm Regularized Estimation of Panel Regression Models\footnote{%
  We thank Michael Jansson, Co-editor, an Associate Editor, and two referees for helpful comments
and suggestions. We are also grateful for comments and suggestions from the participants of the 
  2016 IAAE Conference, 2017 UCLA-USC Conference, 2017 Cambridge INET Panel 
  Conference, 2018 International Conference of Econometrics in Chengdu, 2018 
  EMES, 2018 AMES, and seminars at Boston University, Columbia University, 
  Northwestern University, Oxford University, University of Bath, University of 
  Chicago, UC Riverside, University of Iowa, University of Surrey, Yale 
  University, University Carlos III in Madrid, and CEMFI. We also thank 
  Riccardo D'Adamo and Sarah Moon for improving the presentation.
  Weidner acknowledges support from the Economic and Social Research Council through the ESRC Centre for Microdata Methods and Practice grant RES-589-28-0001 and from the European Research Council grants ERC-2014-CoG-646917-ROMIA and
ERC-2018-CoG-819086-PANEDA.}
}
\date{February 2026}

\author{\setcounter{footnote}{2}
Hyungsik Roger Moon\footnote{
Department of Economics, University of Southern California,
Los Angeles, CA 90089-0253, U.S.A. 
Email: {\tt moonr@usc.edu}.}
\and Martin Weidner\footnote{
Department of Economics \& Nuffield College, University of Oxford, Manor Road,
Oxford OX1 3UQ, U.K., and Nuffield College, Email: \texttt{martin.weidner@economics.ox.ac.uk}.
}}

 \maketitle

\vspace{-1cm}
\begin{abstract}
  \noindent 
  In this paper we investigate panel regression models with interactive fixed effects. We propose two new estimation methods that are based on minimizing convex objective functions. The first method minimizes the sum of squared residuals with a nuclear (trace) norm regularization. The second  method minimizes the nuclear norm of the residuals. We establish the consistency of the two resulting estimators.
  Those estimators have a very important computational advantage compared to the existing least squares (LS) estimator, in that they are defined as minimizers of a convex objective function. 
  In addition, the nuclear norm penalization helps to resolve a potential identification problem for interactive fixed effect models, in particular when the regressors are low-rank and the number of the factors is unknown. 
  We also show how to construct estimators that are asymptotically equivalent to the least squares (LS) estimator in \cite{Bai2009} and \cite{MoonWeidner2017} by using our nuclear norm regularized or minimized estimators as initial values for a finite number of LS minimizing iteration steps. This iteration avoids any non-convex minimization, while the original LS estimation problem is generally non-convex, and can have multiple local minima. 
	\\	
  {\sc Keywords: Interactive Fixed Effects, Factor Models, Nuclear Norm Regularization, Convex Optimization, Iterative Estimation}

\end{abstract}

\maketitle

\section{Introduction}
\label{sec:Intro}

In this paper we consider a linear panel regression model of the form
\begin{align}
      Y_{it} &= \sum_{k=1}^K  \, \beta_{0,k} \, X_{k,it}  +
            \sum_{r=1}^{R_0} \, \lambda_{0,ir} \, f_{0,tr}  \, + \, E_{it} \, ,
      \label{ModelBasic}
\end{align}
where $i=1\ldots N$ and $t=1\ldots T$ label the cross-sectional units and the time periods, respectively,
$Y_{it}$ is an observed dependent variable, $X_{k,it}$ are observed regressors, 
$\beta_0=(\beta_{0,1},\ldots,\beta_{0,K})' $ are unknown regression coefficients,
$f_{0,tr}$ and $\lambda_{0,ir}$ are unobserved factors
and factor loadings, $E_{it}$ is an unobserved idiosyncratic error term,
$R_0$ denotes the number of factors, and $K$ denotes the number of regressors.
The factors and loadings  are  also called interactive fixed effects.
They parsimoniously represent heterogeneity in both dimensions of the panel, and they contain the conventional additive error components as a special case. We assume that $R_0 \ll \min ( N,T )$, and for our asymptotic results we will consider $R_0$ 
 and $K$ as fixed,
as $N,T \rightarrow \infty$. We can rewrite this model in matrix notation as
\begin{align}
   Y & = \beta_0 \cdot X + \Gamma_0 + E,
   \label{model:linear.matrix}
\end{align}
where $\beta_0 \cdot X :=  \sum_{k=1}^K X_k \beta_{0,k}$ and $\Gamma_0 := \lambda_0 f_0'$,
and $Y$, $X_k$, $\Gamma_0$ and $E$ are $N \times T$ matrices, while $\lambda_0$ and $f_0$
are $N \times R_0$ and $T \times R_0$ matrices, respectively.
The parameters $\beta_0$ and $\Gamma_0$ are treated as non-random throughout the whole paper, that is,
all stochastic statements are implicitly conditional on their realization.
Without loss of generality we assume $R_0 = {\rm rank}(\Gamma_0)$.

One widely used estimation technique for interactive fixed effect panel regressions is the least squares (LS) method,\footnote{
The LS estimator in this context is also sometimes called concentrated least squares estimator,
and was originally proposed by \cite{Kiefer1980}.
 } 
which treats $\lambda$ and $f$ as parameters to estimate (fixed effects).\footnote{%
	Other estimation methods of panel regressions with interactive fixed effects  include the quasi-difference approach (e.g., \citealt{HoltzEakin-Newey-Rosen1988}), generalized method of moments estimation
	(e.g. \citealt{AhnLeeSchmidt2001,AhnLeeSchmidt2013}),	
	the common correlated random effect method (e.g., \citealt{Pesaran2006}), the decision theoretic approach (e.g., \citealt{Chamberlain2009}), and Lasso type shrinkage methods on fixed effects (e.g., \citealt{ChengLiaoSchorfheide2016}, \citealt{LuSu2016}, \citealt{SuShiPhillips2016}).
} 
Let the Frobenius norm of an $N \times T$ matrix $A$ be $\|A\|_2 := \left( \sum_{i=1}^N \sum_{t=1}^T A_{it}^2 \right)^{1/2}$.
Then, the LS estimator for $\beta$ reads
\begin{align}
\widehat{\beta}_{{\rm LS},R} &:= 
   \argmin_{\beta \in \mathbb{R}^K}  
   L_R(\beta) ,
    &
     L_R(\beta)
     &:=    \min_{\left\{ \lambda \in \mathbb{R}^{N \times R}, \, f \in \mathbb{R}^{T \times R} \right\}} \,
      \frac{1}{2NT} \left\| Y -\beta \cdot X- \lambda f^{\prime} \right\|_2^2  ,
   \label{DefLSestimator}   
\end{align}
where $R$ is the number of factors chosen in the estimation.
A matrix $\Gamma \in \mathbb{R}^{N \times T}$ can be written as $\Gamma = \lambda f'$, 
for some $\lambda \in \mathbb{R}^{N \times R}$ and $f \in \mathbb{R}^{T \times R}$,
if and only if ${\rm rank}(\Gamma) \leq R $. The profiled least square objective function $ L_R(\beta) $
can therefore equivalently be expressed as
\begin{align}
L_R(\beta) &= 
 \min_{\left\{ \left. \Gamma \in \mathbb{R}^{N \times T} \, \right| \, {\rm rank}(\Gamma) \leq R  \right\}}  \;  \frac{1}{2NT} 
\left\| Y -\beta \cdot X- \Gamma \right\|_2^2.
    \label{DefQLS}
\end{align}
It is known that under appropriate regularity conditions (including exogeneity of $X_{k,it}$ with respect to $E_{it}$), 
for $R \geq R_0$,
and as $N,T \rightarrow \infty$ at the same rate, the LS estimator $\widehat{\beta}_{{\rm LS},R}$ is $\sqrt{NT}$-consistent and asymptotically normal, with a bias in the limiting distribution that can be corrected for
 (e.g., \citealt{Bai2009}, \citealt{MoonWeidner2015, MoonWeidner2017}).

The LS estimation approach is convenient, because it does not restrict the relationship between the unobserved heterogeneity ($\Gamma_0$) and the observed explanatory variables ($X_1,...,X_K$).  However, the calculation of $\widehat{\beta}_{{\rm LS},R} $ requires
solving a non-convex optimization problem. While  $\left\| Y -\beta \cdot X- \Gamma \right\|_2^2$ is a convex function
of $\beta$ and $\Gamma$, the profiled objective function $L_R(\beta)$ is in general not convex in $\beta$,
and can have multiple local minima, as will be discussed in Section~\ref{sec:ConvexProblem} in more detail. The reason for 
the non-convexity is that the constraint ${\rm rank}(\Gamma) \leq R$ is non-convex.
 In practice, researchers often use iterative methods that alternate between estimating $\beta$ given $\Gamma$ and updating $\Gamma$ given $\beta$. However, due to the non-convexity, such iterative procedures may converge to local minima or inconsistent estimators if initialized with arbitrary values. The nuclear norm regularization approach proposed in this paper provides a consistent initial estimator, which can then be used to initialize iterative procedures that are guaranteed to converge to the correct solution.
This non-convexity can also become a serious computational obstacle, for example, when the number of regressors is too large to allow
for a simple grid-search over $\beta$ (for a given $\beta$, the optimization over $\Gamma$ is a principal components problem that
can generally be solved quickly in the linear regression case). For generalizations to non-linear  panel regression models 
(e.g., \citealt{chen2014estimation}, \citealt{ChenFernandezValWeidner2014}), 
the optimization of the  non-convex objective function with respect to the 
high-dimensional factors and loadings
becomes even more challenging.

In this paper, we make several important contributions to the literature of 
panel regression with interactive fixed effects:

Our first contribution is to overcome the non-convexity issue of the least squares estimation problem above
by considering a panel regression with nuclear norm regularization, which provides a convex relaxation of the rank constraint in \eqref{DefQLS}. 
To be more specific, let
$s(\Gamma) := [ s_1(\Gamma) , \allowbreak s_2(\Gamma) , \allowbreak \ldots , 
s_{\min(N,T)}(\Gamma)]$
be the vector of   singular values of $\Gamma$.\footnote{
   The  non-zero singular values of $\Gamma$ are the square roots of non-zero eigenvalues of $\Gamma \Gamma'$.
  Singular values are non-negative by definition.
}
The rank of a matrix is equal to the number of non-zero singular values, that is, 
${\rm rank}(\Gamma) = \left\| s(\Gamma) \right\|_0$, where 
$\left\| v \right\|_0$ equals the number of non-zero elements of the vector $v$
(sometimes calles the ``$\ell^0$-norm'' of $v$).
The nuclear norm of $\Gamma$ is defined by 
$\| \Gamma \|_1 :=  \left\| s(\Gamma) \right\|_1 = \sum_{r=1}^{\min(N,T)} s_r( \Gamma )$,
that is, the nuclear norm of the matrix $\Gamma$ is simply the $\ell^1$-norm of the vector $s(\Gamma)$.\footnote{
The nuclear norm $\| \Gamma \|_1$ is the
convex envelope of ${\rm rank}(\Gamma)$ over the set of matrices with spectral norm at most one,
see e.g.\ \cite{RechtFazelParrilo2010}.
The nuclear norm is also sometimes called trace norm, Schatten 1-norm, or Ky Fan n-norm.
Our index notation for the  nuclear norm $  \| \Gamma \|_1$,
Frobenius norm $  \| \Gamma \|_2$, and spectral norm $  \| \Gamma \|_\infty = \lim_{q \rightarrow \infty} \| \Gamma \|_q$ 
is motivated by the unifying formula  $\| \Gamma \|_q  = \left\{ \sum_{r=1}^{\min(N,T)} \left[s_r(\Gamma)\right]^q \right\}^{1/q}$.
}
A convex relaxation of \eqref{DefQLS} can then be obtained
by replacing the non-convex constraint ${\rm rank}(\Gamma) \leq R$
by the convex constraint $\| \Gamma \|_1 \leq c$, for some constant $c$.
This gives
\begin{align}
& \min_{\left\{  \Gamma \in \mathbb{R}^{N \times T} \, \big| \, \| \Gamma \|_1 \leq c_\psi  \right\}}  \;  \frac{1}{2NT} 
\left\| Y -\beta \cdot X- \Gamma \right\|_2^2
  \nonumber \\ & \qquad \qquad \qquad \qquad
=   \min_{  \Gamma \in \mathbb{R}^{N \times T}}
  \left[ \frac{1}{2NT}  \left\| Y -\beta \cdot X- \Gamma \right\|_2^2 +  \frac{ \psi } {\sqrt{NT}} \left\| \Gamma \right\|_1  \right]
   =: Q_\psi(\beta) ,
    \label{DefQpsi}
\end{align}
where in the second line we replace the constraint on the nuclear norm by a 
nuclear-norm penalty term.\footnote{The normalizations with $1/(2NT)$ and 
$1/\sqrt{NT}$ in \eqref{DefQpsi} are somewhat arbitrary, but turn out to be 
convenient for our purposes.}
Choosing a particular penalization parameter $\psi > 0 $ is equivalent to choosing a particular value for $c= c_\psi$,
and we find it more convenient to parameterize the convex relaxation $Q_\psi(\beta) $ of $L_R(\beta)$
by $\psi$ instead of $c$. 

For a given $\psi>0$ the nuclear-norm regularized estimator reads
\begin{align*}
     \widehat{\beta}_{\psi} &:= 
   \argmin_{\beta \in \mathbb{R}^K}  
   Q_{\psi}(\beta) .
\end{align*}
We also define $\widehat{\beta}_{*} := \lim_{\psi \rightarrow 0}  \widehat{\beta}_{\psi}$ for fixed $N$ and $T$.\footnote{
Here, the limit $\psi \rightarrow 0$ is for fixed $N$ and $T$, and has nothing to do with our large $N$, $T$ asymptotic considerations.
}
We will show in Section~\ref{sec:NucNormMin}
that $ \widehat{\beta}_{*}  =    \argmin_{\beta}    \left\| Y - \beta \cdot X  \right\|_1$,
that is, $\widehat{\beta}_{*}$ can alternatively be obtained by minimizing the 
nuclear norm of $Y - \beta \cdot X$. The   nuclear norm minimizing estimator 
$\widehat{\beta}_{*} $ is novel to the literature as 
far as we know. 

  The second contribution of the paper is to derive asymptotic 
properties of $\widehat{\beta}_{\psi} $ and $\widehat{\beta}_{*} $. 
We establish asymptotic results for $\widehat \beta_\psi$
and $\widehat \beta_*$ when both panel dimensions become large.
Under appropriate regularity conditions we show $\sqrt{\min ( N,T 
)}$-consistency of these estimators.
We also show how to use $\widehat \beta_\psi$
and $\widehat \beta_*$ as initial values for a finite 
iteration procedure (avoiding a non-convex 
optimization) that gives improved estimates that are  asymptotically 
equivalent to the LS estimator.

The third contribution of the paper is to solve a potential identification 
problem for interactive fixed effect models by employing the nuclear norm 
penalization.  
Notice that without  restrictions on the parameter matrix $\Gamma_0$ in 
\eqref{model:linear.matrix},
we cannot separate $\beta_0 \cdot X$ and $\Gamma_0$ uniquely, because for any other parameter $\beta$
we can write
\begin{align*}
	Y &= \beta_0 \cdot  X + \Gamma_0 + E = \beta  \cdot  X + \Gamma(\beta,X) + E, 
\quad	\text{where}  
\quad 
	\Gamma(\beta,X)  := \Gamma_0 - (\beta - \beta_0) X ,
\end{align*}
implying that $(\beta_0, \Gamma_0)$ and $(\beta, 	\Gamma(\beta,X))$ are observationally equivalent. 
If any non-trivial linear combination of the regressors $X_k$ is a high-rank matrix, then
the assumption that 
$R_0 = {\rm rank}( \Gamma_0) \ll \min(N,T)$ is sufficient to identify $\beta_0$, because
 ${\rm rank}[\Gamma(\beta,X)]$ will be large for any other value of $\beta$.
However, if some of the regressors $X_k$ have low rank (as, for example, in 
\citealt{GobillonMagnac2016}, where $X_k$ is a panel of treatment variables) 
and the true number of factors $R_0$ is unknown,
then there is an identification problem, and some regularization device is needed to resolve this.
In Section~\ref{sec:Motivation} we show that,  under appropriate assumptions on the covariates, the nuclear norm penalization indeed provides  such a regularization device
to uniquely identify $\beta_0$.
 
Nuclear norm penalized estimation has been widely studied in the machine 
learning and statistical learning literature.
There, the parameter of interest is usually the matrix that we call $\Gamma$ in our model.
In particular, there are many papers that use this penalization method in matrix completion
(e.g., \citealt{RechtFazelParrilo2010} and \citealt{Hastieetal2015} for recent surveys),
and for reduced rank regression estimation
(e.g., \citealt{RohdeTsybakov2011}).
More recently, nuclear norm penalization has also been used in the econometrics literature:
\cite{bai2019rank} use it to improve estimation in a pure factor 
model.
\cite{Atheyetal2017} apply  nuclear norm penalization to treatment effect estimation with unbalanced panel data due to missing observations together with a regularization on the high dimensional regression coefficients -- their primary interest is to predict the left-hand side variable using the regularization.  \cite{chernozhukov2018inference} consider panel regression models with heterogeneous coefficients, while in this paper
we focus on panel regression with homogenous coefficients.	
To the best of our knowledge, our results here on the estimates of the common regression
coefficients $\beta_0$ are new in this literature,  and the  nuclear norm minimizing estimator $\widehat \beta_*$ has also not been proposed previously.

Since the 2018 working paper version of this paper, many related studies have used nuclear norm regularization and similar methods in panel data, factor models, and network models. In network settings, \cite{alidaee2020recovering} and \cite{ma2022detecting} use penalized methods to recover structure in networks. In panel data models, \cite{chernozhukov2018inference} use low-rank methods to study models with varying slopes, and \cite{belloni2019high}, \cite{wang2022low}, and \cite{feng_2023} extend this to quantile regressions, which relates to our work in Appendix~\ref{sec:Nonlinear:appendix}. \cite{miao2020panel} and \cite{miao2023high} look at threshold models and VARs with factor structures. \cite{chen2022unified} gives a general framework for conditional factor models using nuclear norm regularization, and also shows how to choose tuning parameters and prove consistency. \cite{beyhum2022factor} allow for weak (non-strong) factors, and \cite{fernandez2021low} use nuclear norm penalties to estimate nonseparable panel models. \cite{beyhum2019square} treat the low-rank structure as an approximation, allowing the rank to grow with sample size. \cite{chetverikov2022spectral} and \cite{mugnier2025simple} study grouped fixed effects, and use nuclear norm methods in the first step of their slope estimation. \cite{armstrong2022robust} also use the nuclear norm to make slope estimates more robust in interactive fixed effects models.
 \cite{hong2023profile} develop a profile GMM method that uses nuclear norm regularization and allows for endogeneity. \cite{vogt2022cce} build a high-dimensional version of the CCE estimator, combining factor projections with Lasso. \cite{gao2024robust} study robust inference when the errors are not Gaussian and may be dependent over time and across units, using thresholded long-run covariance estimation. Finally, \cite{chen2022unified} (already mentioned above) develop a flexible method that applies to many types of conditional factor models.
Overall, this recent work shows how nuclear norm regularization and related convex methods are useful for many problems involving panel data, factor models, and networks.

The paper is organized as follows. Section~\ref{sec:Motivation} discusses theoretical motivations for employing nuclear norm regularization instead conventional rank restrictions. In Section~\ref{sec:Consistency}  we  derive consistency results on $\widehat{\beta}_{\psi}$ and $\widehat{\beta}_{*}$.
Section~\ref{sec:post estimation}  shows how to use these two estimates as a preliminary step to construct an estimator through iterations that achieves asymptotic equivalence to the fixed effect estimator. 
Section~\ref{sec:MC} investigates finite sample properties via simulations,  
and Section~\ref{sec:conc} concludes the paper. All technical derivations and proofs are presented in the appendix
or supplementary appendix.

\section{Motivation of Nuclear Norm Regularization}
\label{sec:Motivation}

In this section we provide further motivation and explanation of the nuclear norm regularized estimation method.
This estimation approach comes with the 
computational advantage of having a convex objective function, and
it also provides a solution to the identification problem of interactive fixed effect models
with low-rank regressors. 
\subsection{Convex Relaxation}
\label{sec:ConvexProblem}

We have already introduced the profile LS objective function $L_R(\beta)$ 
and its convex relaxation $Q_\psi(\beta)$ in the introduction.
Here, we explain those objective functions further.
Firstly,  we want to briefly explain why $Q_\psi(\beta)$
is indeed convex.
We have introduced the nuclear norm as
$\| \Gamma \|_1 :=  \sum_{r=1}^{\min(N,T)} s_r( \Gamma )$,
but it is not obvious from this definition that $\| \Gamma \|_1$ is convex in $\Gamma$, because the singular values $s_r( \Gamma )$
themselves are generally not convex functions of $\Gamma$, except for $r=1$. A  useful alternative definition of the nuclear norm is 
\begin{align}
     \| \Gamma \|_1  =  \max_{\left\{ A \in \mathbb{R}^{N \times T} \big|  \|A\|_{\infty} \leq 1 \right\}} {\rm Tr}(A' \, \Gamma) ,
     \label{NuclearNormAlternativeDef}
\end{align}
that is, the nuclear norm is dual to the spectral norm $\| \cdot \|_\infty$.
 From this  it is easy to see that 
$\| \cdot \|_1$ is indeed a matrix norm, and thus convex in $\Gamma$.\footnote{
 Let $B$ and $C$ be matrices of the same size. Then, by \eqref{NuclearNormAlternativeDef} there exists a matrix $A$ of the same size
 with $\|A\|_{\infty} \leq 1$  
 such that $\|B+C\|_1 = {\rm Tr}[A' \, (B+C)] =  {\rm Tr}(A' B) + {\rm Tr}(A' C) \leq \|B\|_1 + \|C\|_1 $,
 which is the triangle inequality for the nuclear norm. Together with absolute homogeneity of 
 $\left\| \cdot \right\|_1$ this implies convexity.
}
Therefore, the nuclear norm regularized objective function
\[ 
	\frac{1}{2NT}  \left\| Y -\beta \cdot X- \Gamma \right\|_2^2 +  \frac{ \psi } {\sqrt{NT}} \left\| \Gamma \right\|_1
\] 
as a function of $(\beta,\Gamma)$ is convex. 
Profiling with respect to $\Gamma$ preserves convexity, that is, $Q_\psi(\beta)$ is also convex.

By contrast, the least squares objective function $\frac{1}{2NT} \left\| Y 
-\beta \cdot X- \lambda f^{\prime} \right\|_2^2$
is generally non-convex in the parameters $\beta$, $\lambda$ and $f$.
However, the non-convexity of the LS minimization over $\lambda$ and $f$ is actually not a serious problem in computing the profile objective function  $L_R(\beta)$, as long as the
regression model is linear and one of the dimensions $N$ or $T$ is not too large.\footnote{The optimal $\widehat \lambda$ and $\widehat f$ are simply given by the leading $R$ principal components of $Y - \beta \cdot X$. Calculating them requires finding the eigenvalues and eigenvectors of either the $N\times N$ matrix $(Y - \beta \cdot X) (Y - \beta \cdot X)'$
or the $T \times T$ matrix $(Y - \beta \cdot X)' (Y - \beta \cdot X)$, which takes at most a few seconds on modern
computers, as long as $\min(N,T) \lessapprox 5.000$, or so. The non-zero eigenvalues of 
$(Y - \beta \cdot X) (Y - \beta \cdot X)'$ and $(Y - \beta \cdot X)' (Y - \beta \cdot X)$ are identical,
and are equal to the square of the non-zero singular values of $Y - \beta \cdot X$.} Recall that $s_r(Y - \beta \cdot X)$ is the $r^{th}$ largest singular value of the matrix $(Y - \beta \cdot X)$, for $r=1,...,\min(N,T)$. One can show (see  \citealt{MoonWeidner2017}) that the profile least squares objective function is
\begin{align}
       L_R(\beta) &= \frac 1 {2 \, NT} \sum_{r=R+1}^{\min(N,T)} [s_r( Y - \beta \cdot X)]^2 ,
     \label{RewriteQLS0}
\end{align} 
where the largest $R$ singular values are omitted in the sum - because they are 
absorbed by the
 principal component estimates $\widehat \lambda$ and $\widehat f$. 
The remaining problem in calculating $\widehat \beta_{{\rm LS},R}$ is the 
 generally non-convex minimization of $L_R(\beta) $ over $\beta$.\footnote{
  In our discussion here we focus on the calculation of $\widehat \beta_{{\rm LS},R}$ via
  minimization of the profile objective function $L_R(\beta)$. More generally,
  $\widehat \beta_{{\rm LS},R}$ can be obtained by any method that minimizes 
  $ \left\| Y -\beta \cdot X- \lambda f^{\prime} \right\|_2^2$ over $\beta$, $\lambda$, $f$,
  see e.g.\ \cite{Bai2009} or the supplementary appendix in \cite{MoonWeidner2015}.
  For any such method the non-convexity of the objective function is a potential problem, 
  because the algorithm may converge to a local minimum, or potentially even to a critical point that is not a local minimum.
}
To illustrate the potential difficulty caused by this non-convexity, in Figure \ref{figure:LS_Nuc_together} we plot $L_R(\beta)$ for the simple example described in Appendix~\ref{subsec:ap.ex.nonconvex}. 
In this example $L_R(\beta) $ is non-convex and has two local minima, one of which (the global one) is close to the
true parameter $\beta_0=2$. The figure also shows that $Q_\psi(\beta)$ is convex and only has a single local minimum.

	\graphicspath{ {Figures/LS_Comp_Example/} }
	\begin{figure}
	\begin{center}
		\includegraphics[width=0.5\textwidth]{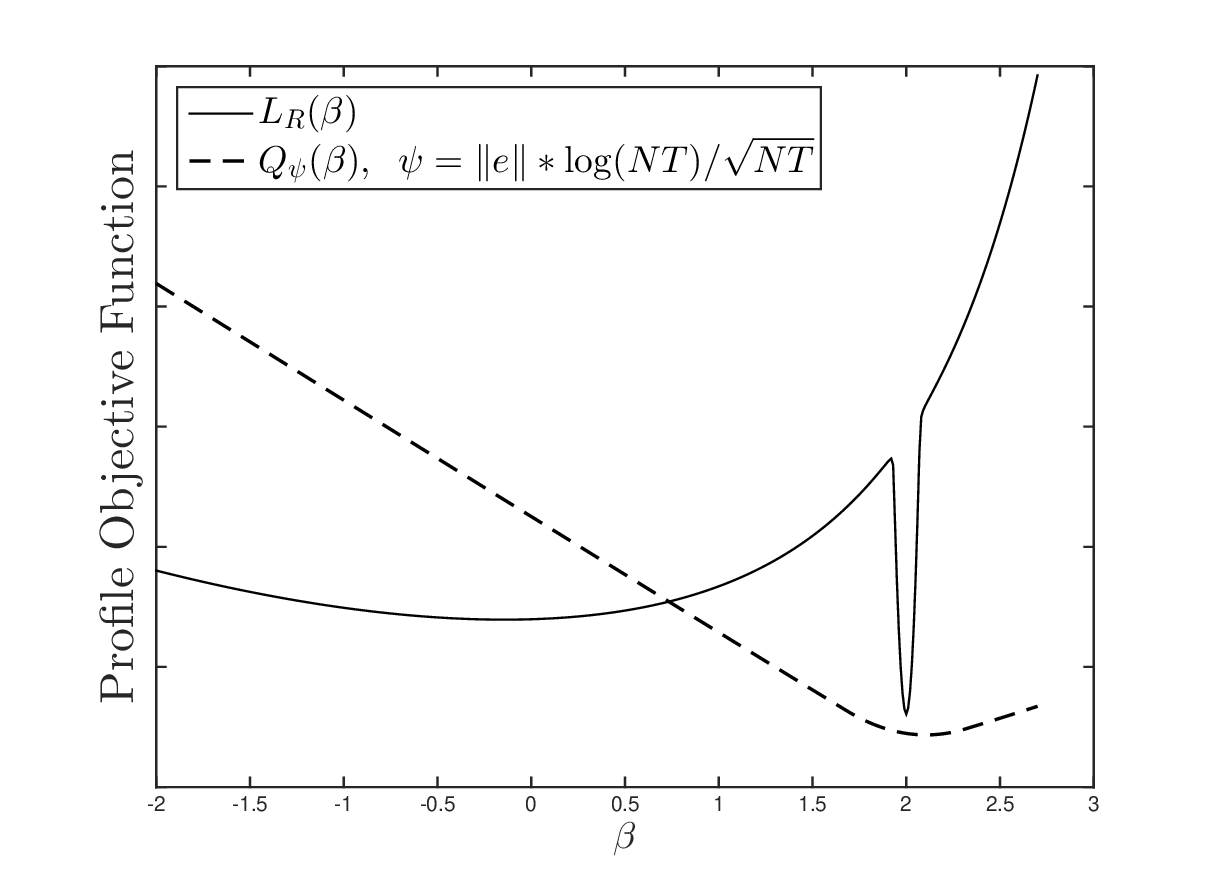}
		\caption{\label{figure:LS_Nuc_together}Plot of $L_R(\beta)$ and $Q_{\psi}(\beta)$ for the example detailed in Appendix~\ref{subsec:ap.ex.nonconvex}. The true parameter is $\beta_0=2$.}
	\end{center}
	\end{figure}

For any $\psi > 0$
define the functions $\ell_\psi: [0,\infty) \mapsto [0,\infty)$ and $q_\psi: [0,\infty) \mapsto [0,\infty)$
by
\begin{align}
\ell_\psi(s)  & := \left\{ 
\begin{array}{ll}
\frac 1 {2} \; s^2 , & \text{for $s <  \psi$,}
\\
0  ,  & \text{for $s \geq \psi$,}
\end{array}
\right. 
   &
q_\psi(s) := \left\{ 
\begin{array}{ll}
\frac 1 {2} \; s^2 , & \text{for $s <  \psi$,}
\\
\psi s   - \frac {\psi^2} 2  ,  & \text{for $s \geq \psi$.}
\end{array}
\right.    
    \label{DefSmallLQ}
\end{align} 
For an $N \times T$ matrix $A$ let 
$
    \ell_\psi(A)   := \sum_{r=1}^{\min(N,T)} \ell_\psi(s_r(A))$ 
 and   
$
     q_\psi(A) := \sum_{r=1}^{\min(N,T)} q_\psi(s_r(A))$.
We can then rewrite \eqref{RewriteQLS0} as
\begin{align}
     L_R(\beta)  &=   \ell_{\psi(\beta,R)}\left( \frac {Y - \beta \cdot X} {\sqrt{NT}}    \right) ,
     \label{RewriteQLS}
\end{align}
where $\psi(\beta,R)$ satisfies
 \begin{align}
    s_{R+1}\left( \frac {Y - \beta \cdot X} {\sqrt{NT}}    \right)  < \psi(\beta,R) \leq  s_{R}\left( \frac {Y - \beta \cdot X} {\sqrt{NT}}    \right) .
    \label{RelationRpsi}
\end{align}
Here, the normalization with $1/\sqrt{NT}$ is natural, because under standard assumptions the largest singular value of $Y - \beta \cdot X$
is of order $\sqrt{NT}$, as $N$ and $T$ grow.
The formulation \eqref{RewriteQLS}  is interesting for us,
because  the following lemma shows that we have a very similar representation for $Q_\psi(\beta) $.
\begin{lemma}
    \label{lemma:QPsiRewrite}
    For any $\beta \in \mathbb{R}^K$ and any $\psi > 0$
    we have
    \begin{align*}
         Q_\psi(\beta) &= q_\psi\left( \frac {Y - \beta \cdot X} {\sqrt{NT}}    \right) .
    \end{align*}   
\end{lemma}

\graphicspath{ {Figures/} }
\begin{figure}[tb]
\centering
\includegraphics[width=0.4\textwidth]{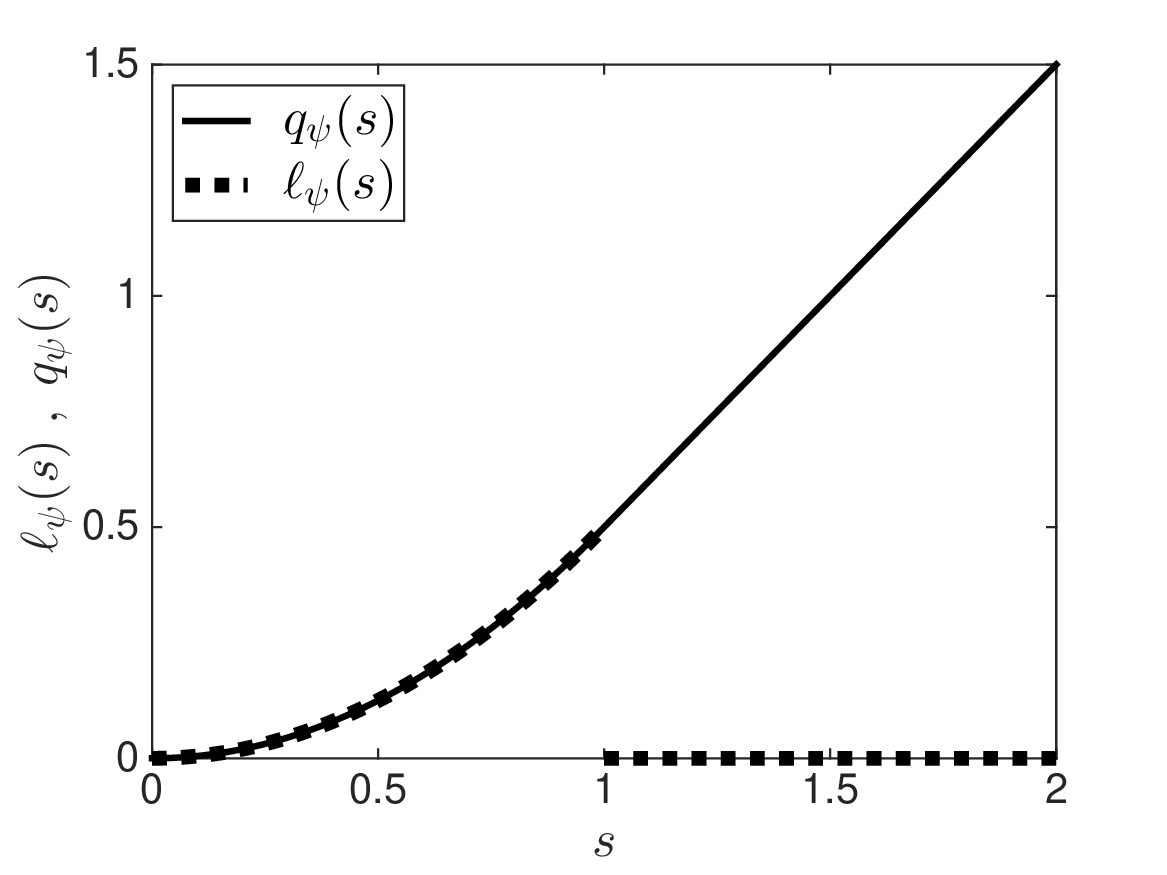}
\caption{\label{fig:FunctionsGH}Plot of the functions $q_\psi(s)$ and $\ell_\psi(s)$ for $\psi=1$.}
\end{figure}

The proof is given in the appendix.
 We note that Lemma~\ref{lemma:QPsiRewrite} implies convexity of $Q_\psi(\beta)$ in $\beta$, since $(Y - \beta \cdot X)/\sqrt{NT}$ is linear in $\beta$ and $A \mapsto \sum_{r} q_\psi(s_r(A))$ is convex in $A$. The latter follows from a standard result in convex analysis: if $g: \mathbb{R} \rightarrow \mathbb{R}$ is convex and nondecreasing on $[0,\infty)$, then $A \mapsto \sum_{r} g(s_r(A))$ is convex in the matrix $A$ (see, e.g., \citealt{lewis1995convex}).
Figure~\ref{fig:FunctionsGH} shows the functions $q_\psi(s)$ and $\ell_\psi(s)$ for real valued arguments $s$ and $\psi=1$.

Comparing $L_R(\beta)$ and $Q_\psi(\beta)$
we see that  the parameter $R$ that counts the number of factors
    is replaced by the parameter $\psi$ that characterizes the magnitude
       at which the singular values of $(Y -\beta \cdot X) / \sqrt{NT}$ are considered to be factors,
and for a given $\beta$ the relationship between $R$ and $\psi$ is given by \eqref{RelationRpsi}.        
Large $R$ corresponds to small $\psi$, and vice versa.
Furthermore, $\widehat{\Gamma}_{\psi}(\beta) := \argmin_{\Gamma} \frac{1}{2NT}  \left\| Y -\beta \cdot X- \Gamma \right\|_2^2 +  \frac{ \psi } {\sqrt{NT}} \left\| \Gamma \right\|_1$ has singular values\footnote{See Lemma \ref{lemma:MinGamma} in the supplementary appendix
for details.}
       \[ 
       		s_r\left( \widehat{\Gamma}_{\psi}(\beta) \right) = \max \left(s_r\left( (Y -\beta \cdot X) / \sqrt{NT} \right) - \psi,0 \right) \quad r  = 1,...,\min(N,T),
       \]
that is, the nuclear norm penalization  shrinks the singular values of       $Y 
-\beta \cdot X$ towards zero by a fixed amount.

    Fixing $\psi$ as opposed to fixing $R$ already changes the functional form of the profile objective function,
       because according to \eqref{RelationRpsi} their relationship depends on $\beta$.              
In addition, the objective function is convexified by replacing the function $ \ell_\psi(s)$ that is applied
    to the singular values of $(Y - \beta \cdot X) / \sqrt{NT}$ with
    the function $q_\psi(s)$, as defined in \eqref{DefSmallLQ}.
       The function $q_\psi(s)$ provides a convex continuation of $ \ell_\psi(s)$ for $s \geq \psi$.

Using the closed-form expression for $Q_{\psi}(\beta)$ in Lemma \ref{lemma:QPsiRewrite}, and noticing that it is convex in $\beta$, one can compute the minimizer $\widehat{\beta}_{\psi}$ of $Q_{\psi}(\beta)$ using various optimizing algorithms for a convex function (see chapter 5 of \citealt{Hastieetal2015}). If the dimension of $\beta$ is small, then one may even use a simple grid search method to find $\widehat{\beta}_{\psi}$.
 We will discuss a data dependent choice of the penalty parameter $\psi$ in Section~\ref{sec:MC}.

\subsection{Unique Matrix Separation}
\label{sec:MatrixSeparation}

When estimating the interactive fixed effect model \eqref{ModelBasic}, in 
practice
both $\beta_0$ and $R_0$ are unknown. Showing that $\beta_0$ and $R_0$ can be  consistently estimated jointly
is a difficult problem in general.\footnote{%
The problem of joint identification of $\beta_0$ and $R_0$ is often avoided in the literature.
Some papers (e.g. \citealt{Bai2009}, \citealt{li2016panel}, \citealt{MoonWeidner2017})
 assume that the number of factors $R_0$ is known when showing consistency 
for an estimator of $\beta_0$.
Alternatively, \cite{LuSu2016} allow for unknown $R_0$, but assume consistency of their estimator for $\beta_0$.
}
Within the interactive fixed effects estimation framework this joint inference problem has only been successfully addressed when both
of the following assumptions are satisfied:\footnote{
Some existing estimation methods avoid specifying $R$ when estimating $\beta_0$, but always at the cost of
 some additional assumptions on the data generating process.  For example, the common correlated effects estimator
 of \cite{Pesaran2006} avoids choosing $R$, but requires assumptions on how the factors $f_0$ enter into the 
 observed regressors $X_k$, and requires all regressors of interest to be high-rank.
}
\begin{itemize}
     \item[(C1)] There is a known upper bound $R_{\max}$ such that $R_0 \leq R_{\max}$.
     \item[(C2)] All the regressors $X_k$ are ``high-rank regressors'', that is, ${\rm rank}(X_k)$ is large for all $k$.
\end{itemize}
Under those assumptions (and other  regularity conditions) 
the consistency proofs of  \cite{Bai2009} and \cite{MoonWeidner2015} are applicable to the LS estimator  for $\beta$
that uses $R=R_{\max} \geq R_0$ factors in the estimation, and
one can also  show  the convergence rate result
$\big\| \widehat \beta_{{\rm LS},R_{\max}} - \beta_0 \big\| = O_P \left( \min(N,T)^{-1/2} \right)$,  as $N,T \rightarrow \infty$.
To obtain a consistent estimator for $R_0$ one can then apply inference methods from pure factor models without regressors
 (e.g.\ \citealt{BaiNg2002}, \citealt{Onatski2010}, \citealt{AhnHorenstein2013}) 
 to the matrix $Y - \widehat \beta_{{\rm LS},R_{\max}} \cdot X$.

The condition (C2) above is particularly strong, because ``low-rank regressors'' are quite common in practice.
If we can write $X_{k,it} = w_{k,i} v_{k,t}$, then we have ${\rm rank}(X_k) = 1$, and the condition (C2) is violated.
For example, \cite{GobillonMagnac2016} estimate an interactive fixed effects model in a panel treatment effect setting,
where the main regressor of interest indeed can be multiplicatively decomposed in this way,
with  $w_{k,i}$ being the treatment indicator of unit $i$, and $v_{k,t}$ being the time indicator of treatment. Those 
 interactive fixed effects models for panel treatment effect applications have  grown very popular recently.\footnote{%
Other recent applications in the same vein as \cite{GobillonMagnac2016} are \cite{chan2016policy}, \cite{powell2017synthetic}, \cite{gobillon2017local},
\cite{adams2017identification}, \cite{piracha2017immigration},
 \cite{li2018inference}, to list just a few. 
This literature is also related to the synthetic control method 
(\citealt{abadie2003economic}, \citealt{abadie2010synthetic}, \citealt{abadie2015comparative};
see also \citealt{hsiao2012panel}).
}
However, when $R_0$ is unknown, then the presence of such low-rank regressors creates an identification problem,
as illustrated by the following example.

\begin{example}
	\label{ex:ID}
	Consider  a single ($K=1$)
	low-rank regressor $X_1=v w'$, with vectors $v \in \mathbb{R}^N$ and $w \in \mathbb{R}^T$. 
	Let $R_{\bigstar} = R_0 + 1$, $\lambda_{\bigstar} = [\lambda_0, v ]$, and $f_{\bigstar} = [f_0,  (\beta_{0,1} - \beta_{\bigstar,1}) w]$.
	Then, model \eqref{ModelBasic} with parameters $\beta_0$, $R_0$, $\lambda_0$, $f_0$
	is observationally equivalent to the same model with parameters $\beta_{\bigstar}$, $R_{\bigstar}$, $\lambda_{\bigstar}$, $f_{\bigstar}$,
	because we have $\beta_{0,1} X_1 + \lambda_0 f_0' = \beta_{\bigstar,1} X_1 + \lambda_\bigstar f_\bigstar'$.
	Thus,  $\beta_0$ is observationally equivalent to any other value $\beta_\bigstar$
	if the true number of factors is unknown.
\end{example}
The example shows that regression coefficients of low-rank regressors are  not identified if $R_0$ is unknown, because 
$\beta \cdot X$ could simply be absorbed into the factor structure $\lambda f'$, which is also a low-rank matrix. Therefore, without some additional assumption or regularization device, the two low-rank matrices $ \beta_0 \cdot X$ and $\Gamma_0 = \lambda_0 f_0'$ cannot be uniquely disentangled, which is what we mean by ``unique matrix separation'' in the title of this section.

\subsubsection*{Nuclear Norm Minimizing Estimation}
\label{sec:NucNormMin}

Both the nuclear norm regularized estimator $\widehat{\beta}_{\psi}$ and the nuclear norm minimizing estimator $\widehat{\beta}_{*}$ can be computed without specifying the number of factors $R_0$, thus avoiding the restriction~(C1). In this subsection we characterize $\widehat{\beta}_{*}$ and provide intuition for why these estimators can achieve consistency even when $R_0$ is unknown.
We already introduced $\widehat{\beta}_{*} = \lim_{\psi \rightarrow 0}  \widehat{\beta}_{\psi}$ in Section~\ref{sec:Intro}.
Using Lemma~\ref{lemma:QPsiRewrite} we can now characterize $\widehat{\beta}_{*}$ differently.
It is easy to see that $ \lim_{\psi \rightarrow 0}  \psi^{-1} q_\psi(s) = s$, for $s \in [0,\infty)$,
and therefore $ \lim_{\psi \rightarrow 0}  \psi^{-1} q_\psi(A) =  \| A\|_1$, for $A \in \mathbb{R}^{N \times T}$.
Lemma~\ref{lemma:QPsiRewrite} thus implies that
$\lim_{\psi \rightarrow 0}  \psi^{-1} Q_\psi(\beta)  = \| (Y - \beta \cdot X) / \sqrt{NT} \|_1$.
Another way to see this is as follows. 
According to \eqref{RelationRpsi},
the limit $\psi \rightarrow 0$
corresponds to choosing $R$ very large, i.e., $R = \min(N,T)$. In this case, $\widehat{\Gamma}_{\psi}(\beta) = Y - \beta \cdot X$, and the profile objective function is $Q_{\psi}(\beta) = \frac{\psi}{\sqrt{NT}} \| \widehat{\Gamma}_{\psi} (\beta) \|_1 = \frac{\psi}{\sqrt{NT}} \| Y - \beta \cdot X \|_1.$ From this we deduce $\lim_{\psi \rightarrow 0} \psi^{-1} Q_{\psi}(\beta) = \frac{\| Y - \beta \cdot X \|_1}{\sqrt{NT}}.$

Notice that for $\psi=0$ we   trivially have $Q_0(\beta) = 0$,
but the rescaled objective function $  \psi^{-1} Q_\psi(\beta) $ nevertheless has a non-trivial limit as $\psi \rightarrow 0$.
Since rescaling the objective function by a constant does not change the minimizer we
thus find that
\begin{align}
      \widehat{\beta}_{*}
      =    \argmin_{\beta \in \mathbb{R}^K}    \left\| Y - \beta \cdot X  \right\|_1,
     \label{NuclearNormMin}
\end{align}
that is, the small $\psi$ limit of the nuclear norm regularized estimator $ \widehat{\beta}_{\psi}$   is 
the nuclear norm minimizing estimator $\widehat{\beta}_{*}$. The objective function  $\left\| Y - \beta \cdot X  \right\|_1$ is
convex in $\beta$.

We cannot expect the LS estimator $\widehat{\beta}_{{\rm LS},R} $
to have good properties (in particular consistency) if we choose the number of factors equal to, or close to, its maximum possible
value $R=\min(N,T)$.
It is therefore somewhat surprising that $ \widehat{\beta}_{\psi}$ has a well-defined limit as $\psi \rightarrow 0$,
and that we are able to show consistency of the limiting estimator $\widehat{\beta}_{*}$ under appropriate
regularity conditions in the following sections, because the resulting estimator for $\Gamma$ is certainly
not consistent for $\Gamma_0$ in that limit.\footnote{
The $\psi \rightarrow 0$ limit (for fixed $N$, $T$) of the optimal $\Gamma$ in \eqref{DefQpsi}
is   $Y -  \widehat{\beta}_{*} \cdot X$, which as $N$ and $T$ grow converges
to $\lambda_0 f_0'  + E$ for consistent $\widehat{\beta}_{*}$, that is,
the estimator for $\Gamma$ that corresponds to $\widehat{\beta}_{*}$ is not consistent for $\lambda_0 f_0' $. 
}

The main significance of $\widehat{\beta}_{*}$ is that it provides an estimator for $\beta$
that does not require any choice of ``bandwidth parameter'', because neither $R$ nor $\psi$ needs to be specified.
It thus provides a method to estimate $\beta_0$ consistently
without requiring knowledge of an upper bound on $R_0$ as in the condition (C1) above.
In a second step we can then estimate $R_0$ consistently by applying, for example,
the \cite{BaiNg2002} method for pure factor models without regressors to the matrix $Y - \widehat{\beta}_{*} \cdot X$.

Notice that the pooled OLS estimator $\beta_0$ minimizes $\| Y - \beta \cdot X \|_2^2 = \sum_{r=1}^{\min(N,T)} s_r( Y - \beta \cdot X)^2$,
 the $\ell^2$-norm of the singular values of the residual matrix, $Y - \beta \cdot X$, while the nuclear norm minimizing estimator $\widehat{\beta}_*$ minimizes the $\ell^1$-norm, $\| Y - \beta \cdot X \|_1 = \sum_{r=1}^{\min(N,T)} s_r( Y - \beta \cdot X),$ of the residual matrix. 
 The relationship between these two estimators is therefore analogous to that of the OLS estimator and the LAD (least absolute deviation) estimator
for cross-sectional samples.
$\widehat{\beta}_*$ is robust with respect to the unobserved factors, which are ``outliers" in the singular value spectrum, while the pooled OLS estimator is not robust
 towards the presence of those unobserved factors (because they may be correlated with the regressors).

To understand intuitively why $\widehat{\beta}_{*}$ can be consistent, consider that at any $\beta \neq \beta_0$ the residual matrix is $Y - \beta \cdot X = \Gamma_0 + E + X \cdot (\beta_0 - \beta)$. For $\widehat{\beta}_{*}$ to correctly identify $\beta_0$, we need that $\| Y - \beta \cdot X \|_1 > \| Y - \beta_0 \cdot X \|_1 = \| \Gamma_0 + E \|_1$ for all $\beta \neq \beta_0$, which requires that the term $X \cdot (\beta_0 - \beta)$ does not cancel out singular values of $\Gamma_0 + E$. Since $\Gamma_0$ is low-rank and $E$ is generically full-rank, this imposes different requirements depending on the rank of $X$: (i) if $X$ is low-rank, then $X$ must operate under sufficiently different row and column spaces than $\Gamma_0$, so that $X \cdot (\beta_0 - \beta)$ cannot be absorbed into the factor structure; (ii) if $X$ is high-rank, then $X$ must be sufficiently different from $E$, so that $X \cdot (\beta_0 - \beta)$ does not simply offset the idiosyncratic errors. The formal conditions in Section~\ref{sec:Consistency} make these requirements precise. Notice also that even when $\widehat{\beta}_{*}$ is consistent for $\beta_0$, the corresponding residual $Y - \widehat{\beta}_{*} \cdot X$ estimates $\Gamma_0 + E$, not $\Gamma_0$ alone.

\subsubsection*{Nuclear Norm Penalization Approach for Matrix Separation}
\label{subsec:MatrixSep}

Next, we provide an identification result showing how the nuclear norm regularization approach helps to overcome
the restrictions (C2) above, that is, how to estimate regression coefficients 
for low-rank regressors when $R_0$ is unknown.
The goal is to provide conditions on the regressors $X_k$ under which the nuclear norm penalization approach 
indeed solves the matrix separation problem for low-rank regressors
and interactive fixed effects.

To build intuition, we first present an identification result based on the expected objective function for fixed $N$ and $T$.
We consider
\begin{align}
\bar \beta_\psi  := \argmin_{\beta} \min_\Gamma 
  \left\{ \frac{1}{2NT} \, \mathbb{E}\left[  \left\| Y -\beta \cdot X- \Gamma \right\|_2^2 \Big| X \right]
   +  \frac{ \psi } {\sqrt{NT}} \left\| \Gamma \right\|_1  \right\} .
   \label{DefBarBeta}
\end{align}
Here, the expectation is conditional on all the regressors $(X_1,\ldots,X_K)$, and also implicitly on
all the parameters $\beta_0$ and $\Gamma_0$, because those are treated as non-random.\footnote{
$\bar \beta_\psi$ can be viewed as a population version of $\widehat \beta_\psi$
for an appropriately defined population distribution of $Y$ conditional on $X$.
Proposition~\ref{prop:unique.separation} below establishes identification of $\beta_0$ for this population objective function. The sample analog results for $\widehat{\beta}_{\psi}$ and $\widehat{\beta}_{*}$ are established in Section~\ref{sec:Consistency} using different proof techniques.
}

For a matrix $A$, let $\P_A := A (A'A)^{\dagger} A'$ and $\M_A := \mathbf{I}  - \P_A$ be
the projectors onto and orthogonal to the column span of $A$,
where $\mathbf{I} $ is the identity matrix of appropriate dimensions,
and $\dagger$ refers to the Moore-Penrose generalized inverse.
Remember also our notation $\alpha \cdot X := \sum_{k=1}^K \alpha_k X_k$
for $\alpha \in \mathbb{R}^K$. For vectors $v$ we write $\|v\|$ for the Euclidian norm.

\begin{proposition}\label{prop:unique.separation}
	Suppose that $N$, $T$, $R_0$ and $K$ are fixed.
	Let $\mathbb{E}( \left.  E_{it} \, \right|  \, X   )  = 0$, and $  \mathbb{E}\left(  \left. E_{it}^2  \, \right|  \, X    \right) < \infty$, for all $i,t$.
             For all $\alpha \in \mathbb{R}^{K} \setminus \{0\}$ assume that
	\begin{equation}
		\left\|   \M_{\lambda_0} (\alpha \cdot X) \M_{f_0} \right\|_1 
		>   \left\| \P_{\lambda_0} (\alpha \cdot X) \P_{f_0} \right\|_1  . 
		\label{prop:unique.separation.regularity.cond}
	\end{equation}
	Then, 	$\left\| \bar \beta_\psi - \beta_0 \right\| = O(\psi)$, as $\psi \rightarrow 0$.
\end{proposition}

The proof is given in the appendix. 
The proposition considers fixed $N$, $T$, with only $\psi \rightarrow 0$.\footnote{
Display \eqref{BoundDeltaBeta} in the appendix
 provides a bound on $\| \bar \beta_\psi - \beta_0 \|$
for finite $\psi$, but the limit $\psi \rightarrow 0$ is what matters most to us, because
that limit allows to identify $\beta_0$.
}
The statement of the proposition implies that $\lim_{\psi \rightarrow 0} \bar \beta_\psi = \beta_0 $.
Thus, the proposition provides conditions under which the nuclear norm regularization approach identifies
the true parameter $\beta_0$.
The proposition does not restrict the rank
of the regressors, so the result is applicable to both low-rank and high-rank regressors.
The assumption $\mathbb{E}( \left.  E_{it} \, \right|  \, X   )  = 0$
 requires strict exogeneity of all regressors, but we will allow for 
  pre-determined regressors in consistency results of Section~\ref{sec:GeneralX} below.

The beauty of Proposition~\ref{prop:unique.separation} is that it provides a very easy to interpret
non-collinearity condition on the regressors $X_k$. It requires that for any linear combination
of the regressors the part $  \M_{\lambda_0} (\alpha \cdot X) \M_{f_0}$, which cannot be explained
by neither $\lambda_0$ nor $f_0$, is larger in terms of nuclear norm than the
part  $ \P_{\lambda_0} (\alpha \cdot X) \P_{f_0}$, which can be explained by both $\lambda_0$
and $f_0$. For a single ($K=1$) regressor with $X_{1,it}=v_i w_t$, 
as in Example~\ref{ex:ID}, the condition simply becomes
 $\|   \M_{\lambda_0} v \| \| \M_{f_0} w \|  >   \|   \P_{\lambda_0} v \| \| \P_{f_0} w \| $.
Here, $\|   \M_{\lambda_0} v \|^2$ and $\|   \P_{\lambda_0} v \|^2$ are the 
residual sum of squares 
 and the explained sum of squares, respectively, of a regression of $v_i$ on 
 the 
 $\lambda_{0,i}$ 
 and analogously for $\| \M_{f_0} w \|^2 $ and $\| \P_{f_0} w \|^2 $.
 In Example~\ref{ex:ID} we obviously have $\|   \M_{\lambda_\bigstar} v \| =0$
 and $ \| \M_{f_\bigstar} w \| =0$, that is, the parameters
 $R_{\bigstar}$, $\beta_{\bigstar}$, $\lambda_{\bigstar}$, $f_{\bigstar}$
 are ruled out by the condition on the regressors in Proposition~\ref{prop:unique.separation}.

Related to the regularity condition (\ref{prop:unique.separation.regularity.cond}) of Proposition \ref{prop:unique.separation},  it is possible to show  (see  \citealt{Bai2009}, \citealt{MoonWeidner2017}) that the weaker condition 
$   \M_{\lambda_0} (\alpha \cdot X) \M_{f_0}  \neq 0$  for any linear combination $\alpha \neq 0$
is sufficient for local identification of $\beta$ in a sufficiently small neighbourhood around $\beta_0$.
However, that weaker condition is not sufficient for global identification of $\beta_0$,
as illustrated by the examples in the supplementary appendix S.3
of \cite{MoonWeidner2017}.  The stronger condition  (\ref{prop:unique.separation.regularity.cond}) in Proposition~\ref{prop:unique.separation}
guarantees global identification of $\beta_0$ when using the nuclear norm penalization approach
as a regularization device.

Providing such global identification conditions for models
with low-rank regressors and unknown $R_0$ is a new contribution to the interactive fixed effects literature.\footnote{
If the model would not have any idiosyncratic errors (i.e.\ $E=0$), then $Y - \beta \cdot X =  (\beta_0 - \beta) \cdot X + \Gamma_0$, and a natural solution to this identification problem would be to choose $\beta$ as the solution to the rank minimization problem 
$
\min_{\beta \in \mathbb{R}^K}  \; {\rm rank}\left(  Y - \beta \cdot X \right) ,
$
where at the true parameters we have ${\rm rank}\left(  Y - \beta_0 \cdot X \right)= {\rm rank}(\Gamma_0) = R_0$, that is,
we are minimizing the number of factors required to describe the data.
However, once idiosyncratic errors $E$ are present, then this rank minimization 
does not work, because $Y - \beta\cdot X$ is of large rank
for all $\beta$.
}
Our approach here is similar to the ``Identification via a Strict Convex Penalty''
proposed in \cite{chen2012estimation}.
As noted above, the proof of Proposition~\ref{prop:unique.separation} applies to the population objective function $\bar{\beta}_\psi$ for fixed $N$ and $T$. The consistency results for the sample estimators $\widehat{\beta}_{\psi}$ and $\widehat{\beta}_{*}$ in Section~\ref{sec:Consistency} are established using different asymptotic arguments, in particular relying on Lemma~S.3 in the supplementary appendix rather than on Proposition~\ref{prop:unique.separation} directly.

\section{Consistency of $\widehat \beta_\psi $ and $ \widehat \beta_*$} 
\label{sec:Consistency}

Proposition~\ref{prop:unique.separation} above provides an identification result for $\beta_0$ for fixed $N$ and $T$,
based on the expected objective function. 
We now turn to the actual estimates $\widehat \beta_\psi$ and $\widehat \beta_*$ 
and   investigate their  properties as $N,T \rightarrow \infty$.

All our consistency results for $\widehat \beta_\psi$
are for asymptotic sequences where
 $\psi = \psi_{NT} \rightarrow 0$, as $N,T \rightarrow \infty$,
 but we do not usually make the dependence
 of $\psi$ on the sample size explicit.
 In addition, we assume that the number of the regressors $K$
 and the true number of factors $R_0 = {\rm rank}(\Gamma_0)$ are both fixed.
However, we do not restrict whether the factors are strong or weak, nor do we restrict 
 the magnitude of $\Gamma_0$ in any matrix norm.

\subsection{Consistency Results for Low-Rank Regressors} 
\label{sec:LowRank}

Here, we consider a special case where the regressors $X_1,...,X_K$ are of low rank.
This section is short, because the results here are relatively straightforward extensions of
Section~\ref{sec:MatrixSeparation}.
The more general case that allows both high-rank and low-rank regressors will be discussed in the following subsection.

\begin{theorem}
	\label{th:LowRankConsistency}
Consider $N,T \rightarrow \infty$ with $\psi  \rightarrow 0$, and assume that
   \begin{itemize}
		\item[(i)] There exists a constant $c$ such that 
		\begin{align}
		\min_{\left\{ \alpha \in \mathbb{R}^{K} \, : \, \| \alpha\|=1 \right\}}  
		\left\| \frac { \M_{\lambda_0} (\alpha  \cdot X) \M_{f_0} } {\sqrt{NT}}\right\|_1
		-   \left\| \frac{ \P_{\lambda_0} (\alpha \cdot X) \P_{f_0}} {\sqrt{NT}} \right\|_1  
		\geq c >0 ,
		\label{LowRankNonColl}
		\end{align}
		for all sample sizes $N,T$.
		\item[(ii)]  $\| E \|_\infty = O_P(\sqrt{\max(N,T)})$,
		and $\sum_{k=1}^K {\rm rank}(X_k) = O_P(1)$.
 \end{itemize}
 Then we have
	\begin{align*}
	\left\| \widehat \beta_\psi  - \beta_0 \right\| &= O_P(\psi) + O_P\left( \frac 1 {\sqrt{\min(N,T)}} \right) ,
        &
	\left\| \widehat \beta_*  - \beta_0 \right\| &=  O_P\left( \frac 1 {\sqrt{\min(N,T)}} \right).
	\end{align*}
		
\end{theorem}

Various examples of DGP's for $E$  that satisfy the assumption $\| E \|_\infty = O_P(\sqrt{\max(N,T)})$
can be found in the supplementary appendix S.2 of \cite{MoonWeidner2017}. 
Loosely speaking, that condition is satisfied as along as the entries $E_{it}$ have zero mean, some appropriately bounded moments,
and are not too strongly correlated across $i$ and over $t$. The condition $\sum_{k=1}^K {\rm rank}(X_k) = O_P(1)$ requires all regressors to be low-rank.
The interpretation of condition \eqref{LowRankNonColl} is the same as for condition \eqref{prop:unique.separation.regularity.cond}
in Proposition~\ref{prop:unique.separation}, and  Theorem~\ref{th:LowRankConsistency} is indeed a sample version of that proposition,
 except that low-rank regressors are required here.

The theorem shows that both $\widehat \beta_*$ and $\widehat \beta_\psi $, for $\psi = \psi_{NT} =  O\left(  1 / {\sqrt{\min(N,T)}} \right)$,
converge to $\beta_0$ at a rate of at least $\sqrt{\min(N,T)}$. 
The proof of the theorem is provided in the appendix, and is a relatively easy generalization of the proof of Proposition~\ref{prop:unique.separation}.

We note that condition \eqref{LowRankNonColl} requires a lower bound of order $c > 0$ (a positive constant), whereas condition \eqref{prop:unique.separation.regularity.cond} in Proposition~\ref{prop:unique.separation} only requires a strict inequality. This difference reflects the distinction between a population identification condition (Proposition~\ref{prop:unique.separation}, for fixed $N$ and $T$) and a sample condition for asymptotic analysis (Theorem~\ref{th:LowRankConsistency}, as $N, T \rightarrow \infty$). The nuclear norms appearing in \eqref{LowRankNonColl} are of order $\sqrt{NT}$, so requiring their difference to be bounded below by a positive constant $c$ uniformly over all sample sizes ensures the condition does not become slack asymptotically.

Returning to Example~\ref{ex:ID}, the DGP with low-rank regressor $X_1 = vw'$ satisfies condition \eqref{LowRankNonColl} when the vectors $v$ and $w$ are sufficiently different from the factor loadings $\lambda_0$ and factors $f_0$. Specifically, if $v$ is not well-approximated by linear combinations of the columns of $\lambda_0$, and similarly $w$ is not well-approximated by linear combinations of the columns of $f_0$, then the projection $\mathbf{M}_{\lambda_0} X_1 \mathbf{M}_{f_0} = (\mathbf{M}_{\lambda_0} v)(\mathbf{M}_{f_0} w)'$ retains substantial norm, while $\mathbf{P}_{\lambda_0} X_1 \mathbf{P}_{f_0} = (\mathbf{P}_{\lambda_0} v)(\mathbf{P}_{f_0} w)'$ remains small. In the extreme case where $v$ lies in the column space of $\lambda_0$ or $w$ lies in the column space of $f_0$, condition \eqref{LowRankNonColl} fails and $\beta_0$ is not identified, consistent with the non-identification illustrated in Example~\ref{ex:ID}.

\subsection{Consistency Results for General Regressors} 
\label{sec:GeneralX}

The previous subsection considered the case where all regressor  matrices $X_k$ are low-rank. We now study situation
where all or some of the regressor matrices $X_k$ are high-rank. 

\subsubsection{Consistency of $\widehat{\beta}_{\psi}$ and $\widehat{\Gamma}_{\psi}$} \label{subsec:nuc.regularized}

Applying Lemma~\ref{lemma:QPsiRewrite} and the model for $Y$ we have
    \begin{align*}
         Q_\psi(\beta) &= q_\psi\left( \frac {E + \Gamma - (\beta-\beta_0) \cdot X} {\sqrt{NT}}    \right) 
         = 
         \sum_{r=1}^{\min(N,T)} q_\psi\left( s_r\left( \frac {E + \Gamma - (\beta-\beta_0) \cdot X} {\sqrt{NT}}    \right) \right) .
    \end{align*}   
The proof strategy for Theorem~\ref{th:LowRankConsistency} requires 
 that both $\Gamma$ and $X_k$ are  low-rank, which allows to (approximately) separate off $E$ in this expression for  $Q_\psi(\beta) $. But  if one of the regressors $X_k$ is a high-rank matrix that proof
strategy turns out not to work anymore, because  the singular value spectrum of the sum of two high-rank matrices $E$ and $X_k$
does not decompose (or approximately decompose) into a contribution from $E$ and from $X_k$, but instead all singular values depend on both
of those high-rank matrices in a complicated non-linear way.

We therefore now follow a different strategy, where instead of studying the objective function after profiling out $\Gamma$, 
we now explicitly study the properties of the estimator for $\Gamma$. Let 
\begin{align*}
(\widehat{\beta}_{\psi},\widehat{\Gamma}_{\psi}) &=
\bigg[  \argmin_{\beta,\Gamma} 
\underbrace{ \frac{1}{2NT} \| Y - \beta \cdot X - \Gamma \|_2^2 }_{ =: \, L(\beta,\Gamma)}  + \frac{\psi} {\sqrt{NT}} \| \Gamma \|_1 \bigg] .
\end{align*}
For the results
in this subsection we are going to first show consistency of $\widehat{\Gamma}_{\psi}$,
and afterwards use that to obtain consistency of $\widehat{\beta}_{\psi}$.
This is a very different logic than in the preceding section, where consistency of $\widehat{\Gamma}_{\psi}$ 
is usually not achieved, because we do not impose any lower bound on $\psi$.
In order to achieve consistency of $\widehat{\Gamma}_{\psi}$ one requires $\psi$ not be too small.
The approach here is much more similar to the machine learning literature
(e.g., \citealt{Negahbanetal2012}), where 
the matrix that we call $\Gamma$ is usually the object
of interest, and correspondingly a lower bound on the penalization parameter is required.
We also follow that literature here by imposing a so-called ``restricted strong convexity'' condition below,
which is critical to show consistency of $\widehat{\Gamma}_{\psi}$ and consequently of $\widehat{\beta}_{\psi}$ in the following.

It is convenient to introduce some additional notation:
Let $\rm vec(A)$ be the vector that vectorizes the columns of $A$. Denote $\rm mat(\cdot)$ as the inverse operator of $\rm vec(\cdot)$, so for $a = \rm vec(A)$ we have $\rm mat(a) = A$. 
We use small letters to denote vectorized variables and parameters. Let $y = {\rm vec}(Y), x_k ={\rm vec}(X_k), \gamma_0 = {\rm vec}(\Gamma_0),$ and $e = {\rm vec}(E)$.
Define $x = (x_1,...,x_k)$. Using this, we express the model (\ref{model:linear.matrix}) as
$y = x \beta_0 + \gamma_0 + e$, where all the summands are  $NT$-vectors,
and the least-squares objective function reads 
$L(\beta,\Gamma) =  \frac{1}{2NT}   (y - x \beta - \gamma)' (y - x \beta - \gamma)$.

\begin{assumption}[\bf Restricted Strong Convexity] \label{ass:RSC}
	$\phantom{a}$ \\
	\noindent
	Let $
	\mathbb{C}
	= \left\{
	\Theta \in \mathbb{R}^{N \times T} 
	\;|\; \| \M_{\lambda_0} \Theta \M_{f_0}  \|_1 
	\leq 3 \| \Theta - \M_{\lambda_0} \Theta \M_{f_0} \|_1
	\right\}. 
	$
	We assume that
	there exists $\mu > 0$, independent from $N$ and $T$, such that  for any
	$\theta \in \mathbb{R}^{NT}$ with
	${\rm mat}(\theta) \in \C$
	we have
	$ \theta' \M_x \theta \geq \mu \, \theta' \theta $,
	for all $N$, $T$.
\end{assumption}

  The intuitive interpretation of Assumption~\ref{ass:RSC} is very similar to condition
  \eqref{prop:unique.separation.regularity.cond} in Proposition~\ref{prop:unique.separation}:
The cone $\mathbb{C}$ contains matrices $\Theta$ that are close to $\Gamma_0 = \lambda_0 f_0'$,
in the sense that the part $\M_{\lambda_0} \Theta \M_{f_0}$ of $\Theta$ that cannot be explained by $\lambda_0$
and $f_0$ is small compared to the remaining part of $\Theta$, in terms of 
nuclear norm.
The assumption then imposes that all those matrices $\Theta \in \C$ in the cone 
are sufficiently different from the regressors, in the sense that $\theta = {\rm vec}(\Theta)$ cannot 
be perfectly explained by $ x_k ={\rm vec}(X_k)$.

Specifically, 
the condition assumes that the quadratic term, $\frac{1}{2NT} (\gamma - \gamma_0)'\M_x(\gamma - \gamma_0)$, of the profile likelihood function, $\min_{\beta} L(\beta,\Gamma)$, is bounded below by a strictly convex function, $\frac{\mu}{2NT} (\gamma - \gamma_0)'(\gamma - \gamma_0)$, if $\Theta = \Gamma - \Gamma_0$ belongs in the cone $\C$. Notice that without any restriction on the parameter 
$\theta = \gamma - \gamma_0$, we cannot find a strictly positive constant $\mu > 0$ such that $\min_{\Gamma} (\gamma - \gamma_0)'\M_x(\gamma - \gamma_0) \geq \mu (\gamma - \gamma_0)'(\gamma - \gamma_0)$. 
Assumption \ref{ass:RSC} imposes that if we restrict the parameter set to be the cone $\C$, then we can find a strictly convex lower bound of the quadratic term of the profile likelihood. Assumption \ref{ass:RSC} corresponds to the restricted strong convexity condition in \cite{Negahbanetal2012}, and it plays the same role as the restricted eigenvalue condition in recent LASSO literature (e.g., see \cite{CandesTao2007} and \cite{Bickeletal2009}). 

Notice that for $R_0=0$ we have $\M_{\lambda_0} = \mathbb{I}_N$ and $\M_{f_0} =\mathbb{I}_N $,
and therefor $\mathbb{C} = \{ 0_{N \times T} \}$, implying that Assumption~\ref{ass:RSC} is trivially satisfied
for any $\mu>0$.

Assumption~\ref{ass:RSC} requires that for ${\rm mat}(\theta) \in \C$ a lower bound of $\theta'\M_x \theta$ is given by the strictly convex function $\mu \, \theta' \theta $. To have some intuition, suppose that the regressor is scalar and assume that $\| X \|_2 = (x'x)^{1/2} = 1$ without loss of generality because the projection operator $\M_x$ is invariant to the scale change. Also assume that $\theta \neq 0$.  Then,
\begin{align*}
\theta' \M_x \theta &= \theta'\theta - (\theta'x)^2 
 = (\theta'\theta) \left( 1 - \frac{(\theta'x)^2}{\theta'\theta} \right)  = (\theta'\theta) \left(x'x - x'\theta (\theta'\theta)^{-1} \theta'x \right) \\
& \geq  (\theta'\theta) \min_{\theta \in \C} \| x - \theta \|^2.
\end{align*}
In this case, if the limit of the distance between the regressor and the restricted parameter set is positive, Assumption \ref{ass:RSC} is satisfied if $\mu := \liminf_{N,T}\min_{\theta \in \C} \| x - \theta \|^2$, the distance of the normalized regressor $x$ and convex cone $\C$, is positive. An obvious necessary condition for this is that the normalized regressor does not belong in the cone $\C$, that is, 
\[ 
\| \M_{\lambda_0} X \M_{f_0}  \|_1 
> 3 \| X - \M_{\lambda_0} X \M_{f_0} \|_1 .
\]  
For example, if $X$ has an approximate factor structure
	\[
	X =  \lambda_x f_x' + E_{x},
	\]	 
with $E_{x,it} \sim i.i.d. {\cal N}(0,\sigma^2)$,	
then 	we can use random matrix theory results to show that Assumption~\ref{ass:RSC} is satisfied.

\begin{lemma}[\bf Convergence Rate of $\widehat{\Gamma}_{\psi}$]
	\label{lemma:ConTheta}
	Let Assumption~\ref{ass:RSC} holds
	and assume that  
	\begin{align}
	\psi \geq \frac{2}{\sqrt{NT}}   \| {\rm mat}(\M_x e)  \|_\infty. \label{def.psi}
	\end{align}
	Then we have
	\[
	\frac 1 {\sqrt{NT}} \left\| \widehat{\Gamma}_{\psi} - \Gamma_0 \right\|_2 \leq \frac{3\sqrt{2R_0}}{\mu} \, \psi. 
	\]
\end{lemma}

\medskip

The lemma shows that once we impose restricted strong convexity and a lower bound on $\psi$, then we can indeed
bound the difference between $\widehat{\Gamma}_{\psi} $ and $\Gamma_0$. This lemma is obviously key to obtain
a consistency result for $\widehat{\Gamma}_{\psi} $.
Notice furthermore that
	\begin{align*}
		\widehat{\beta}_{\psi} -\beta_0 &= 
		(x'x)^{-1}x' \left( y - \widehat{\gamma}_{\psi} \right)
		= \left( x'x \right)^{-1} \left[ x'e -  x'(\widehat{\gamma}_{\psi} - \gamma_0) \right] ,
	\end{align*}
that is, once we have a consistency result for $\widehat{\Gamma}_{\psi} $ (or equivalently $\widehat{\gamma}_{\psi} $), then we can 
also show consistency of $\widehat{\beta}_{\psi}$. 
 Using that derivation strategy we obtain the following theorem,
which provides a consistency result for  both $\widehat{\Gamma}_{\psi}$ and $\widehat{\beta}_{\psi} $.

\begin{theorem}\label{theorem:first.stage.estimator}
	Let Assumption~\ref{ass:RSC} hold, and as $N, T \rightarrow \infty$ assume that
	\begin{itemize}
		\item[(i)] 
		$\| E \|_{\infty} = {\Op}\left( \max (N,T )^{1/2} \right)$,
		
		\item[(ii)]   
		$\frac 1 {\sqrt{NT}} \, e' x =  {\Op}(1)$,
		
		\item[(iii)] 
		$\frac 1 {NT} \, x' x  \, \rightarrow_p \, \Sigma_x > 0$,
		
		\item[(iv)]   
		$\psi  = \psi_{NT} \rightarrow 0$
		such that $ \sqrt{\min(N,T)} \, \psi_{NT}      \rightarrow \infty$.
	\end{itemize}		
	Then we have	
	\begin{align*}
	 \quad \frac 1 {\sqrt{NT}} \left\| \widehat{\Gamma}_{\psi} - \Gamma_0 \right\|_2 
	&\leq \Op(  \psi)   ,
	&
	 \quad	\left\| \widehat{\beta}_{\psi} - \beta_0 \right\| 
	&\leq \Op(  \psi )   .
	\end{align*}
\end{theorem}
\medskip
\noindent
The additional regularity conditions imposed in Theorem~\ref{theorem:first.stage.estimator} are weak and quite general. 
As mentioned before, various examples of $E$ that satisfy (i) can be found in
the supplementary appendix S.2 of \cite{MoonWeidner2017};
these include weakly dependent errors, and nonidentical but independent sequences of errors. Condition (ii) is satisfied if the regressors are exogenous with respect to the error, $\E(x_{it}e_{it}) = 0$, and $x_{it}e_{it}$ are weakly correlated over $t$ and across $i$ so that $\frac{1}{NT} \sum_{i,j=1}^N \sum_{t,s=1}^T \allowbreak  \E(x_{k,it}x_{l,js}e_{it}e_{js})$ is bounded asymptotically. 
Condition (iii) is the standard non-collinearity condition for the regressors. Condition (iv) restricts the choice of the regularization parameter $\psi$, which has to converge to zero (as discussed before for identification and consistency of $\beta_0$), but not too quickly 
      (if $\psi$ is too small, then $\widehat{\Gamma}_{\psi} $ picks up all the noise $E$ and cannot be consistent).
   The conditions (i) and (iv) are sufficient regularity conditions for (\ref{def.psi}). To see this in more detail, since ${\rm mat}(\M_x e) = E - \sum_{k=1}^K\widehat{E}_k$ with $\widehat{E}_k = X_k (x_k'x_k)^{-1}(x_k' e)$, we have
	\begin{align*}
	\left\| {\rm mat}(\M_x e) \right\|_{\infty} 
	&= \left\| E - \sum_{k=1}^K\widehat{E}_k \right\|_{\infty}  
	 \leq  \| E \|_{\infty} + \sum_{k=1}^K \| \widehat{E}_k \|_{\infty} \\
	 &= \| E \|_{\infty} + \sum_{k=1}^K \left\|\frac{X_k}{\sqrt{NT}} \right\|_{\infty} \left( \frac{x_k'x_k}{NT}   \right)^{-1} 
	\left| \frac{x_k'e}{\sqrt{NT}} \right| 
	 \leq \| E \|_{\infty} \left( 1 + \frac{\Op(1)}{\| E \|_{\infty}}\right).
	\end{align*}
	Then, choosing $\psi \geq  \frac{2}{\sqrt{NT}}   \| E  \|_\infty  \left( 1 + \frac{\Op(1)}{\| E \|_{\infty}}\right)$ makes $\psi$ satisfy $(\ref{def.psi})$ with probability approaching one,
	and the rate condition in   (iv) guarantees this.

Theorem~\ref{theorem:first.stage.estimator}  requires $\psi = \psi_{NT}$ to grow  faster than 
$1 / \sqrt{\min(N,T)} $. By choosing $\psi$ appropriately we can therefore obtain
a convergence rate of $\widehat{\beta}_{\psi}$ that is just below $\sqrt{\min(N,T)}$,
which is essentially the same convergence rate that we found in Section~\ref{sec:LowRank}
for the case of only low-rank regressors.

For the special case $R_0=0$
we have
$\Gamma_0 = \mathbf{0}_{N \times T}$,
and if $\psi_{NT}$ then satisfies (\ref{def.psi}), one can show that
\begin{equation}
\| \widehat{\Gamma}_{\psi} -\Gamma_0 \|_1 = 0,
\label{special.Gammahat}
\end{equation}
with probability approaching one (wpa1), see the appendix for a proof of this.
In this case, the regularized estimator of $\beta$ becomes the pooled OLS estimator,
$
\widehat{\beta}_{\psi} = (x'x)^{-1}x'y,
$
wpa1.

\subsubsection{Consistency of $\widehat{\beta}_*$}  \label{subsec:nuc.minimization}

Here, 
we establish consistency of the nuclear norm minimization estimator $\widehat{\beta}_*$
for high-rank regressors. For simplicity we only discuss the case of a single regressor ($K=1$) in the main text, and
we simply write $X$ for the $N \times T$ regressor matrix $X_1$ in this subsection.
The general case of multiple regressors ($K>1$) is discussed in Appendix~\ref{sec:proofs-of-section-refsecnucminimization}.

Remember that
$\widehat{\beta}_*$ is the minimizer of the objective function
$ \| Y - \beta \cdot X  \|_1 = \| E + (\beta_0-\beta)  X + \Gamma_0  \|_1 = 
  \sum_{r} s_r\left(  E + (\beta_0-\beta)  X + \Gamma_0  \right)$.
Asymptotically separating the contribution of the low-rank matrix  $\Gamma_0$ to the singular values
of the sum   $E + (\beta_0-\beta)  X + \Gamma_0 $ is possible under a strong factor assumption.\footnote{
In \cite{MoonWeidner2015,MoonWeidner2017} we use the perturbation theory of linear operator to do exactly that.}
However,  
characterizing the singular values of the sum of two high-rank matrices
$ E + (\beta_0-\beta)  X $ requires  results from random matrix theory that are usually only shown under relatively strong assumptions
on the distribution of the matrix entries. We therefore first provide a theorem under high-level assumptions,
and afterwards discuss how to verify those assumptions using results from random matrix theory.
We write SVD for ``singular value decomposition'' in the following.

\begin{theorem}
 	\label{th:ConsistencyNormMinEst} 
     Suppose that $K=1$, and assume that as $N,T \rightarrow \infty$, with $N>T$, we have
	\begin{itemize}
		\item[(i)] $\| E \|_{\infty} = \Op(\sqrt{N})$, and $\| X \|_{\infty} = \Op( \sqrt{NT} )$.
		\item[(ii)] There exists a finite positive constant $c_{\rm up}$ such that 
		$\frac 1 {T \, \sqrt{N} } \| E \|_1 \leq \frac 1 2  c_{\rm up}$, wpa1.
		\item[(iii)]  Let $U_E S_E V_E' $ be the SVD of $\M_{\lambda_0} E \M_{f_0}$.\footnote{That is, $U_E S_E V_E' = \M_{\lambda_0} E \M_{f_0}$ and $U_E$ is an $N \times {\rm rank}(\M_{\lambda_0} E \M_{f_0})$ matrix of singular vectors,
			$S_E$ is a ${\rm rank}(\M_{\lambda_0} E \M_{f_0}) \times {\rm rank}(\M_{\lambda_0} E \M_{f_0})$ diagonal matrix,
			and $V_E$ is an $T \times {\rm rank}(\M_{\lambda_0} E \M_{f_0})$ matrix of singular vectors.} We assume
		$
		{\rm Tr}\left(  X' U_E V_E' \right)   = \Op ( \sqrt{NT} ).
		$
		\item[(iv)] There exists a constant $c_{\rm low}>0$ such that 
		$T^{-1} N^{-1/2} \| \M_{\lambda_0} X  \M_{f_0} \|_1  \geq c_{\rm low}$, wpa1.
		\item[(v)] Let $U_x S_x V_x' = \M_{\lambda_0} X \M_{f_0} $
		be the SVD of the matrix $\M_{\lambda_0} X  \M_{f_0}$.
		We assume that there exists $ c_x  \in (0,1)$ such that 
		$
		{\rm Tr} \left(  U_E'U_x S_x U_x' U_E \right)  \leq (1-c_x) {\rm Tr}(S_x)
		$, wpa1.
	\end{itemize}
    We then have $\sqrt{T} \left(  \widehat{\beta}_{*} - \beta_0 \right) = \Op(1)$.
\end{theorem}

\begin{remark} \label{remark:identification.conditions}
Several identification-type conditions appear in this paper, and we briefly summarize their relationships here:
\begin{itemize}
\item Condition \eqref{prop:unique.separation.regularity.cond} in Proposition~\ref{prop:unique.separation} is a population identification condition for fixed $N$ and $T$. It applies to both low-rank and high-rank regressors and requires that the part of any linear combination of regressors not explained by the factor structure dominates (in nuclear norm) the part that can be explained.

\item Condition \eqref{LowRankNonColl} in Theorem~\ref{th:LowRankConsistency} is the sample analog of \eqref{prop:unique.separation.regularity.cond} for low-rank regressors. The $c\sqrt{NT}$ scaling ensures the condition holds uniformly as $N, T \rightarrow \infty$, reflecting the order of magnitude of nuclear norms in the asymptotic analysis.

\item Assumption~\ref{ass:RSC} (Restricted Strong Convexity) applies when regressors are high-rank. It requires that the quadratic form $\theta' \mathbf{M}_x \theta$ is bounded below on the cone $\mathbb{C}$, which contains matrices close to the factor structure.

\item The conditions in Theorem~\ref{th:ConsistencyNormMinEst} provide sufficient primitive conditions for consistency of $\widehat{\beta}_{*}$ when regressors are high-rank. In particular, conditions (iv) and (v) ensure that $X$ is sufficiently different from both the factor structure and the idiosyncratic errors.
\end{itemize}
All these conditions share a common intuition: the regressors must be sufficiently ``different'' from the factor structure $\Gamma_0$ (for low-rank $X$) or from the idiosyncratic errors $E$ (for high-rank $X$) to ensure that $\beta_0$ is identified.

We note that our conditions are stated in stochastic form (holding with probability approaching one or for all realizations). In related work, such as \cite{belloni2019high} and \cite{feng_2023}, analogous conditions are sometimes stated in expectation form. Relaxing our conditions to expectation form would require verifying that appropriate concentration inequalities hold for the nuclear norm quantities involved, which we have not pursued here. Our stochastic form is standard in this setting and sufficient for our proofs.
\end{remark}

The theorem considers the case $N>T$, because the two panel dimensions are not treated symmetrically
in the assumptions and proof of this theorem.
 Alternatively, we could consider  $T<N$, but then we also need to swap $N$ and $T$,
and replace $X$ by $X'$ and $E$ by $E'$ in all the assumptions (the case $T=N$ is ruled out here for technical reasons). 
For both $N>T$ and $T<N$ the statement of theorem can be written as
 $\sqrt{\min(N,T)} \left(  \widehat{\beta}_{*} - \beta_0 \right) = \Op(1)$,
 that is, we have the same convergence rate result here for $\widehat{\beta}_{*} $
 as in Theorem~\ref{th:LowRankConsistency} above.

Condition (i) in the theorem is quite weak, we already discussed the rate restriction on $\| E \|_{\infty}$ above,
and we have $\| X \|_\infty \leq \|X\|_2 = \sqrt{\sum_i \sum_t X_{it}^2} =  \Op(\sqrt{NT})$
as long as $\sup_{it} \E (X_{it}^2)$ is finite.
Condition (ii) almost follows from $\| E \|_{\infty} = \Op(\sqrt{N})$, because
we have $ \| E \|_1 \leq  {\rm rank}(E)  \, \| E \|_{\infty}  \leq T \| E \|_{\infty} =\Op(T \sqrt{N}) $,
and the assumption is only slightly stronger than this in assuming a fixed upper bound with probability approaching one,
which can also be verified for many error distributions. 
Condition (iii) is a high level condition and will be satisfied if 
	\begin{equation}
	\sup_r \E | V_{E,r}' X' U_{E,r} | \leq M \label{eq.remark.Tr(X'UEVE')} ,
	\end{equation}
for some finite constant $M$, where $U_{E,r}$ and $V_{E,r}$ are the $r^{th}$ columns of $U_{E,r}$ and $V_{E}$, respectively. An example of DGP's of $X$ and $E$ that satisfies condition (\ref{eq.remark.Tr(X'UEVE')}) is given by Assumption LL (i) and (ii) in \cite{MoonWeidner2015}. 
Condition (iv) rules out ``low-rank regressors'',
for which we typically have $ \| \M_{\lambda_0} X  \M_{f_0}    \|_1 = \Op(\sqrt{NT})$,
but is satisfied generically for ``high-rank regressors'', for which $\M_{\lambda_0} X \M_{f_0}  $
has $T$ singular values of order $\sqrt{N}$, so that $ \| \M_{\lambda_0} X  \M_{f_0}    \|_1$ is of order
$T \sqrt{N}$. 
Condition (v) requires that the singular vectors of $\M_{\lambda_0} X  \M_{f_0}$ are sufficiently
different from the singular vectors $\M_{\lambda_0} E  \M_{f_0}$. If $X$ and $E$ are independent, 
then we expect that assumption to hold quite generally. %

\section{Post Nuclear Norm Regularized Estimation} 
\label{sec:post estimation}

In Section~\ref{sec:Consistency}  we have shown
that $\widehat{\beta}_{\psi}$ and $\widehat{\beta}_{*}$ are
consistent for $\beta_0$ at a $\sqrt{\min(N,T)}$-rate, 
which is a slower convergence rate than the $\sqrt{NT}$-rate 
at which the LS estimator $\widehat{\beta}_{{\rm LS},R}$ converges to $\beta_0$
under appropriate regularity conditions.
Our Monte Carlo results in Section~\ref{sec:MC} confirm this relatively slow rate of convergence of 
$\widehat{\beta}_{\psi}$ and $\widehat{\beta}_{*}$, that is, those rates are not an artifact of our proof strategy, 
but are a genuine property of those estimators.
 In this section we investigate how to establish an estimator that is asymptotically equivalent to the LS estimator, and yet avoids minimizing any non-convex objective function.
Our suggestion is to use either $\widehat{\beta}_{\psi}$ or $\widehat{\beta}_{*}$ as a preliminary estimator and iterate estimating $\Gamma_0 = \lambda_0 f_0'$ and $\beta_0$ a finite number of times.

The conditions that are usually needed to show that the global minimizer $\widehat{\beta}_{{\rm LS},R}$ 
 of the objective function $L_R(\beta)$ is consistent for $\beta_0$ 
 (i.e.\ Assumption~A in \cite{Bai2009}, or Assumption~4  in \cite{MoonWeidner2017})
are not required here, because we have already shown consistency of $\widehat{\beta}_{\psi}$ or $\widehat{\beta}_{*}$ 
under different conditions (our discussion in Section~\ref{sec:MatrixSeparation} highlights those differences).
It is therefore convenient to introduce a local version of the LS estimator in \eqref{DefLSestimator} as
\begin{align}
  \widehat{\beta}^{\rm \, local}_{{\rm LS},R} &:= 
   \argmin_{\beta \in {\cal B}(\beta_0, r_{NT})}  
   L_R(\beta)  ,
   &
    {\cal B}(\beta_0, r_{NT})
    :=
    \left\{ \beta \in \mathbb{R}^K \, : \, \| \beta - \beta_0 \| \leq  r_{NT} \right\} ,
     \label{DefLSestimatorLOCAL} 
\end{align}
where $r_{NT}$ is a sequence of positive numbers such that $r_{NT} \rightarrow 0$
and $\sqrt{NT} \, r_{NT} \rightarrow \infty$.
Those rate conditions guarantee that $\widehat{\beta}^{\rm \, local}_{{\rm LS},R} $
is an interior point of $ {\cal B}(\beta_0, r_{NT})$, wpa1, under the assumptions of Theorem~\ref{thm:post.matrix.lasso} below.
If the global minimizer $\widehat{\beta}_{{\rm LS},R}$ is consistent,
then we expect $\widehat{\beta}_{{\rm LS},R} = \widehat{\beta}^{\rm \, local}_{{\rm LS},R} $ wpa1, but 
$\widehat{\beta}^{\rm \, local}_{{\rm LS},R} $ is consistent by definition even if $\widehat{\beta}_{{\rm LS},R}$ is not.
Our goal in the following is to obtain an estimator that is asymptotically equivalent to $\widehat{\beta}^{\rm \, local}_{{\rm LS},R} $.
 
For simplicity, we first discuss the case where the number of factors $R_0$ is known. 
For unknown $R_0$  we recommend using a consistent estimator of $R_0$ instead, 
and
we discuss  estimation of $R_0$ in Section~\ref{sec:MC} below.
Starting from our initial nuclear norm regularized or minimized estimators 
we consider the following iteration procedure to obtain improved estimates of $\beta$:
\begin{itemize}
     \setlength{\itemindent}{0.25in}
	\item[\bf Step 1:] For $s=0$ set $\widehat{\beta}^{(s)} = \widehat{\beta}_{\psi}$ (or $=\widehat{\beta}_{*}$),
	the preliminary consistent estimate for $\beta_0$.

	\item[\bf Step 2:] Estimate the factor loadings and the factors of the $s-$step residuals $Y - \widehat{\beta}^{(s)} \cdot X$ by the principle component method: 
\[
(\widehat{\lambda}^{(s+1)}, \widehat{f}^{(s+1)}) \in \argmin_{\lambda \in \R^{N \times R_0}, f \in \R^{T \times R_0}} \left\| Y - \widehat{\beta}^{(s)} \cdot X - \lambda f'\right\|_2^2.
\]
	\item[\bf Step 3:] Update the $s$-stage estimate $\widehat{\beta}^{(s)}$ by
\begin{align}
\widehat{\beta}^{(s+1)} &= \argmin_{\beta \in \mathbb{R}^K} \min_{g \in \R^{T \times R_0}, h \in \R^{N \times R_0}} 
\left\| Y - X \cdot \beta - \widehat{\lambda}^{(s+1)} \, g' + h \, \widehat{f}^{(s+1)\prime} \right\|_2^2 \nonumber  \\
& = \left( x' \left( \M_{\widehat{f}^{(s+1)}} \otimes 
\M_{\widehat{\lambda}^{(s+1)}} \right) x \right)^{-1}  x' \left( \M_{\widehat{f}^{(s+1)}} \otimes 
\M_{\widehat{\lambda}^{(s+1)}} \right) y. \label{def.betahat.second.stage}
\end{align}

\item[ \bf Step 4:] Iterate step 2 and 3 a finite number of times.
\end{itemize}

\medskip

\begin{remark} \label{remark:step3.intuition}
Step 3 minimizes over $(\beta, g, h)$ jointly rather than using the simpler regression of $Y - \widehat{\lambda}^{(s+1)} \widehat{f}^{(s+1)\prime}$ on $X$. The key advantage of this joint minimization is that it closely approximates a Newton-Raphson step for minimizing the profile objective function $L_{R_0}(\beta)$. As shown in Theorem~\ref{thm:post.matrix.lasso}, this yields convergence similar to Newton-Raphson: the distance to the local LS estimator shrinks rapidly with each iteration, allowing the procedure to be stopped after just a few steps while still achieving asymptotic equivalence to $\widehat{\beta}^{\rm \, local}_{{\rm LS},R_0}$.
\end{remark}

The following theorem shows that
if the initial estimator $\widehat{\beta}^{(0)}$ is consistent, then 
  $\widehat{\beta}^{(s)}$ gets close to $\widehat{\beta}^{\rm \, local}_{{\rm LS},R_0} $ as the number of iteration $s$ increases.
This result is very similar to the quadratic convergence result of a Newton-Raphson algorithm for minimizing a
smooth objective function, and the above iteration step is indeed very similar to performing a  Newton-Raphson step
to minimize $L_{R_0}(\beta)$. 

\begin{theorem} \label{thm:post.matrix.lasso}
     Assume that $N$ and $T$ grow to infinity at the same rate,
      and that 
   \begin{itemize}
      \item[(i)] $\plim_{N,T \rightarrow \infty}\left(\lambda_0^{\prime} \lambda_0/N\right) > 0$, and
       $\plim_{N,T \rightarrow \infty} \left( f_0^{\prime} f_0 / T \right) > 0$.
        
       \item[(ii)] $\| E \|_{\infty} = {\Op}\left( \max (N,T )^{1/2} \right)$, and $\| X_k \|_{\infty} = {\Op}\left( (NT)^{1/2} \right)$ for all $k \in \{1,\ldots,K\}$.
        
       \item[(iii)] $\plim_{N,T \rightarrow \infty} \frac 1 {NT} \, x' \left( \M_{f_0} \otimes \M_{\lambda_0} \right) x > 0$.
       
       \item[(iv)] $\frac 1 {\sqrt{NT}} \, x' \left( \M_{f_0} \otimes \M_{\lambda_0} \right) e =  {\Op}(1)$.
    \end{itemize}
      Then, if the sequence $r_{NT} >  0$ in \eqref{DefLSestimatorLOCAL}
     satisfies $r_{NT} \rightarrow 0$ and $\sqrt{NT} \, r_{NT} \rightarrow \infty$
     we have
     \begin{align*}
        \sqrt{NT} \left(  \widehat{\beta}^{\rm \, local}_{{\rm LS},R_0} -  \beta_0 \right) &=  \Op(1) .
     \end{align*}
     Assume furthermore that 
     \begin{itemize}
     \item[(iv)]  $\| \widehat{\beta}^{(0)} - \beta_0 \| = \Op(c_{NT})$, for a sequence $c_{NT} >  0$
     such that  $c_{NT}  \rightarrow 0$.
     \end{itemize}
       For $s \in \{1,2,3,\ldots\}$ we then have
     $$
         \left\| \widehat{\beta}^{(s)} -  \widehat{\beta}^{\rm \, local}_{{\rm LS},R_0} \right\| = \Op \left\{ c_{NT} \left( c_{NT} +  \frac 1 {\sqrt{\min (N,T )}} \right)^{s} \, \right\} .
     $$
\end{theorem}

\medskip
\noindent
Here, assumption (i) is a strong factor condition, and is often used in the literature on interactive fixed effects. 
The conditions in assumption (ii) of the theorem have been discussed in previous sections and are quite weak
(remember that $\| X_k \|_{\infty} \leq \| X_k \|_{2} = \sqrt{x_k' x_k}$). Assumption (iii) guarantees that
$ L_R(\beta)$ is locally convex around $\beta_0$ -- that condition can equivalently be written
as $\plim_{N,T \rightarrow \infty} \|   \M_{\lambda_0} (\alpha \cdot X) \M_{f_0} \|_2  > 0$  for any 
 $\alpha \in \mathbb{R}^K \setminus \{0\}$, which connects more closely to our discussion in Section~\ref{sec:MatrixSeparation}.
This is a non-collinearity condition on the regressors after profiling out both $\lambda_0$ and $f_0$.
Only the true values $\lambda_0$ and $f_0$ appear in that non-collinearity condition,
and it is therefore much weaker than the corresponding assumptions
required for consistency of $\widehat{\beta}_{{\rm LS},R_0}$ in \cite{Bai2009} and \cite{MoonWeidner2017}.
Our results from the previous sections show that $\| \widehat{\beta}^{(0)} - \beta_0 \| = \Op(c_{NT})$
for both $\widehat{\beta}^{(0)} = \widehat{\beta}_{\psi}$ and $\widehat{\beta}^{(0)} =\widehat{\beta}_{*}$,
under appropriate assumptions, where $c_{NT}$ is typically either $c_{NT} = 1 / \sqrt{\min(N,T)}$
or slightly slower than this, if $\psi = \psi_{NT}$ is chosen appropriately.

The following corollary is an immediate consequence of Theorem~\ref{thm:post.matrix.lasso}.

\begin{corollary}
     \label{cor:post.matrix.lasso}
     Let the assumptions of Theorem~\ref{thm:post.matrix.lasso} hold, 
      and assume that $c_{NT} = o( (NT)^{-1/6} ) $.
     For $s \in \{2,3,4,\ldots\}$ we then have
     \begin{align*}
        \sqrt{NT} \left(  \widehat{\beta}^{(s)} -  \widehat{\beta}^{\rm \, local}_{{\rm LS},R_0} \right) &= o_P(1) ,
         &
         \sqrt{NT} \left(  \widehat{\beta}^{(s)} -  \beta_0 \right) &=  \Op(1) .
     \end{align*}
 
\end{corollary}     
The first statement of the corollary 
 shows that if the initial estimators $ \widehat{\beta}_{\psi}$ and $ \widehat{\beta}_{*}$ satisfy typical convergence
rates results derived in the previous sections,   
then the iterated estimator $ \widehat{\beta}^{(s)}$ 
is asymptotically equivalent to $  \widehat{\beta}^{\rm \, local}_{{\rm LS},R_0} $ after $s=2$ iterations or more.
Remember that if $  \widehat{\beta}_{{\rm LS},R_0} $  is consistent, then we have 
$  \widehat{\beta}^{\rm \, local}_{{\rm LS},R_0}  = \widehat{\beta}_{{\rm LS},R_0}  $ wpa1, but 
by showing asymptotic equivalence with $  \widehat{\beta}^{\rm \, local}_{{\rm LS},R_0}  $ here we avoid
 imposing conditions that require consistency of  $ \widehat{\beta}_{{\rm LS},R_0}  $.

From the results in \cite{Bai2009} and \cite{MoonWeidner2017} we also know that 
$ \widehat{\beta}^{\rm \, local}_{{\rm LS},R_0} $ is asymptotically
normally distributed, but potentially with a bias in the limiting distribution.
According to the corollary the same is therefore true for  $\widehat{\beta}^{(s)}$ for $s \geq 2$.
Asymptotic bias corrections could then also be applied 
to $\widehat{\beta}^{(s)}$, $s \geq 2$, to eliminate the bias in the limiting distribution and allow for inference on $\beta_0$. See  
 \cite{Bai2009} and \cite{MoonWeidner2017} for details.

\section{Implementation and Monte Carlo Simulations}
\label{sec:MC}

To implement the nuclear norm regularized estimator we need to choose the regularization parameter $\psi$,
and for the  post estimator $\widehat{\beta}^{(s)}$ we need to determine the number of factors $R_0$.
In this section we suggest a data dependent choice of $\psi$ as well as an 
estimate of $R_0$.
We assume that an upper bound $R_{\rm max} \geq R_0$ is known. 
 
\subsubsection*{Data Dependent Choice of $\psi$.}

We suggest the following procedure to choose $\psi$.
\begin{itemize}
     \setlength{\itemindent}{0.25in}
	\item[\bf Step 1:] Calculate  the nuclear norm minimizing estimator $\widehat{\beta}_*$, and the corresponding residuals
	$$
	\widehat{E}_* = Y - \widehat{\beta}_* \cdot X.
	$$
	
	\item[\bf Step 2:] Choose $R_{\rm max}$,  calculate  $R_{\rm max}$ principal components of $\widehat{E}_*$,
	$$
	\left\{ \widehat{\lambda}_{\rm max},  \widehat{f}_{\rm max} \right\} \in \argmin_{ \lambda \in \mathbb{R}^{N \times R_{\rm max}}, \; f \in  \mathbb{R}^{T \times R_{\rm max}} } \left\| \widehat{E}_* - \lambda f' \right\|_2^2,
	$$
	and use those to eliminate all the factors in $\widehat{E}_*$. The new residuals are
	$$
	\widetilde{E}_*  = \widehat{E} - \widehat{\lambda}_{\rm max} \widehat{f}_{\rm max}^{\, \prime}.
	$$

	\item[\bf Step 3:] Choose 
	\[
	\widehat{\psi}  = \frac{2 \| \widetilde{E}_* \|_{\infty}}{\sqrt{NT}}. 
	\]
\end{itemize} 
This choice of $\widehat{\psi}$  is motivated by the condition (\ref{def.psi}) in Lemma~\ref{lemma:ConTheta},
which guarantees that $\psi$ is sufficiently large to obtain estimates $\widehat \Gamma_\psi$
that are close to $\Gamma_0$.
Notice also that the  nuclear norm minimizing estimator $\widehat{\beta}_{*}$
in step 1 does not require any regularization parameter to be specified.

\subsubsection*{Estimation of $R_0.$}
The post nuclear norm regularized estimator introduced above assumes that the number of factors $R_0$ is known. 
In practice $R_0$ needs to be estimated, for example,
by  applying a consistent estimation method for the number of the factors in a pure factor model to the residuals  $\widehat{E}_*$,
see  e.g. \cite{BaiNg2002}, \cite{Onatski2010} and \cite{AhnHorenstein2013}. 

For our Monte Carlo simulations below we use an alternative estimation method 
that thresholds the singular values of $\widehat{E}_*$ using the estimate 
$\widehat{\psi}$ introduced above. Namely, we estimate $R_0$ by 
\[
\widehat{R} = \sum_{r=1}^{\min(N,T)} \mathbbm{1} 
\left\{ s_r\left( \widehat{E}_* \right) \geq  2 \, \sqrt{NT}  \, \widehat{\psi}  \right\}.
\]
The motivation behind this estimator is that those singular values of $\widehat{E}_*$ that are significantly larger than 
$\sqrt{NT} \, \widehat{\psi} $ should correspond to factors, while singular values close to $\sqrt{NT} \, \widehat{\psi} $ and smaller should 
originate from idiosyncratic noise. The choice of the factor 2 in the formula for $\widehat{R}$ is somewhat arbitrary,
 any alternative factor larger than one would also be plausible here.

\subsubsection*{Monte Carlo Results}

We generate data from the following linear panel model regression model with two regressors (including the intercept) and two factors:
\begin{align}
Y_{it} &= \beta_{0,1} + \beta_{0,2} \, X_{it} + \sum_{r=1}^2 \lambda_{0,ir} f_{0,tr} + E_{it} ,
\nonumber \\
X_{it} &= 1+ E_{x,it} + \sum_{r=1}^2(\lambda_{0,ir} + \lambda_{x,ir})(f_{0,tr} + f_{0,t-1,r}), 
\label{MCdesign1}
\end{align} 
where $f_{0,tr} \sim i.i.d. \, {\cal N}(0,1)$; $\lambda_{0,ir},\lambda_{x,ir} \sim i.i.d. \, {\cal N}(1,1)$; $E_{x,it}, E_{it} \sim i.i.d. \, {\cal N}(0,1)$;  
all  mutually independent.
Table~\ref{table:MC1design} reports the bias and standard deviation for the various estimators
for different combinations of $N$ and $T$.

	\begin{table}[tbh]\label{table:MC1design}
	\begin{center}
		\includegraphics[width=\textwidth]{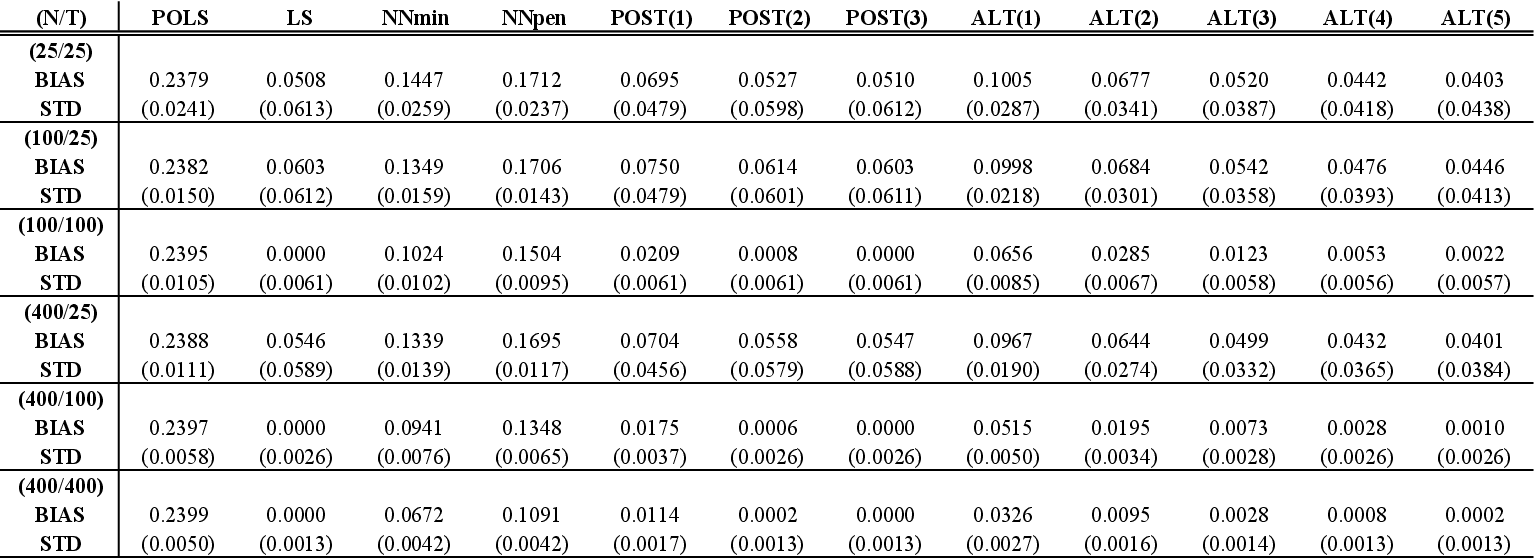}
		\caption{\label{table:MC1design} 
		Monte Carlo results based on 1000 repetitions for the design specified in display \eqref{MCdesign1}.	
		Reported are the bias and standard deviation for the pooled OLS estimator (POLS),
		the least squares estimator with $R_0=2$ factors (LS),
		the nuclear norm minimizing estimator $\widehat \beta_*$ (NNmin),
		 the nuclear norm penalized estimator with $\psi = \widehat \psi$ (NNpen),
		 the post estimator $\widehat \beta^{(s)}$ for $s=1,2,3$ iterations and using $R=\widehat R$ factors (POST$(s)$),
		 and the alternative bias correction method (see the appendix) using $R=\widehat R$ factors and $s=1,2,3,4,5$ iterations.
		}
	\end{center}
	\end{table}

As shown in Table \ref{table:MC1design}, the nuclear norm regularized estimator 
$\widehat{\beta}_{\psi}$ and the nuclear norm minimization estimator 
$\widehat{\beta}_*$ have biases due to the regularization which vanish slowly 
as the sample size increases. This confirms that those estimators
are indeed not $\sqrt{NT}$  consistent, but only have a $\sqrt{\min(N,T)}$ convergence rate to $\beta_0$. The table also shows that
the post nuclear norm regularized estimation ($\widehat{\beta}^{(s)}$) quickly reduces the bias, and essentially agrees with the 
LS estimator (which is a consistent estimator in this MC design) after two iterations, as the theory predicts.
The columns  ALT(1) - ALT(5) in that table contain  the  results for an alternative bias corrected estimator that is presented in the appendix. It turns out that the alternative bias correction method is less effective in reducing the bias, and we therefore do not discuss it in the main text.
Our recommendation in practice for inference on $\beta_0$ is the iteration procedure for $\widehat{\beta}^{(s)}$ explained in the previous section.

\section{Conclusions}
\label{sec:conc}

In this paper we analyze two new estimation methods for interactive fixed effect panel regressions
that are based on convex objective functions: (i) nuclear norm penalized estimation, and (ii) nuclear norm minimizing estimation. 
The resulting estimators can also be applied in situations 
 where the LS estimator may not be consistent, in particular when low-rank regressors are present and the true number of 
 factors is unknown.
We provide consistency and convergence  rate results for the new estimators of the regression coefficients, and we show how to use them as a preliminary estimator to achieve asymptotic equivalence to the local version of the LS estimator. 
We have focused on the linear model with homogenous coefficients, which is a natural starting point to understand the usefulness 
of nuclear norm penalization approach for panel regression models, but there are several
ongoing extensions, including developing a unified method to deal with non-linear models, heterogeneous coefficients, treatment effect estimation,
nonparametric sieve estimation, and high-dimensional regressors   --- see \cite{Atheyetal2017} and \cite{chernozhukov2018inference}.
In Appendix~\ref{sec:Nonlinear:appendix} we briefly discuss how the nuclear norm regularization approach extends to nonlinear panel models, including maximum likelihood, weighted least squares, and quantile regression settings. We note that sharper convergence rates than those presented in the appendix have been established for specific nonlinear models in the subsequent literature, including \cite{belloni2019high}, \cite{wang2022low}, and \cite{feng_2023} for quantile regression.

\ifx\undefined\BySame
\newcommand{\BySame}{\leavevmode\rule[.5ex]{3em}{.5pt}\ }
\fi
\ifx\undefined\textsc
\newcommand{\textsc}[1]{{\sc #1}}
\newcommand{\emph}[1]{{\em #1\/}}
\let\tmpsmall\small
\renewcommand{\small}{\tmpsmall\sc}
\fi

\newpage

\appendix 
\section{Appendix}

\renewcommand{\theequation}{A.\arabic{equation}}  \setcounter{equation}{0}
\renewcommand{\thetheorem}{A.\arabic{theorem}}  \setcounter{theorem}{0}
\renewcommand{\thelemma}{A.\arabic{lemma}}  \setcounter{lemma}{0}
\renewcommand{\thecorollary}{A.\arabic{corollary}}  \setcounter{corollary}{0}
\renewcommand{\theexample}{A.\arabic{example}}  \setcounter{example}{0}
\renewcommand{\thefigure}{A.\arabic{figure}}  \setcounter{figure}{0}

\subsection{An Example of a Non-convex LS Profile Objective Function}\label{subsec:ap.ex.nonconvex}

As an example for a non-convex LS profile objective function we consider the following linear model with one regressor and two factors:
\begin{align*}
Y_{it} &= \beta_0 \, X_{it} + \sum_{r=1}^2 \lambda_{0,ir} f_{0,tr} + E_{it} , \\
X_{it} &=  0.04 E_{x,it} + \lambda_{0,i1}f_{0,t2} + \lambda_{x,i} f_{x,t},
\end{align*}
where 
$$
\textstyle
 \lambda_{0,i} = {\lambda_{0,i1} \choose \lambda_{0,i2}}
\sim iid N \left( 
\left( 
\begin{array}{c}
0 \\ 
0
\end{array} 
\right), 	
\left( 
\begin{array}{cc}
1	& 0.5  \\ 
0.5	&  1 
\end{array}
\right) 
\right) , 
\;
 f_{0,t} = {f_{0,t1} \choose f_{0,t2}} 
\sim iid N 
\left( 
\left( 
\begin{array}{c}
0 \\ 
0
\end{array} 
\right), 	
\left( 
\begin{array}{cc}
1	& 0.5  \\ 
0.5	&  1 
\end{array}
\right) 
\right),$$ 
and
$\lambda_{x,i} \sim iid \,\, 2 \chi^2(1)$, $f_{x,t} \sim iid \,\, 2 \chi^2(1)$, 
$E_{x,it}, E_{it} \sim i.i.d. \, {\cal N}(0,1)$, 
and 
$ \{ \lambda_{0,i} \}$, $\{ f_{0,t} \}$, $\{ \lambda_{x,i} \}$, $\{ f_{x,t} \}$, $\{ E_{x,it} \}, \{  E_{it} \}$ are all independent of each other. For $(N,T) = (200,200)$, we generate the panel data for $(Y_{it},X_{it})$, and plot the LS objective function \eqref{DefLSestimator} in
Figure~\ref{figure:LS_Nuc_together}, which is discussed in the main text.

\subsection{Alternative Bias Correction}

In this section, we discuss an alternative bias reduction method used in the Monte Carlo simulations in Section \ref{sec:MC}. The alternative method reduces the bias of the score function of the regularized least squares objective function $Q_{\psi}(\beta)$. We introduce the procedure in a heuristic way {\it without} presenting a rigorous proof.   
We have implemented this alternative method in our Monte Carlo simulations, and while it indeed improves the  nuclear-norm penalized estimates 
(see Table~\ref{table:MC1design}), it does not perform better than the iteration method described in Section~\ref{sec:post estimation}.

Recall that $L_R(\beta,\Gamma) = \frac{1}{2NT} \left\| Y -\beta \cdot X- \Gamma \right\|_2^2$, where $\Gamma = \lambda f'$. 
Define \[
\widehat{\Gamma}_{R}(\beta) :=  \argmin_{\Gamma: {\rm rank}(\Gamma) \leq R } L_R(\beta,\Gamma).
\]
We can write
\begin{align*}
L_R(\beta) &= L_R(\beta,\widehat{\Gamma}_{R}(\beta)) = \frac{1}{2} \sum_{r=R+1}^{\min(N,T)} s_r\left( \frac{ Y - \beta \cdot X}{\sqrt{NT}} \right)^2. 
\end{align*}
Let $\widehat{\Gamma}_{\psi}(\beta) =  \argmin_{\Gamma} Q_{\psi}(\beta,\Gamma)$, and 
\begin{align*}
\widehat{R}(\beta,\psi) 
&:= \sum_{r=1}^{\min(N,T)} \mathbb{I} \{ s_r( Y - \beta \cdot X) \geq \sqrt{NT} \psi \} 
= {\rm rank} \left( \widehat{\Gamma}_{\psi}(\beta) \right).
\end{align*} 
Suppose that we choose $\psi$ such that 
\begin{equation}
	\widehat{R}(\beta,\psi)  = R_0 \label{Rhat=R0}
\end{equation}
Then, in view of (\ref{DefSmallLQ}) and (\ref{RelationRpsi}), we write the profile objective function of the regularized least squares as
\begin{align}
Q_{\psi}(\beta) 
&:= q_{\psi} \left( \frac{ Y - \beta X}{\sqrt{NT}} \right) \nonumber \\
&= \frac{1}{2} \sum_{r = R_0 + 1}^{\min(N,T)} s_r\left( \frac{ Y - \beta X}{\sqrt{NT}} \right)^2
+ \psi \sum_{r=1}^{R(\beta,\psi)} s_r\left( \frac{ Y - \beta X}{\sqrt{NT}} \right) 
- \frac{1}{2} \psi^2 R(\beta,\psi) \nonumber \\
&= L_{ R_0} (\beta) 
+ \psi \sum_{r=1}^{ R_0 } s_r\left( \frac{ Y - \beta X}{\sqrt{NT}} \right) 
- \frac{1}{2} \psi^2 R_0 \nonumber \\
&= L_{ R_0} (\beta) 
+ \psi  \left\| \widehat{\Gamma}_{R_0}(\beta) \right\|_1 
- \frac{1}{2} \psi^2 R_0. \label{expansion.qpsi}
\end{align}
This shows that the term $\left\| \widehat{\Gamma}_{R_0}(\beta) \right\|_1$ is the main source of the regularization bias. 
We suggest to  approximate $\left\| \widehat{\Gamma}_{R_0}(\beta) \right\|_1 $ as follows,
\begin{align}
    \left\| \Gamma_0 - (\beta - \beta_0) \cdot X + E  \right\|_1  
  \approx \| \Gamma_0 \|_1  - (\beta - \beta_0)' B_{NT},  \label{expansion.Gammahat.nuclearnorm}
\end{align}
where $B_{NT} = (B_{NT,1},...,B_{NT,K})'$, with 
\begin{equation*}
B_{NT,k} := \frac{1}{\sqrt{NT}} {\rm Tr}\left[  (\lambda_0'\lambda_0)^{-1/2}\lambda_0' X_k f_0(f_0' f_0)^{-1/2} \right] .
\label{Bias.Nuc.Reg.Obj}
\end{equation*}
From (\ref{expansion.qpsi}) and (\ref{expansion.Gammahat.nuclearnorm}) we expect that
$Q_{\psi}(\beta) +  \psi (\beta - \beta_0) B_{NT} $
should be a good approximation to $ L_{R_0}(\beta)$. 
This heuristic suggests that we may reduce the bias of the nuclear norm 
regularized estimation by modifying the objective function.

For this, suppose that $\widehat{\psi}$ is a data dependent choice of $\psi$ that satisfies the condition (\ref{Rhat=R0}). Let $\widehat{R}$ be a consistent estimator of $R_0$. Let $\widehat{\beta}^{(0)}_{\rm alt}$ be an preliminary estimator. For example, $\widehat{\beta}^{(0)} = \widehat{\beta}_{\widehat{\psi}}$ or $\widehat{\beta}^{(0)} = \widehat{\beta}_{*}$.

For $s=0,1,2,...$, define 
\[
(\widehat{\lambda}^{(s)}, \widehat{f}^{(s)} ) \in \argmin_{\lambda \in \mathbb{R}^{N \times \widehat{R}}, f \in \mathbb{R}^{T \times \widehat{R}}} \left\| Y - \widehat{\beta}^{(s)}_{\rm alt} \cdot X - \lambda f'\right\|_2^2 ,
\]
and 
\begin{equation*}
\widehat{B}_{NT,k}^{(s)} 
:= \frac{1}{\sqrt{NT}} {\rm Tr}\left[  (\widehat{\lambda}^{(s) \prime} \widehat{\lambda}^{(s)})^{-1/2}\widehat{\lambda}^{(s) \prime} X_k \widehat{f}^{(s)}( \widehat{f}^{(s) \prime} \widehat{f}^{(s)})^{-1/2} \right].
\end{equation*}
We modify the nuclear norm regularized objective function as
\begin{equation*}
Q_{\widehat{\psi}}^{\rm bc, s+1 }(\beta) := Q_{\widehat{\psi}}(\beta) + \widehat{\psi} (\beta - \widehat{\beta}^{(s)}_{\rm alt}) \widehat{B}_{NT}^{(s)}
\end{equation*} 
and update the estimator as
\[
\widehat{\beta}^{(s+1)}_{\rm alt} := \argmin_{\beta} Q_{\widehat{\psi}}^{\rm bc, s+1}(\beta).
\]

\section{Extension to Nonlinear Panel Data Models}
\label{sec:Nonlinear:appendix}

  This appendix extends the nuclear norm regularization approach to nonlinear panel data models. The results here are less sharp than those in the main text for linear models; see Remark~\ref{remark:nonlinear:caveat} below and the references cited therein for sharper results in specific nonlinear settings.

We now consider the following generalization of the penalized LS estimator,
\begin{align*}
\left( \widehat{\beta}_{\psi}, \widehat \Gamma_\psi \right) &\in 
\argmin_{\beta \in \mathbb{R}^K, \, \Gamma \in \mathbb{R}^{N \times T}}  
Q_\psi(\beta,\Gamma)  ,
&
Q_\psi(\beta,\Gamma) 
&:= 
\frac{1}{NT} \sum_{i=1}^N \sum_{t=1}^T \, m_{it} \left( X_{it}' \beta + \Gamma_{it}  \right)
+      \frac{ \psi } {\sqrt{NT}} \left\| \Gamma \right\|_1  ,
\end{align*}
where $m_{it}(z) :=m(W_{it},z)$ is a known convex function of the single index $z \in \mathbb{R}$, which also depends on the
observed variables $W_{it}$.
The single index $X_{it}' \beta + \Gamma_{it} $ has the same structure as the conditional mean of the linear model \eqref{model:linear.matrix},
and for $W_{it} = Y_{it}$ and  $m_{it}( z ) =  \frac 1 2 (Y_{it} - z )^2$ we obtain the penalized LS estimator that
was studied in previous sections. The nuclear norm penalty term is unchanged.

Let $\overline m_{it}(z) = \mathbb{E}( m_{it}(z)  | X)$ be the expected objective function,
conditional on $X=\{ X_{it} : i=1,\ldots,N; t=1,\ldots,T\}$,\footnote{%
	Remember that we consider $\Gamma_0$ as non-random, that is, all expectations are implicitly conditional on $\Gamma_0$ as well.
	Also, we condition on all the observed $X$ here, implying that we only consider strictly exogenous regressors in this section, but in principle
	the results could be extended to dynamic models.
}
and denote derivatives of $m_{it}(z) $ and $\overline  m_{it}(z) $ with respect to $z$ by $\partial_z m_{it}(z)$, 
$\partial_z \overline m_{it}(z)$,
$\partial_{z^2} \overline m_{it}(z)$, etc.
Let $z^0_{it} = X_{it}' \beta_0 + \Gamma_{0,it} $ be the index evaluated at the true parameters.
Let ${\cal W}$ denote the domain of $W_{it}$.
We make the following assumptions on the objective function.

\begin{assumption} \label{ass:NonlinearObj}
	Let ${\cal Z} \subset \mathbb{R}$  
	be such that $\cup_{i,t} [z^0_{it} - \epsilon, z^0_{it} + \epsilon] \subset {\cal Z} $,
	for some  $\epsilon>0$.
	Assume:
	\begin{itemize}
		\item[(i)] $W_{it}$ is independently distributed across $i$ and over $t$, conditional on $X$.
		
		\item[(ii)]  
		The objective function $m(w,z) $  is convex in $z$,
		and once continuously differentiable in $z$ almost everywhere in ${\cal W} \times {\cal Z}$.
		For any function $z_{it} = z_{it}(X) \in {\cal Z}$ the first derivative 
		$\partial_z  m_{it}(z_{it})$ exists almost surely,
		and satisfies $\max_{i,t,N,T} \mathbb{E}\left\{ \left. \left[ \partial_z  m_{it}(z_{it}) \right]^4 \, \right| \, X  \right\} < \infty$.
		
		\item[(iii)] 
		$\overline m_{it}(z)$ is four times continuously differentiable in ${\cal Z}$, with derivatives bounded uniformly 
		over $i,t,N,T$, ${\cal Z}$. There exists  $b>0$ such that 
		$\min_{i,t,N,T} \min_{z \in {\cal Z}} \partial_{z^2} \overline  m_{it}(z)  \geq b$.
		
		\item[(iv)] $\partial_z \overline m_{it}(z_{it}^0) = 0 $, for all $i,t$.

	\end{itemize}      
\end{assumption}
Here, the last assumption crucially connects the distribution of $W_{it}$ conditional on $X_{it}$ 
with the chosen objective function $m_{it}(z)$. For the LS case we have  $\partial_z  m_{it}(z_{it}^0) = E_{it}$,
and Assumption~\ref{ass:NonlinearObj} then becomes the familiar mean independence condition $\mathbb{E}(E_{it} | X) = 0$. This condition excludes a predetermined regressor.
Some further examples for data generating processes and corresponding objective functions are

\begin{itemize}
	\item[(a)] Maximum likelihood: Let $Y_{it}$ conditional on $X$ have probability mass or density function  $p(y | z^0_{it})$,
	set $W_{it} = Y_{it}$ and $m_{it}( z ) = -  \log p(Y_{it} | z)$, and assume that  $m_{it}( z )$ is strictly convex in $z$
	and three times continuously differentiable.    
	A concrete example is a binary choice probit model, where $p(y | z) = \mathbbm{1}(y=1) \Phi(z) + \mathbbm{1}(y=0) [1-\Phi(z)]$,
	and $\Phi(.)$ is the cdf of ${\cal N}(0,1)$.

	\item[(b)] Weighted Least Squares:
	Let outcomes $Y_{it}$ be generated from the linear model \eqref{model:linear.matrix} with 
	$\mathbb{E}(E_{it} | X_{it}, S_{it}) = 0$,
	and
	let $m_{it}( z ) = \frac 1 2  S_{it} (Y_{it} - z )^2$,
	and $W_{it} = (Y_{it},S_{it})$. Here, the $S_{it} \geq 0$ 
	are observed weights for each observation.
	A special case is $S_{it} \in \{0,1\}$, where $S_{it}$ is an indicator of a missing outcome $Y_{it}$.
	
	\item[(c)] Quantile Regression: 
	Let outcomes $Y_{it}$ be generated from the linear model \eqref{model:linear.matrix},
	but instead of the mean restriction for $E_{it}$ we impose the quantile restriction
	$\mathbb{E}[ \mathbbm{1}(E_{it} \leq 0) | X_{it} ] = \tau$,
	and we let $ m_{it}( z ) = \rho_{\tau}( Y_{it} - z )$,
	and   $W_{it} = Y_{it}$,
	where $\rho_{\tau}(u) = u \cdot [\tau-\mathbbm{1}(u<0)]$
	is the quantile regression objective function, and $\tau \in (0,1)$ is a chosen quantile of interest.
	See also  \cite{belloni2019high}, \cite{wang2022low}, and \cite{feng_2023}  for this case.

\end{itemize}
Some additional regularity conditions are needed to guarantee that those examples satisfy Assumption~\ref{ass:NonlinearObj}.
For many models (e.g.\ quantile regressions and binary choice likelihood)
we have $\lim_{z \rightarrow \pm \infty} \partial_{z^2} \overline  m_{it}(z) =0 $.
Then, the lower bound on $\partial_{z^2} \overline  m_{it}(z)  $ in Assumption~\ref{ass:NonlinearObj}(iii) will require us to
impose that ${\cal Z}$ is a bounded set, which can be guaranteed by assuming that $X_{it}$ and $\Gamma_{0,it}$ are uniformly bounded.
Apart from that it is straightforward to verify Assumption~\ref{ass:NonlinearObj} under standard regularity conditions for the respective model.
Notice also that Assumption~\ref{ass:NonlinearObj}(ii) is formulated with the quantile regression case in mind,
where  $\partial_z  m_{it}(z_{it}) = \tau-\mathbbm{1}(Y_{it} - z <0)$ is not well-defined at $z_{it}= Y_{it}$, but 
that is a probability zero event for continuously distributed $Y_{it}$.

In the following theorem we show that $\widehat{\beta} - \beta_0 = 
\Op(\psi^{1/2})$ for $\psi \rightarrow 0$ with $\psi \sqrt{NT} \rightarrow 
\infty$ and the regressors have a generalized factor structure. We present the 
special case where there exists a single regressor (i.e., $K=1$) and the 
regressor is strictly exogenous for technical simplicity.

\begin{theorem}
	\label{th:Nonlinear}
	Let Assumption~\ref{ass:NonlinearObj} be satisfied.
	Let  $N,T \rightarrow \infty$, $\psi \rightarrow 0$, and  $\sqrt{NT} \psi \rightarrow \infty$. Let $K=1$.
	Assume that  
	\begin{itemize}
		\item[(i)] $\| \Gamma_0 \|_1 = O( \sqrt{NT} ) $. 
		\item[(ii)]The regressor can be decomposed as $X = 	X^{(1)} + X^{(2)}$
		such that $\| X^{(1)} \|_1 = o_P(\sqrt{NT} \, \psi^{-1/2}) $, 
		and $\|  X^{(2)} \|_\infty = o_P(\sqrt{NT} \, \psi^{1/2}) $.
		\item[(iii)] $W := \frac 1 {NT} \sum_{i=1}^N \sum_{t=1}^T  (X_{it}^{(2)})^2$ satisfies $W \rightarrow_P W_\infty > 0$.
	\end{itemize}
	Then we have
	$\widehat \beta_\psi - \beta_0 = \Op(\psi^{1/2})$. 	
\end{theorem}

\begin{remark} \label{remark:nonlinear:caveat}
The convergence rate $O_P(\psi^{1/2})$ in Theorem~\ref{th:Nonlinear} is slower than the rate $O_P(\psi)$ obtained for the linear model in Section~\ref{sec:Consistency}. This slower rate is likely an artifact of our proof strategy. For specific nonlinear models, sharper results have been established in the subsequent literature. For example, \cite{belloni2019high} and \cite{feng_2023} show that nuclear norm penalized estimators for quantile regression can achieve the same $O_P(\psi)$ rate as in the linear case.
\end{remark}

\medskip
\noindent
Condition (i) of the theorem is a restriction on the growth rate of the nuclear norm of $\Gamma_0$,
which was not required for the results in Section \ref{sec:Consistency}, where we assumed only that $R_0 = {\rm rank}(\Gamma_0)$ is fixed. 
However, this condition (i) imposes only an upper bound on the growth of 
$\Gamma_0$; 
it allows that $\Gamma_0$  contains both strong factors and weak factors.\footnote{%
For a discussion of weak factors we refer to \cite{Onatski2012}.}

Condition (ii) is satisfied if the regressor has a generalized factor structure, 
$$X = \underbrace{ \lambda_x f_x^{\prime}}_{ X^{(1)}} + \underbrace{E_x}_{  X^{(2)}},$$
 where $\| \lambda_x f_x^{\prime} \|_1 = \Op(\sqrt{NT})$ and $\| E_x \|_{\infty} = \Op(\sqrt{\max(N,T)}$, and we have $\psi \rightarrow 0$ with $\min(N,T) \psi \rightarrow \infty$. 

The proof of  Theorem \ref{th:Nonlinear} is presented in the appendix, where we also discuss how the result could in principle be extended to $K > 1$ regressors,
which requires some additional technical restrictions. 
Notice also that the convergence rate of $\psi^{1/2}$ in Theorem~\ref{th:Nonlinear} is different from the convergence rate $\psi$
obtained in Section~\ref{sec:Consistency}, but this is likely an artifact of our proof strategy for Theorem~\ref{th:Nonlinear}.
Finally, the analog of the nuclear-norm minimizing estimator $\widehat \beta_*$ to non-linear models is given by
$ \lim_{\psi \rightarrow 0}  \widehat{\beta}_{\psi}$ (limit for fixed $N,T$), but we do not provide results for that limiting estimator here.
The goal of this section was not to fully discuss the non-linear case, but to 
highlight the potential of the nuclear norm penalization approach beyond the linear model that is main focus of this paper.

\newpage

\section{Supplementary Appendix}

\renewcommand{\theequation}{S.\arabic{equation}}  \setcounter{equation}{0}
\renewcommand{\thetheorem}{S.\arabic{theorem}}  \setcounter{theorem}{0}
\renewcommand{\thelemma}{S.\arabic{lemma}}  \setcounter{lemma}{0}
\renewcommand{\thecorollary}{S.\arabic{corollary}}  \setcounter{corollary}{0}
\renewcommand{\theexample}{S.\arabic{example}}  \setcounter{example}{0}
\renewcommand{\thefigure}{S.\arabic{figure}}  \setcounter{figure}{0}

\subsection{Proofs for Section~\ref{sec:ConvexProblem}}

For matrix $A$, let the singular value decomposition of $A$ be given by $A=U_A S_A V_A'$,
where $S_A = {\rm diag}(s_1,\ldots,s_{q})$, with $q={\rm rank}(A)$.
\begin{lemma}
	\label{lemma:MinGamma}
	For any $\psi>0$ we have
	\begin{align*}
	  \min_{\Gamma} \left( \frac{1}{2} \left\| A - \Gamma \right\|_2^2 +  \psi \| \Gamma \|_1 \right)
	&= q_\psi(A), \\
	 \argmin_{\Gamma} \left( \frac{1}{2} \left\| A - \Gamma \right\|_2^2 + \psi \| \Gamma \|_1 \right) &=  U_A {\rm diag}( (s_1 - \psi)_{+},\ldots,(s_q - \psi)_{+}) V_A',
	\end{align*}
	where the minimization is over all matrices $\Gamma$ of the same size as $A$ and $(s)_{+} = \max(0,s)$.
\end{lemma}

\begin{proof}[\bf Proof of Lemma~\ref{lemma:MinGamma}]
	The dependence of the various quantities on $\psi$ is not made explicit in this proof. 
	Let ${\cal Q}(A) =   \min_{\Gamma} \left( \frac 1 2 \left\| A - \Gamma \right\|_2^2 + \psi \| \Gamma \|_1 \right)$.	
	A possible value for $\Gamma$ is $\Gamma^* = U_A S^* V_A'$,
	where $S^* = {\rm diag}(s^*_1,\ldots,s^*_{q})$ and $s^*_r = \max(0, s_r - \psi)$,
	and therefore we have
	\begin{align*}
	{\cal Q}(A) &\leq  \frac 1 2 \left\| A - \Gamma^* \right\|_2^2 + \psi \| \Gamma^* \|_1 
	= \frac 1 2 \left\| S_A - S^* \right\|_2^2 + \psi \| S^* \|_1
	\\
	&= \sum_{r=1}^q  \left[ \frac 1 2 \left(s_r - s_r^* \right)^2 + \psi   s_r^* \right]
	=  \sum_{r=1}^q q_\psi(s_r)  = q_\psi(A)  .
	\end{align*}
	The nuclear norm satisfies $\| \Gamma \|_1 = \max_{\| B \|_\infty \leq 1} {\rm Tr}( \Gamma' B)$.
	A possible value for $B$ is $B^* = U_A D^* V_A'$,
	where $D^* = {\rm diag}(d^*_1,\ldots,d^*_{q})$ and $d^*_r = \min(1, \psi^{-1} s_r)$,
	which indeed satisfies $\| B^* \|_\infty = \| D^*  \|_\infty = \max_r |d^*_r| \leq 1$,    
	and therefore we have
	\begin{align*}
	{\cal Q}(A) 
	&\geq \min_{\Gamma} \left[ \frac 1 2 \left\| A - \Gamma \right\|_2^2 + \psi \, {\rm Tr}(  \Gamma' B^*) \right]
	=  \frac 1 2 \left\| A - (A -  \psi  B^*) \right\|_2^2 + \psi \, {\rm Tr}[  (A -  \psi  B^*)' B^*] 
	\\
	&=  \psi \, {\rm Tr}( A' B^*) 
	-  \frac { \psi^2 } 2 \left\|   B^* \right\|_2^2
	=  \psi \, {\rm Tr}( S_A' D^*) 
	-  \frac { \psi^2 } 2 \left\|   D^* \right\|_2^2  
	\\
	&=  
	\sum_{r=1}^q  \left[ \psi \, s_r  d_r^*  -  \frac { \psi^2 } 2  \, (d_r^*)^2  \right]  
	=  \sum_{r=1}^q q_\psi(s_r)  = q_\psi(A)  ,
	\end{align*} 
	where in the second step we found and plugged in the minimizing 
	$\Gamma =  A -  \psi  B^*$.
	By combining the above upper and lower bound on ${\cal Q}(A)$ we obtain $  {\cal Q}(A)  = q_\psi(A)$,
	which is the first statement of the lemma.		
	Since $ \argmin_{\Gamma} \left( \frac 1 2 \left\| A - \Gamma \right\|_2^2 + \psi \| \Gamma \|_1 \right)$ is unique, we deduce that $\Gamma^* = U_A S^* V_A'$ is the minimizing value, which is the second statement in the lemma.
\end{proof}

\begin{proof}[\bf Proof of Lemma~\ref{lemma:QPsiRewrite}]
	The lemma follows from the first statement of Lemma~\ref{lemma:MinGamma} by replacing $A$ and $\Gamma$ in Lemma \ref{lemma:MinGamma} with $\frac{Y - \beta \cdot X}{\sqrt{NT}}$ and $\frac{1}{\sqrt{NT}} \Gamma $, respectively.  
\end{proof}

\subsection{Proofs for Section~\ref{sec:MatrixSeparation}}

The function $ q_\psi(s)$ that appears in Lemma~\ref{lemma:MinGamma} was defined in \eqref{DefSmallLQ}.
We now define a similar function $g_\psi : [0,\infty) \rightarrow [0,\infty)$
by $g_\psi(s) = \psi^{-1} q_\psi(s)$ for $\psi>0$, and $g_\psi(s)=s$ for $\psi=0$,
that is, we have
\begin{align}
\label{DefGscalar}
g_\psi(s) := \left\{ 
\begin{array}{ll}
\frac 1 {2 \psi} \; s^2 , & \text{for $s <  \psi$,}
\\
s   - \frac {\psi} 2  ,  & \text{for $s \geq \psi$,}
\end{array}
\right.    
\end{align} 
and for matrices $A$ we define $g_\psi(A)  := \sum_{r=1}^{{\rm rank}(A)} g_\psi(s_r(A))$.
Using Lemma~\ref{lemma:MinGamma} and the definition of the nuclear norm we can write
\begin{align}
\label{DefGmatrix}
g_\psi(A)
= \left\{ 
\begin{array}{ll}
\min_{\Gamma} \left(    \frac {1} {2 \, \psi} \left\| A - \Gamma \right\|_2^2 + \| \Gamma \|_1   \right) , & \text{for $ \psi>0$,}
\\
\| A \|_1  ,  & \text{for $\psi=0$.}
\end{array}
\right.    
\end{align}
As already discussed in the main text, it is natural to rescale the profiled nuclear norm penalized objective function by $\psi^{-1}$,
because it then has a non-trivial limit as $\psi \rightarrow 0$. Using $g_\psi$
instead of $q_\psi$ therefore helps to clarify the scaling with $\psi$ in 
various expressions.
The following lemma summarizes some properties of the function $g_\psi(A)$,
which are useful for the subsequent proofs.

\begin{lemma}
	\label{lemma:PropsG}
	Let $A$ and $B$ be $N \times T$ matrices, $\lambda$ be an $N \times R_1$ matrix,
	and $f$ be a $T \times R_2$ matrix. We then have
	\begin{itemize}
		\item[(i)]  $g_\psi(A) \geq  \| A \|_1  - \frac \psi 2 {\rm rank}(A) $.
		
		\item[(ii)]  $g_\psi(A+B) \leq g_\psi(A) + \| B \|_1$, 
		\; and \; $g_\psi(A+B) \geq g_\psi(A) - \| B \|_1$.
		
		\item[(iii)]  $g_\psi(A) \geq g_\psi(\M_\lambda \, A \,  \M_f) + g_\psi(\P_\lambda \, A \, \P_f)$.
	\end{itemize}
\end{lemma}

\begin{proof}[\bf Proof of Lemma~\ref{lemma:PropsG}]
	\# \underline{Part (i):}
	From the definition of $g_\psi(s)$ in \eqref{DefGscalar} one finds
	$g_\psi(s) \geq s - \frac{\psi} 2$ for all $s \geq 0$. We thus obtain
	\begin{align*}
	g_\psi(A)  = \sum_{r=1}^{{\rm rank}(A)} g_\psi(s_r(A))
	\geq  \sum_{r=1}^{{\rm rank}(A)} \left[ s_r(A) - \frac \psi 2  \right]
	=  \| A \|_1  - \frac \psi 2 {\rm rank}(A)  .
	\end{align*}
	
	\# \underline{Part (ii):}     
	For $\psi=0$ this is just the triangle inequality for the nuclear norm.
	For $\psi>0$ we use \eqref{DefGmatrix} to write
	\begin{align*}
	g_\psi(A+B)
	&= \min_{\Gamma} \left(    \frac {1} {2 \, \psi} \left\| A +B - \Gamma \right\|_2^2 + \| \Gamma \|_1   \right)
	= \min_{\Gamma} \left(    \frac {1} {2 \, \psi} \left\| A - \Gamma \right\|_2^2 + \| \Gamma + B \|_1   \right)
	\\
	&\leq \min_{\Gamma} \left(    \frac {1} {2 \, \psi} \left\| A - \Gamma \right\|_2^2 + \| \Gamma  \|_1   \right)
	+ \| B \|_1
	=     g_\psi(A) + \| B \|_1 ,
	\end{align*}
	where in the second step we reparameterized
	$\Gamma \mapsto \Gamma + B$ in the minimization problem,
	in the third step we used the  triangle inequality for the nuclear norm,
	and in the final step we employed again \eqref{DefGmatrix}.
	We have thus shown the first statement of this part.
	The second statement is obtained from the first statement  by replacing
	$B \mapsto -B$ and $A \mapsto A+B$.
	
	\# \underline{Part (iii):}    
	We first show the result for $\psi=0$. 
	Let $\M_\lambda \, A \,  \M_f = U_1 S_1 V_1'$
	and $ \P_\lambda \, A \, \P_f  = U_2 S_2 V_2'$ be the singular value decompositions
	of those $N \times T$ matrices.    
	We then have $\| \M_\lambda \, A \,  \M_f  \|_1 = {\rm Tr}[ V_1 (\M_\lambda \, A \,  \M_f) U_1' ]$
	and $\| \P_\lambda \, A \,  \P_f  \|_1 = {\rm Tr}[ V_2 (\P_\lambda \, A \,  \P_f) U_2' ]$.
	Furthermore, we have $g_0(A) = \|A\|_1 = \max_{\|C\| \leq 1} {\rm Tr}(C'A) $.
	By choosing $C^* = U_1  V_1' + U_2 V_2'$ we obtain
	\begin{align}
	\|A\|_1 \geq {\rm Tr}(C^{* \prime} A)  =  {\rm Tr}[ V_1 (\M_\lambda \, A \,  \M_f) U_1' ] + {\rm Tr}[ V_2 (\P_\lambda \, A \,  \P_f) U_2' ]
	= \left\| \M_\lambda \, A \,  \M_f \|_1 +  \| \P_\lambda \, A \, \P_f \right\|_1 ,
	\label{ResultPart3psi0}
	\end{align}
	which is the statement of part (iii) of the lemma for $\psi=0$.    
	For $\psi>0$ we find
	\begin{align*}
	&  g_\psi(A)
	=  \min_{\Gamma} \left(    \frac {1} {2 \, \psi} \left\| A - \Gamma \right\|_2^2 + \| \Gamma  \|_1   \right)
	\geq 
	\min_{\Gamma} \left(    \frac {1} {2 \, \psi} \left\| A - \Gamma \right\|_2^2 
	+ \left\| \M_\lambda  \Gamma   \M_f \right\|_1 + \left\| \P_\lambda  \Gamma  \P_f \right\|_1   \right)
	\\
	&=      \min_{\Gamma}  \bigg[   \frac {1} {2 \, \psi} \left(
	\left\|  \M_\lambda  (A-\Gamma)   \M_f \right\|_2^2 + \left\| \P_\lambda  (A - \Gamma)  \P_f \right\|_2^2 
	+    \left\|  \P_\lambda  (A-\Gamma)   \M_f \right\|_2^2 + \left\| \M_\lambda  (A - \Gamma)  \P_f \right\|_2^2 
	\right)  
	\\ & \qquad \qquad\qquad\qquad\qquad  \qquad \qquad\qquad\qquad\qquad   \qquad \qquad\qquad\qquad  \quad 
	+ \left\| \M_\lambda  \Gamma   \M_f \right\|_1 + \left\| \P_\lambda  \Gamma  \P_f \right\|_1
	\bigg]    
	\\
	&=        \min_{\Gamma}  \bigg[   \frac {1} {2 \, \psi} \left(
	\left\|  \M_\lambda  (A-\Gamma)   \M_f \right\|_2^2 + \left\| \P_\lambda  (A - \Gamma)  \P_f \right\|_2^2 
	\right) 
	+ \left\| \M_\lambda  \Gamma   \M_f \right\|_1 + \left\| \P_\lambda  \Gamma  \P_f \right\|_1
	\bigg]       
	\\
	&\geq  \min_{\Gamma}  \left(   \frac {1} {2 \, \psi} 
	\left\|  \M_\lambda  (A-\Gamma)   \M_f \right\|_2^2 
	+ \left\| \M_\lambda  \Gamma   \M_f \right\|_1 \right)
	+   \min_{\Gamma}  \left(   \frac {1} {2  \psi} 
	\left\|  \P_\lambda  (A-\Gamma)   \P_f \right\|_2^2 
	+ \left\| \P_\lambda  \Gamma   \P_f \right\|_1 \right)     
	\\
	&\geq  \min_{\Gamma}  \left(   \frac {1} {2 \, \psi} 
	\left\|  \M_\lambda  A   \M_f  -\Gamma \right\|_2^2 
	+ \left\|    \Gamma   \right\|_1 \right)
	+   \min_{\Gamma}  \left(   \frac {1} {2 \, \psi} 
	\left\|  \P_\lambda   A    \P_f -\Gamma \right\|_2^2 
	+ \left\|    \Gamma   \right\|_1 \right)     
	\\
	&=       g_\psi(\M_\lambda  A   \M_f) +        g_\psi(\P_\lambda  A   \P_f) ,
	\end{align*}
	where in the first step we used \eqref{DefGmatrix};
	in the second step we used \eqref{ResultPart3psi0} with $A$ replaced by $\Gamma$;
	in the third step we decomposed $\left\| A - \Gamma \right\|_2^2 $ into four parts;
	in the fourth step we used that the minimization over $\Gamma$ implies that 
	$\left\|  \P_\lambda  (A-\Gamma)   \M_f \right\|_2^2 = 0$
	and  $\left\| \M_\lambda  (A - \Gamma)  \P_f \right\|_2^2 = 0$ at the optimum,
	because the components $\P_\lambda \Gamma   \M_f$ and $\M_\lambda \Gamma  \P_f $
	of $\Gamma$ appear nowhere else in the objective function, so that choosing
	$\P_\lambda \Gamma   \M_f = \P_\lambda A   \M_f$ and $\M_\lambda \Gamma  \P_f = \M_\lambda A \P_f $ is optimal;
	the fifth step is obvious (it is actually an equality, which is less obvious, but not required for our argument);
	in the sixth step we  replaced $\M_\lambda  \Gamma   \M_f$ and $ \P_\lambda  \Gamma   \P_f $ by
	an unrestricted $\Gamma$ in the minimization problems, which can only make the minimizing values smaller 
	(again, this is actually an equality, but $\leq$ is sufficient to show here); and the final step again employs \eqref{DefGmatrix}.
	We have thus shown the desired result.
\end{proof}

Before presenting the next lemma it is useful to introduce some further notation.
For $\beta \in \mathbb{R}^K$
let $\Delta \beta : = \beta - \beta_0$.
Let $\lambda_X$ be an $N \times R_{\rm c} $ matrix such that
the column span of $\lambda_X$ equals the column span of the $N \times TK$ 
matrix $[X_1,\ldots,X_K] $.
Analogously, let $f_X$ be an $T \times R_{\rm r}$ matrix such that
the column span of $f_X$ equals the column span of the $T \times NK$ matrix 
$[X_1',\ldots,X_K']$.

\begin{lemma}
	\label{lemma:DiffQbound}
	Let model \eqref{ModelBasic} hold.
	Then,  the penalized profiled objective function $Q_\psi(\beta)$ 
	defined in  \eqref{DefQpsi} satisfies, for all $\beta \in \mathbb{R}^K$, and all $\psi >0$,
	\begin{align*}
	\frac{	Q_\psi(\beta ) - Q_\psi(\beta_0) } {\psi}
	&\geq       
	g_\psi\left( \frac { \M_{\lambda_0} (\Delta \beta \cdot X) \M_{f_0} } {\sqrt{NT}} \right)
	-   \left\| \frac{ \P_{\lambda_0} (\Delta \beta \cdot X) \P_{f_0}} {\sqrt{NT}} \right\|_1 
	-  \frac {\psi} 2 \,  {\rm rank}(\Gamma_0)     
	\\
	& \qquad \qquad \qquad
	-  \left\| \frac{ \P_{[\lambda_0,\lambda_X]}  \, E \,  \P_{[f_0,f_X]} } {\sqrt{NT}} \right\|_1
	- \left\| \frac{ E -  \M_{[\lambda_0,\lambda_X]} E \M_{[f_0,f_X]} } {\sqrt{NT}} \right\|_1.
	\end{align*}
	For $\psi=0$ the same bound holds if one replaces 
	$\psi^{-1} \left[ Q_\psi(\beta ) - Q_\psi(\beta_0) \right]$
	by its $\psi \rightarrow 0$ limit 
	$   \big\| (Y - \beta \cdot X) / \sqrt{NT}   \big\|_1 -     \big\| (Y - \beta_0 \cdot X) / \sqrt{NT}   \big\|_1$. 
\end{lemma} 

\begin{proof}[\bf Proof of Lemma~\ref{lemma:DiffQbound}]
	We have
	\begin{align*}
	g_\psi\left( \frac{Y - \beta \cdot X} {\sqrt{NT}}  \right)
	&=  
	g_\psi\left( \frac{\Gamma_0     -  \Delta \beta \cdot X + E} {\sqrt{NT}}  \right)
	\\  
	&\geq   
	g_\psi\left( \frac{\P_{[\lambda_0,\lambda_X]} ( \Gamma_0     -  \Delta \beta \cdot X + E) \P_{[f_0,f_X]}} {\sqrt{NT}}   \right)
	+   g_\psi\left( \frac{\M_{[\lambda_0,\lambda_X]} \, E \,  \M_{[f_0,f_X]}} {\sqrt{NT}}  \right)  
	\\
	&= 
	g_\psi\left( \frac{ \Gamma_0     -  \Delta \beta \cdot X } {\sqrt{NT}}
	+ \frac{ \P_{[\lambda_0,\lambda_X]}  E \P_{[f_0,f_X]}} {\sqrt{NT}}   \right)
	+   g_\psi\left( \frac{\M_{[\lambda_0,\lambda_X]} \, E \,  \M_{[f_0,f_X]}} {\sqrt{NT}}  \right)  
	\\
	&\geq     g_\psi\left( \frac{ \Gamma_0     -  \Delta \beta \cdot X } {\sqrt{NT}} \right)
	- \left\| \frac{ \P_{[\lambda_0,\lambda_X]}  E \P_{[f_0,f_X]}} {\sqrt{NT}} \right\|_1
	+   g_\psi\left( \frac{\M_{[\lambda_0,\lambda_X]} \, E \,  \M_{[f_0,f_X]}} {\sqrt{NT}}  \right).
	\end{align*}
	Here, we first plugged in the model for $Y$,
	then used part (iii) of Lemma~\ref{lemma:PropsG} with $\lambda = [\lambda_0,\lambda_X]$
	and $f = [f_0,f_X]$, and in the final step used part (ii) of Lemma~\ref{lemma:PropsG}.
	In the same way we obtain
	\begin{align*}
	g_\psi\left( \frac{ \Gamma_0     -  \Delta \beta \cdot X } {\sqrt{NT}} \right)
	&\geq 
	g_\psi\left( \frac{ \P_{\lambda_0} ( \Gamma_0     -  \Delta \beta \cdot X) \P_{f_0}} {\sqrt{NT}} \right)
	+   g_\psi\left( \frac{  \M_{\lambda_0} (\Delta \beta \cdot X) \M_{f_0} } {\sqrt{NT}} \right)
	\\
	&=     g_\psi\left( \frac{\Gamma_0 } {\sqrt{NT}} - \frac{ \P_{\lambda_0} (   \Delta \beta \cdot X) \P_{f_0}} {\sqrt{NT}} \right)
	+   g_\psi\left( \frac{  \M_{\lambda_0} ( \Delta \beta \cdot X) \M_{f_0} } {\sqrt{NT}} \right)
	\\
	&\geq     g_\psi\left( \frac{\Gamma_0 } {\sqrt{NT}}  \right)
	- \left\| \frac{ \P_{\lambda_0} (   \Delta \beta \cdot X) \P_{f_0}} {\sqrt{NT}} \right\|_1
	+   g_\psi\left( \frac{  \M_{\lambda_0} ( \Delta \beta \cdot X) \M_{f_0} } {\sqrt{NT}} \right)
	\\  
	&\geq    \left\|  \frac{\Gamma_0 } {\sqrt{NT}}  \right\|_1 - \frac \psi 2 \, {\rm rank}(\Gamma_0)
	- \left\| \frac{ \P_{\lambda_0} (   \Delta \beta \cdot X) \P_{f_0}} {\sqrt{NT}} \right\|_1
	+   g_\psi\left( \frac{  \M_{\lambda_0} ( \Delta \beta \cdot X) \M_{f_0} } {\sqrt{NT}} \right) ,
	\end{align*}
	where in the last step we also used part (i) of Lemma~\ref{lemma:PropsG}.
	Furthermore, we find
	\begin{align*}
	g_\psi\left( \frac{Y - \beta_0 \cdot X} {\sqrt{NT}}  \right)
	&=  
	g_\psi\left( \frac{E + \Gamma_0    } {\sqrt{NT}}  \right)
	=  
	g_\psi\left( \frac{\M_{[\lambda_0,\lambda_X]} E \M_{[f_0,f_X]}  + 
		\left( E -  \M_{[\lambda_0,\lambda_X]} E \M_{[f_0,f_X]} \right) +  \Gamma_0  } {\sqrt{NT}}  \right)  
	\\
	&\leq       g_\psi\left( \frac{\M_{[\lambda_0,\lambda_X]} E \M_{[f_0,f_X]}  } {\sqrt{NT}}    \right)  + 
	\left\| \frac{ E -  \M_{[\lambda_0,\lambda_X]} E \M_{[f_0,f_X]} }  {\sqrt{NT}} \right\|_1   +
	\left\|  \frac{  \Gamma_0  } {\sqrt{NT}}  \right\|_1 ,
	\end{align*}
	where we used part (ii) of Lemma~\ref{lemma:PropsG} and the triangle inequality for the nuclear norm.
	Combining the inequalities in the last three displays gives
	\begin{align*}
	g_\psi\left( \frac{Y - \beta \cdot X} {\sqrt{NT}}  \right) - g_\psi\left( \frac{Y - \beta_0 \cdot X} {\sqrt{NT}}  \right)
	&\geq       
	g_\psi\left( \frac { \M_{\lambda_0} (\Delta \beta \cdot X) \M_{f_0} } {\sqrt{NT}} \right)
	-   \left\| \frac{ \P_{\lambda_0} (\Delta \beta \cdot X) \P_{f_0}} {\sqrt{NT}} \right\|_1 
	-  \frac {\psi} 2   {\rm rank}(\Gamma_0)     
	\\
	& \qquad 
	-  \left\| \frac{ \P_{[\lambda_0,\lambda_X]}   E   \P_{[f_0,f_X]} } {\sqrt{NT}} \right\|_1
	- \left\| \frac{ E -  \M_{[\lambda_0,\lambda_X]} E \M_{[f_0,f_X]} } {\sqrt{NT}} \right\|_1.
	\end{align*}
	The derivation so far was valid for all $\psi \geq 0$.
	For $\psi=0$ the left hand side of the last display simply 
	is $  \big\| (Y - \beta \cdot X) / \sqrt{NT}   \big\|_1 -     \big\| (Y - \beta_0 \cdot X) / \sqrt{NT}   \big\|_1$.
	For $\psi>0$ we have, by \eqref{DefGmatrix},
	\begin{align*}
	\frac{	Q_\psi(\beta ) - Q_\psi(\beta_0) } {\psi}
	&= g_\psi\left( \frac{Y - \beta \cdot X} {\sqrt{NT}}  \right) -  g_\psi\left( \frac{Y - \beta_0 \cdot X} {\sqrt{NT}}  \right) ,
	\end{align*}
	so that we have shown the statement of the lemma.
\end{proof}

\begin{lemma}
	\label{lemma:unique.separationInequality}
	Let model \eqref{model:linear.matrix} hold, and
	let 	$\mathbb{E}[ \left.  E_{it} \, \right|  \, X   ]  = 0$, and  
	$  \mathbb{E}\left[  \left. E_{it}^2  \, \right|  \, X    \right] < \infty$, for all $i,t$.
	Then we have, for all $\psi >0$,
	\begin{align*}  
	g_\psi\left( \frac { \M_{\lambda_0} (\Delta \bar \beta_\psi \cdot X) \M_{f_0} } {\sqrt{NT}} \right)
	-   \left\| \frac{ \P_{\lambda_0} (\Delta \bar \beta_\psi \cdot X) \P_{f_0}} {\sqrt{NT}} \right\|_1 
	\leq  \frac {\psi} 2 \,  {\rm rank}(\Gamma_0)     .
	\end{align*}      	
\end{lemma}

\begin{proof}[\bf Proof of Lemma~\ref{lemma:unique.separationInequality}]
	Using the model and the assumptions on $E_{it}$ in the proposition we find
	\begin{align*}
	\mathbb{E}\left[  \left\| Y -\beta \cdot X- \Gamma \right\|_2^2 \Big| X \right]
	&=    \sum_{i=1}^N \sum_{t=1}^T  \mathbb{E}\left[ \left. \left( \Gamma_{0,it} - \Gamma_{it}     -  X_{it}'  \, \Delta \beta  + E_{it}  \right)^2     \right| X \right]
	\\
	&=   \sum_{i=1}^N \sum_{t=1}^T   \left( \Gamma_{0,it}   - \Gamma_{it}     -  X_{it}' \, \Delta \beta     \right)^2   
	+    \sum_{i=1}^N \sum_{t=1}^T  \mathbb{E}\left( \left. E_{it}^2  \right| X \right)
	\\
	&=   \left\| \Gamma_0  - \Gamma   -  \Delta \beta \cdot X \right\|_2^2 +   \mathbb{E}\left( \left. \left\|  E \right\|_2^2 \right| X \right) ,
	\end{align*}
	where the expectation is also implicitly conditional on $\Gamma_0$, because $\Gamma_0$ is treated as non-random throughout the
	whole paper. Because $ \mathbb{E}\left( \left. \left\|  E \right\|_2^2 \right| X \right)$ is just a constant that does not depend on the parameters
	$\beta$ and $\Gamma$, we can thus rewrite the definition of  $\bar \beta_\psi $   in \eqref{DefBarBeta} as
	\begin{align*}
	\bar \beta_\psi  &= \argmin_{\beta} \overline Q_\psi(\beta )   ,
	&
	\overline Q_\psi(\beta )  
	&:=
	\min_\Gamma 
	\left\{ \frac{1}{2NT} \,   \left\| \Gamma_0  - \Gamma   -  \Delta \beta \cdot X \right\|_2^2
	+  \frac{ \psi } {\sqrt{NT}} \left\| \Gamma \right\|_1  \right\} .
	\end{align*}
	We can obtain  $ \overline Q_\psi(\beta )  $ from the profiled objective function $Q_\psi(\beta )  $ that was defined in \eqref{DefQpsi}
	by simply setting $E=0$ in the model \eqref{model:linear.matrix}. 
	The bound on $\psi^{-1} \left[ Q_\psi(\beta ) - Q_\psi(\beta_0) \right]$ in Lemma~\ref{lemma:DiffQbound} is therefore applicable to 
	$ \overline Q_\psi(\beta )  $  if we just set $E=0$ in that lemma. We thus have,
	for all $\beta \in \mathbb{R}^K$, 
	\begin{align*}
	\frac{\overline	Q_\psi(\beta ) - \overline Q_\psi(\beta_0) } {\psi}
	&\geq       
	g_\psi\left( \frac { \M_{\lambda_0} (\Delta \beta \cdot X) \M_{f_0} } {\sqrt{NT}} \right)
	-   \left\| \frac{ \P_{\lambda_0} (\Delta \beta \cdot X) \P_{f_0}} {\sqrt{NT}} \right\|_1 
	-  \frac {\psi} 2 \,  {\rm rank}(\Gamma_0)     .
	\end{align*}
	We have $ Q_\psi(\bar \beta_\psi ) - Q_\psi(\beta_0)  \leq 0 $, because $\bar \beta_\psi  $ minimizes $Q_\psi(\beta)$,
	and combining this with the result in the last display gives the statement of the lemma.
\end{proof}

\begin{proof}[\bf Proof of Proposition~\ref{prop:unique.separation}]
	Let
	\begin{align*}
	c  & =  \min_{\left\{ \alpha \in \mathbb{R}^{K} \, : \, \| \alpha\|=1 \right\}} C(\alpha) ,
	&
	C(\alpha) =	   
	\frac {   \left\|  \M_{\lambda_0} (\alpha \cdot X) \M_{f_0}\right\|_1
		-   \left\|  \P_{\lambda_0} (\alpha \cdot X) \P_{f_0} \right\|_1  } {\sqrt{NT}}  .
	\end{align*}
	Using the absolute homogeneity of the nuclear norm this definition implies that for any $ \alpha \in \mathbb{R}^{K} $ we have
	\begin{align}
	c \, \left\| \alpha  \right\|  
	\leq   \left\| \frac { \M_{\lambda_0} (\alpha \cdot X) \M_{f_0} } {\sqrt{NT}} \right\|_1
	-   \left\| \frac{ \P_{\lambda_0} (\alpha \cdot X) \P_{f_0}} {\sqrt{NT}} \right\|_1 .
	\label{BoundNormAlpha}
	\end{align}
	Since the ball $	\left\{ \alpha \in \mathbb{R}^{K} \, : \, \| \alpha\|=1 
	\right\}$ is a compact set and $C(\alpha)$ is a continuous function,
	there exists a value $\alpha^* \in  	\left\{ \alpha \in \mathbb{R}^{K} \, : \, \| \alpha\|=1 \right\}$ where the minimum is attained,
	that is, $c= C(\alpha^*)$. By the assumption on the regressors in Proposition~\ref{prop:unique.separation}
	we thus have $c= C(\alpha^*)>0$.

	Next, applying part (i) of Lemma \ref{lemma:PropsG} we obtain 
	\begin{align}
	g_\psi\left( \frac { \M_{\lambda_0} (\Delta \bar \beta_\psi \cdot X) \M_{f_0} } {\sqrt{NT}} \right) 
	\geq  \left\| \frac { \M_{\lambda_0} (\Delta \bar \beta_\psi \cdot X) \M_{f_0} } {\sqrt{NT}} \right\|_1
	- \frac \psi 2  {\rm rank}\left[ \M_{\lambda_0} (\Delta \bar \beta_\psi \cdot X) \M_{f_0} \right] ,
	\label{LowerBoundGfct} 
	\end{align}
	and also using Lemma~\ref{lemma:unique.separationInequality} we thus find that
	\begin{align*}  
	\left\| \frac { \M_{\lambda_0} (\Delta \bar \beta_\psi \cdot X) \M_{f_0} } {\sqrt{NT}} \right\|_1
	-   \left\| \frac{ \P_{\lambda_0} (\Delta \bar \beta_\psi \cdot X) \P_{f_0}} {\sqrt{NT}} \right\|_1 
	&\leq  \frac {\psi} 2 \left\{  {\rm rank}(\Gamma_0)  + {\rm rank}\left[ \M_{\lambda_0} (\Delta \bar \beta_\psi \cdot X) \M_{f_0} \right] \right\}   
	\\
	&\leq  \frac {\psi} 2 \left\{  {\rm rank}(\Gamma_0)  + \max_{ \alpha \in \mathbb{R}^{K} } {\rm rank}\left[ \M_{\lambda_0} (\alpha \cdot X) \M_{f_0} \right] \right\}   .
	\end{align*}
	From this and  \eqref{BoundNormAlpha} with $\alpha = \Delta \bar \beta_\psi $ we obtain for any $\psi>0$ that\footnote{
		The bound \eqref{BoundDeltaBeta} is sufficient for our purposes since we ultimately consider the limit
		$\psi \rightarrow 0$ here, but for a  fixed value of $\psi$ (and $N,T$) this bound
		is potentially very crude if high-rank regressors $X_k$ are present.
		From Lemma~\ref{lemma:unique.separationInequality} one could then obtain a sharper bound on $ \bar \beta_\psi - \beta_0$
		by not using part (i) of Lemma~\ref{lemma:PropsG} to simplify $ g_\psi\left[ \left( \M_{\lambda_0} (\Delta \bar \beta_\psi \cdot X) \M_{f_0} \right) / \sqrt{NT} \right] $.
	}
	\begin{align}
	\left\|  \bar \beta_\psi - \beta_0  \right\| \leq   
	\frac {\psi} {2c} \left\{  {\rm rank}(\Gamma_0)  + \max_{ \alpha \in \mathbb{R}^{K} } {\rm rank}\left[ \M_{\lambda_0} (\alpha \cdot X) \M_{f_0} \right]  \right\}  ,
	\label{BoundDeltaBeta}
	\end{align}
	and therefore $ \left\|  \bar \beta_\psi - \beta_0  \right\| = O(\psi)$, as $\psi \rightarrow 0$.
\end{proof}

\subsection{Proofs for Section~\ref{sec:LowRank}}

\begin{lemma}
	\label{lemma:LowRankConsistency}
	Let $R_{\rm c} :=  {\rm rank}( [X_1,\ldots,X_K] )$
	and $R_{\rm r} := {\rm rank}( [X_1',\ldots,X_K'] ) $. Assume that
	\begin{align*}
	C  & :=  \min_{\left\{ \alpha \in \mathbb{R}^{K} \, : \, \| \alpha\|=1 \right\}}  
	\left\| \frac { \M_{\lambda_0} (\alpha  \cdot X) \M_{f_0} } {\sqrt{NT}}\right\|_1
	-   \left\| \frac{ \P_{\lambda_0} (\alpha \cdot X) \P_{f_0}} {\sqrt{NT}} \right\|_1  
	\end{align*}
	satisfies $C>0$. Then we have, for all $\psi>0$,
	\begin{align*}
	\left\| \widehat \beta_\psi  - \beta_0 \right\|
	\leq \frac 1 C \left[   \left( \frac{\psi} {2} +  \frac{\left\| E \right\|_\infty} {\sqrt{NT}} \right) \left[ R_0 + \min(R_{\rm c}, R_{\rm r}) \right]
	+ \frac{\left\| E \right\|_\infty} {\sqrt{NT}}
	\left(  2\, R_0 + R_{\rm c}  + R_{\rm r} \right) \right] ,
	\end{align*}
	and
	\begin{align*}
	\left\| \widehat \beta_*  - \beta_0 \right\|
	\leq \frac 1 C \, \frac{\left\| E \right\|_\infty} {\sqrt{NT}}  \left[ 3 \, R_0 + R_{\rm c}  + R_{\rm r} + \min(R_{\rm c}, R_{\rm r}) \right] .
	\end{align*}
\end{lemma}

\begin{proof}[\bf Proof of Lemma~\ref{lemma:LowRankConsistency}]
	By definition we have $Q_\psi( \widehat \beta_\psi) - Q_\psi(\beta_0) \leq 0$.
	Combining this
	with Lemma~\ref{lemma:DiffQbound} and equation \eqref{LowerBoundGfct},
	and writing ${\rm rank}(\Gamma_0) = R_0$, we obtain
	\begin{align*}
	0
	\geq       
	\left\| \frac { \M_{\lambda_0} (\Delta \widehat \beta_\psi \cdot X) \M_{f_0} } {\sqrt{NT}}\right\|_1
	-   \left\| \frac{ \P_{\lambda_0} (\Delta \widehat \beta_\psi \cdot X) \P_{f_0}} {\sqrt{NT}} \right\|_1 
	-  \frac {\psi} 2 \left\{ R_0   + \max_{ \alpha \in \mathbb{R}^{K} } {\rm rank}\left[ \M_{\lambda_0} (\alpha \cdot X) \M_{f_0} \right] \right\} 
	\\
	-  \left\| \frac{ \P_{[\lambda_0,\lambda_X]}  \, E \,  \P_{[f_0,f_X]} } {\sqrt{NT}} \right\|_1
	- \left\| \frac{ E -  \M_{[\lambda_0,\lambda_X]} E \M_{[f_0,f_X]} } {\sqrt{NT}} \right\|_1.
	\end{align*}    
	The definition of $c$ in the theorem together with the absolute homogeneity of the nuclear norm implies
	\begin{align*}
	c \, \left\| \Delta \widehat \beta_\psi \right\|
	\leq
	\left\| \frac { \M_{\lambda_0} (\Delta \widehat \beta_\psi \cdot X) \M_{f_0} } {\sqrt{NT}}\right\|_1
	-   \left\| \frac{ \P_{\lambda_0} (\Delta \widehat \beta_\psi \cdot X) \P_{f_0}} {\sqrt{NT}} \right\|_1 .
	\end{align*}
	We have
	\begin{align*}
	\max_{ \alpha \in \mathbb{R}^{K} } {\rm rank}\left[ \M_{\lambda_0} (\alpha \cdot X) \M_{f_0} \right]
	\leq 
	\max_{ \alpha \in \mathbb{R}^{K} } {\rm rank} (\alpha \cdot X)  \leq \min(R_{\rm c}, R_{\rm r}) ,
	\end{align*}
	because we have $\alpha \cdot X = [X_1,\ldots,X_K] (\alpha \otimes \mathbb{I}_T)$,
	and therefore ${\rm rank} (\alpha \cdot X)  \leq R_{\rm c}$,
	and also $(\alpha \cdot X)' =[X_1',\ldots,X_K'] (\alpha \otimes \mathbb{I}_N)$,
	and therefore ${\rm rank} (\alpha \cdot X)  \leq R_{\rm r}$.
	
	We also have
	\begin{align*}
	\left\| \frac{ \P_{[\lambda_0,\lambda_X]}  \, E \,  \P_{[f_0,f_X]} } {\sqrt{NT}} \right\|_1
	& \leq 
	\left\| \frac{ \P_{[\lambda_0,\lambda_X]}  \, E \,  \P_{[f_0,f_X]} } {\sqrt{NT}} \right\|_\infty
	{\rm rank}\left(    \P_{[\lambda_0,\lambda_X]}  \, E \,  \P_{[f_0,f_X]}    \right)
	\\  
	&\leq  \frac{\left\| E \right\|_\infty} {\sqrt{NT}}
	\min\left\{  {\rm rank}\left(   \P_{[\lambda_0,\lambda_X]}   \right) ,  {\rm rank}\left( \P_{[f_0,f_X]}    \right)  \right\}
	\\
	&=     \frac{\left\| E \right\|_\infty} {\sqrt{NT}}
	\min\left\{ R_0 +  R_{\rm c} ,   R_0 +  R_{\rm r}  \right\}
	=    \frac{\left\| E \right\|_\infty} {\sqrt{NT}}
	\left[ R_0 + \min(R_{\rm c}, R_{\rm r}) \right] ,
	\end{align*}
	and similarly
	\begin{align*}
	\left\|  \frac{ E -  \M_{[\lambda_0,\lambda_X]} E \M_{[f_0,f_X]} } {\sqrt{NT}}  \right\|_1
	& =
	\left\|  \frac{ \P_{[\lambda_0,\lambda_X]}  \, E   } {\sqrt{NT}} 
	+  \frac{ \M_{[\lambda_0,\lambda_X]}  \, E \,  \P_{[f_0,f_X]} } {\sqrt{NT}}  \right\|_1
	\\
	&\leq 
	\left\|  \frac{ \P_{[\lambda_0,\lambda_X]}  \, E   } {\sqrt{NT}} \right\|_1
	+  \left\|  \frac{ \M_{[\lambda_0,\lambda_X]}  \, E \,  \P_{[f_0,f_X]} } {\sqrt{NT}}   \right\|_1
	\\
	&\leq      \frac{\left\| E \right\|_\infty} {\sqrt{NT}} \,  {\rm rank}\left(   \P_{[\lambda_0,\lambda_X]}   \right)
	+    \frac{\left\| E \right\|_\infty} {\sqrt{NT}} \,  {\rm rank}\left(   \P_{[f_0,f_X]}    \right)
	\\
	&=  \frac{\left\| E \right\|_\infty} {\sqrt{NT}}   \left(  2\, R_0 + R_{\rm c}  + R_{\rm r} \right)  .
	\end{align*}
	Combining the above inequalities gives the finite sample bound in the theorem,
	\begin{align*}
	c \left\| \widehat \beta_\psi  - \beta_0 \right\|
	\leq    \left( \frac{\psi} {2} +  \frac{\left\| E \right\|_\infty} {\sqrt{NT}} \right) \left[ R_0 + \min(R_{\rm c}, R_{\rm r}) \right]
	+ \frac{\left\| E \right\|_\infty} {\sqrt{NT}}
	\left(  2\, R_0 + R_{\rm c}  + R_{\rm r} \right) ,
	\end{align*}
	and the same bound holds for $ \widehat \beta_*$ if we set $\psi=0$, because
	all bounds above, including Lemma~\ref{lemma:DiffQbound} are applicable for $\psi=0$ as well.
	Finally, the asymptotic statements in the theorem are immediate corollaries of the finite sample bounds.
\end{proof}

\begin{proof}[\bf Proof of Theorem~\ref{th:LowRankConsistency}]
     The theorem follows immediately from Lemma~\ref{lemma:LowRankConsistency},
     because our assumptions guarantee that  $C \geq c > 0$ (and therefore $1/C = O(1)$),
     $R_0 = O_P(1)$,
    $ R_{\rm c}  = O_P(1)$,
    $R_{\rm r} = O_P(1) $,
    and
    $$
         \frac{\left\| E \right\|_\infty} {\sqrt{NT}}    =  O_P\left( \frac 1 {\sqrt{\min(N,T)}} \right) .
    $$
\end{proof}

\subsection{Proofs for Section \ref{subsec:nuc.regularized}}

\begin{lemma} \label{lemma:ap.useful.facts}
	Suppose that $A$ and $B$ are two matrices with ranks of $A$ and $B$ are ${\rm{rank}(A)}$ and ${\rm{rank}(B)}$, respectively.
	\begin{itemize}
		\item[(i)] $\| A \|_{\infty} \leq \| A \|_2 \leq \| A \|_1 \leq \sqrt{{\rm{rank}(A)}} \| A \|_{2} \leq {\rm{rank}(A)} \| A \|_{\infty} $.
		\item[(ii)] $ \| AB \|_{\infty} \leq \| A \|_{\infty} \| B \|_{\infty}$.
		\item[(iii)] $ \| AB \|_2 \leq \| A \|_{\infty} \| B \|_2 \leq \| A \|_2 \| B \|_2$.
		\item[(iv)] If $AB'=0$ and $A'B=0$, then $\| A + B \|_{\infty} = \max ( \| A \|_{\infty}, \| B \|_{\infty} )$.
		\item[(v)] If $A'B = 0$ (or equivalently $B'A = 0$), then $\|A+B\|_{\infty}^2 \leq \|A \|_{\infty}^2 + \| B \|_{\infty}^2$.
	\end{itemize}
\end{lemma}

Recall that the rank of $\Gamma_0 = \lambda_0 f_0^{\prime}$ is $R_0$, which is fixed. Throughout the rest of the appendix, we use the following singular value decomposition of $\Gamma_0$, \begin{equation} 
\Gamma_0 = U S V', \label{eq.svd.Gamma0}
\end{equation}
where $U \in \mathbb{R}^{N \times R_0}$ with $U'U = \mathbf{I}_{R_0}$, $V \in \mathbb{R}^{T \times R_0}$ with $V'V = \mathbf{I}_{R_0}$, $S$ is the $R_0 \times R_0$ diagonal matrix of singular values of $\Gamma_0$. 

Suppose that $f_0$ is normalized as $\frac{1}{T} f_0' f_0  = \mathbf{I}_{R_0}$. Then, we have
\[ 
f_0 = \sqrt{T} V,  \qquad \lambda_0 = \frac{U S}{\sqrt{T}}.
\]

Some further notation:
\begin{align*}
 L(\beta,\Gamma) &= \frac{1}{2 NT} \| Y - \beta \cdot X - \Gamma \|_2^2,
&
Q_{\psi}(\beta,\Gamma) 
&= \frac{1}{2NT} \| Y - \beta \cdot X - \Gamma \|_2^2  + \frac{\psi} {\sqrt{NT}} \| \Gamma \|_1.
\end{align*}
Let   
\begin{align*}
Q_{\psi}(\Gamma)  &:= \inf_\beta \, Q_{\psi}(\beta,\Gamma), &
L(\Gamma)  &:= \inf_\beta \, L(\beta,\Gamma) .
\end{align*}
These are the profile objective functions of $Q_{\psi}(\beta,\Gamma)$ and $L(\beta,\Gamma)$, respectively, which concentrate out parameter the $\beta$.
We also use the notation $\Theta := \Gamma - \Gamma_0$ and $\theta := {\rm vec}(\Theta)$.

\bigskip

\begin{proof}[\bf Proof of Lemma~\ref{lemma:ConTheta}]~
	
	\noindent
	\# \underline{\bf Step 1: Use (\ref{def.psi}) to show $ \widehat{\Theta}_{\psi} \in \mathbb{C}$}
	
	\noindent
	By definition, we have
	\begin{align*}
	0 & \geq Q_{\psi}(\Gamma_0 + \widehat{\Theta}_{\psi}) - Q_{\psi}(\Gamma_0) \\
	&= L(\Gamma_0 + \widehat{\Theta}_{\psi}) - L(\Gamma_0) + \frac{\psi}{\sqrt{NT}}  \left( \| \Gamma_0 + \widehat{\Theta}_{\psi} \|_1 - \| \Gamma_0 \|_1 \right),
	\end{align*}
	where $\widehat{\Theta}_{\psi} := \widehat{\Gamma}_{\psi} - \Gamma_0$. Let $\widehat{\theta}_{\psi}:= {\rm vec }(\widehat{\Theta}_{\psi}), \: \widehat{\Theta}_{\psi,1} := \M_{U_0} \widehat{\Theta}_{\psi} \M_{V_0}$ and $\widehat{\Theta}_{\psi,2} := \widehat{\Theta}_{\psi} - \M_{U_0} \widehat{\Theta}_{\psi} \M_{V_0}$. 
	Then
	\begin{align*}
	L(\Gamma_0 + \widehat{\Theta}_{\psi}) - L(\Gamma_0) 
	& = \frac{1}{2NT} \widehat{\theta}_{\psi}' \M_x \widehat{\theta}_{\psi} - \frac{1}{NT} e' \M_x \widehat{\theta}_{\psi} \\
	& \geq - \frac{1}{NT} e' \M_x \widehat{\theta}_{\psi} \\
	& =  - \frac{1}{NT} \Tr(  \widehat{\Theta}_{\psi}' \,  {\rm mat}(\M_x e)  )
	\\
	& \geq  - \frac{\| \widehat{\Theta}_{\psi} \|_1}{ \sqrt{NT}}      \,  \frac{\| {\rm mat}(\M_x e)  \|_\infty}{ \sqrt{NT}}  
	\\
	& \geq - \frac{\psi}{2}  \frac{ \| \widehat{\Theta}_{\psi} \|_1 }{\sqrt{NT}} \\
	& \geq  - \frac{\psi}{2}  \frac{ \| \widehat{\Theta}_{\psi,1} \|_1 }{\sqrt{NT}} - \frac{\psi}{2}  \frac{ \| \widehat{\Theta}_{\psi,2} \|_1}{\sqrt{NT}}.
	\end{align*}
	Here the first inequality holds since $\widehat{\theta}_{\psi}' \M_x \widehat{\theta}_{\psi} \geq 0$, the second inequality holds by the H\"{o}lder inequality, the third inequality holds by (\ref{def.psi}), and the last inequality holds by the triangle inequality. We furthermore have
	\begin{align*}
	& \frac{\psi}{\sqrt{NT}} \left( \| \Gamma_0 + \widehat{\Theta}_{\psi} \|_1 - \| \Gamma_0 \|_1 \right) \\
	&= \frac{\psi}{\sqrt{NT}} \left( \| \Gamma_0 + \widehat{\Theta}_{\psi,1} + \widehat{\Theta}_{\psi,2} \|_1 - \| \Gamma_0 \|_1 \right) \\
	& \geq \frac{\psi}{\sqrt{NT}} \left( \| \Gamma_0 + \widehat{\Theta}_{\psi,1}  \|_1 - \| \Gamma_0 \|_1 \right) 
	- \frac{\psi}{\sqrt{NT}} \| \widehat{\Theta}_{\psi,2} \|_1 \\
	& = \frac{\psi}{\sqrt{NT}} \|  \widehat{\Theta}_{\psi,1}  \|_1 - \frac{\psi}{\sqrt{NT}} \| \widehat{\Theta}_{\psi,2} \|_1 .
	\end{align*}
	Therefore,
	\begin{align*}
	0 &\geq L(\Gamma_0 + \widehat{\Theta}_{\psi}) - L(\Gamma_0) + \frac{\psi}{\sqrt{NT}} \left( \| \Gamma_0 + \widehat{\Theta}_{\psi} \|_1 - \| \Gamma_0 \|_1 \right) \\
	& \geq - \frac{\psi}{2} \frac{ \| \widehat{\Theta}_{\psi,1} \|_1 }{\sqrt{NT}} - \frac{\psi}{2} \frac{ \| \widehat{\Theta}_{\psi,2} \|_1 }{\sqrt{NT}} 
	+ \psi \frac{ \|  \widehat{\Theta}_{\psi,1}  \|_1 }{\sqrt{NT}} - \psi \frac{ \| \widehat{\Theta}_{\psi,2} \|_1 }{\sqrt{NT}} \\
	& = \frac{\psi}{2} \frac{1}{\sqrt{NT}} \left( \| \widehat{\Theta}_{\psi,1} \|_1 - 3  \| \widehat{\Theta}_{\psi,2} \|_1 \right) .
	\end{align*}
	Thus, we have
	\begin{align*}
	\widehat{\Theta}_{\psi}
	\in
	\mathbb{C}
	:= \left\{
	B \in \mathbb{R}^{N \times T} 
	\;|\; \| \M_{U} B \M_{V}  \|_1 
	\leq 3 \| B - \M_{U} B \M_{V} \|_1
	\right\}. 
	\end{align*}
	
	\noindent
	\# \underline{\bf Step 2: Also use Assumption \ref{ass:RSC} to show the final result:}
	Using Assumption \ref{ass:RSC} and the same derivation as above, we find
	\begin{align*}
	Q_{\psi}(\Gamma_0 + \widehat{\Theta}_{\psi}) - Q_{\psi}(\Gamma_0) 
	& = \frac{1}{2NT} \widehat{\theta}_{\psi}' \M_x \widehat{\theta}_{\psi} - \frac{1}{NT} e' \M_x \widehat{\theta}_{\psi} + \frac{\psi}{\sqrt{NT}} \left( \| \Gamma_0 + \widehat{\Theta}_{\psi} \|_1 - \| \Gamma_0 \|_1 \right) \\
	&\geq
	\frac{\mu}{2NT} \|  \widehat{\Theta}_{\psi}  \|_2^2 
	+
	\frac{\psi}{2} \frac{1}{\sqrt{NT}} \left( \| \widehat{\Theta}_{\psi,1} \|_1 - 3  \| \widehat{\Theta}_{\psi,2} \|_1 \right) 
	\\
	&\geq  \frac{\mu}{2NT} \|  \widehat{\Theta}_{\psi}  \|_2^2 
	- \frac{3 \, \psi}{2}  \frac{1}{\sqrt{NT}}  \| \widehat{\Theta}_{\psi,2} \|_1  .
	\end{align*}
	Because
	$0  \geq Q_{\psi}(\Gamma_0 + \widehat{\Theta}_{\psi}) - Q_{\psi}(\Gamma_0) $
	we thus have
	\begin{align*}
	\frac{\mu}{2NT} \|  \widehat{\Theta}_{\psi}  \|_2^2 
	-  \frac{3 \, \psi}{2}  \frac{1}{\sqrt{NT}}  \| \widehat{\Theta}_{\psi,2} \|_1
	\leq  0 .
	\end{align*}
	Since the rank of $ \widehat{\Theta}_{\psi,2}$ is at most $2R_0$ (e.g., see \cite{RechtFazelParrilo2010}), we have 
	\[
	\| \widehat{\Theta}_{\psi,2} \|_1 \leq \sqrt{2R_0} \|  \widehat{\Theta}_{\psi,2} \|_2
	\]
	and  we also have 
	\[
	\|	 \widehat{\Theta}_{\psi,2} \|_2 \leq \| \widehat{\Theta}_{\psi} \|_2. 
	\]
	Therefore,
	\begin{align*}
	\frac{1}{NT} \|  \widehat{\Theta}_{\psi}  \|_2^2 
	-  \frac{3 \, \psi \, \sqrt{2 R_0}} {\mu} \frac{1}{\sqrt{NT}}  \| \widehat{\Theta}_{\psi} \|_2
	\leq  0 ,
	\end{align*}	
	and
	\begin{align*}
	\frac{ \|  \widehat{\Theta}_{\psi}  \|_2 }{\sqrt{NT}}
	\leq  \frac{3 \, \sqrt{2 R_0} \, \psi} {\mu}  .
	\end{align*}
\end{proof}

\begin{proof} [\bf Proof of Theorem~\ref{theorem:first.stage.estimator}.]~
	
	\noindent 	
	{\bf Part (i).} Part (i) follows by Lemma \ref{lemma:ConTheta} and the condition on $\psi$ in Theorem~\ref{theorem:first.stage.estimator}. 	
	
	\noindent
	{\bf Part (ii).} Let $\widehat{\beta}(\Gamma) = (x'x)^{-1}x'(y - \gamma).$ Then, by definition we have
	\begin{eqnarray*}
		\widehat{\beta}_{\psi} -\beta_0 &:=& \widehat{\beta}({\widehat{\Gamma}}_{\psi}) - \beta_0 
		= \left( \frac{1}{NT} x'x \right)^{-1} \left(\frac{1}{NT} x'e - \frac{1}{NT}x'(\widehat{\gamma}_{\psi} - \gamma_0) \right) .
	\end{eqnarray*}
	Under the assumption of the theorem we have $\left( \frac{1}{NT} x'x \right)^{-1} = \Op(1)$ and $\frac{1}{NT} e'x = \Op(\frac{1}{\sqrt{NT}})$. 
	Also, by Part (a) we have  
	\begin{eqnarray*}
		\left\| \frac{1}{NT}x'(\widehat{\gamma}_{\psi} -\gamma_0)  \right\|_2 &\leq& \frac{1}{\sqrt{NT}} \| X \|_2 \frac{1}{\sqrt{NT}} \| \widehat{\Gamma}_{\psi} -\Gamma_0 \|_2 \\
		&=&  \Op (1) \psi. 
	\end{eqnarray*}
	Combining these, we can deduce the required result for Part (b).
\end{proof}

\begin{proof}[\bf Proof of (\ref{special.Gammahat})]~
	
	Since $\M_x$ is positive semi-definite, $ | e'M_x \widehat{\gamma}_{\psi} | 
	\leq \| \widehat{\Gamma}_{\psi} \|_1 \| {\rm mat}(\M_x e)  \|_\infty$  by 
	the H\"{o}lder inequality, and $\Gamma_0 = 0$, we have
	\begin{align*}
	0 & \geq Q(\widehat{\Gamma}_{\psi}) -  Q(\Gamma_0) \\
	& = \frac{1}{2NT} ( \widehat{\gamma}_{\psi} -\gamma_0)'\M_x ( \widehat{\gamma}_{\psi} -\gamma_0)  - \frac{1}{NT} e'\M_x ( \widehat{\gamma}_{\psi} -\gamma_0)  + \frac{\psi}{\sqrt{NT}} 
	\| \widehat{\Gamma}_{\psi} -\Gamma_0 \|_1 \\
	& \geq - \frac{1}{NT} e'\M_x ( \widehat{\gamma}_{\psi} -\gamma_0) + \frac{\psi}{\sqrt{NT}} \| \widehat{\Gamma}_{\psi} -\Gamma_0 \|_1 \\ 
	& \geq -\frac{1}{ \sqrt{NT}} \| \widehat{\Gamma}_{\psi} -\Gamma_0 \|_1 \frac{1}{\sqrt{NT}} \| {\rm mat}(\M_x e)  \|_\infty + \frac{\psi}{\sqrt{NT}} 
	\| \widehat{\Gamma}_{\psi} -\Gamma_0 \|_1 \\
	& = \left( \psi - \frac{ \| {\rm mat}(\M_x e)  \|_\infty }{\sqrt{NT}} \right) \frac{ \| 
	\widehat{\Gamma}_{\psi} -\Gamma_0 \|_1}{\sqrt{NT}}.
	\end{align*}
	The required result follows since $\psi - \| {\rm mat}(\M_x e)  \|_\infty > 0$.
\end{proof}

\bigskip

\subsection{Sufficient Conditions for Restricted Strong Convexity}

In this section we discuss Assumption~\ref{ass:RSC} in more detail. 
Define the distance $\mathcal{H}(A, \C)$  between a matrix $A \in \mathbb{R}^{N \times T}$ and the cone $\C$
by
\[
\mathcal{H}(A, \C)  :=  \left[ \min_{B \in  \C} \Tr(A - B)'(A - B) \right]^{1/2}. 
\]
The following lemma provides an alternative formulation for our restricted strong convexity assumption.

\begin{lemma}
       \label{lemma:RSC}
	Let there exists a positive constant $\mu>0$ such that
	for any $\alpha  \in \mathbb{R}^K$ with $\alpha'\left(\frac{x'x}{NT}\right)\alpha =1$, the regressors $X_1,...,X_K$ satisfy 
	\[
	\mathcal{H}\left(\alpha \cdot \frac{X}{\sqrt{NT}}, \C \right)^2 \geq \mu > 0, \quad \text{wpa1.}
	\]
        Then 
	Assumption~\ref{ass:RSC} holds.
\end{lemma}

\begin{proof} [\bf Proof of Lemma \ref{lemma:RSC}]~
	Recall the definition $x = [x_1,...,x_K], (NT \times K),$ where $x_k = {\rm vec}(X_k)$.	
	Firstly, if $\theta = 0$, then the required result holds for any constant $\mu > 0$. 
	Secondly, if $\theta'x = 0$, then the required result holds for $\mu = 1$ because
		$ \left( \theta'\theta - \theta'x(x'x)^{-1}x'\theta 
		\right)  =  \theta'\theta$.
        Thus, in the following we only need to consider the case $\theta \neq 0$ and $\theta'x \neq 0$. Also let $x \neq 0$.

		Define $\tilde{x}_{\theta} = \frac{\P_{x} \theta}{\| \P_{x} \theta \|}$, and $\widetilde{X}_{\theta} := {\rm mat}(\tilde{x}_{\theta})$.
		Then, for any $\Theta \in \C$ and $\Theta \neq 0$, we have
		\begin{align}
		& \frac{1}{2NT} \left( \theta'\theta - \theta' x (x'x)^{-1} x'\theta
		\right) \nonumber \\
		& = \frac{1}{2NT} \left( 
		\theta'\theta - \theta' \widetilde{x}_{\theta} \widetilde{x}_{\theta}'\theta
		\right) \quad (\text{by the definition of } \tilde{x}_{\theta}  ) \nonumber \\
		& = \frac{1}{2NT} \|  \Theta \|_2^2  \left( 
		1 - \frac{\theta'\widetilde{x}_{\theta} \widetilde{x}_{\theta}'\theta}{\theta'\theta} 
		\right) \quad (\text{since } \theta \neq 0)  \nonumber \\
		& = \frac{1}{2NT} \|  \Theta \|_2^2  \left( 
		1 - \widetilde{x}_{\theta}' \frac{\theta \theta'}{\theta' \theta} \widetilde{x}_{\theta}  
		\right)  
		= \frac{1}{2NT} \|  \Theta \|_2^2  \left( \widetilde{x}_{\theta}'\widetilde{x}_{\theta} - \widetilde{x}_{\theta}' \frac{\theta \theta' }{\theta'\theta} \widetilde{x}_{\theta} 
		\right)  \nonumber \\
		& = \frac{1}{2NT} \|  \Theta \|_2^2 
		\left( 
		\|  \widetilde{x}_{\theta} - \P_{\theta} \widetilde{x}_{\theta} \|_2^2 
		\right) \nonumber \\
		& \geq \frac{1}{2NT} \|  \Theta \|_2^2 
		\left( 
		\min_{A \in \mathbb{C}}
		\|  \widetilde{x}_{\theta} - \text{vec}(A)  \|^2  
		\right) \nonumber \\
		& = \frac{1}{2NT} \|  \Theta \|_2^2 
		\left( 
		\mathcal{H}(\widetilde{X}_{\theta}, \C)^2 
		\right), \label{eq.lemma.RSC.1}
		\end{align}
		where the inequality holds because ${\rm mat}(\P_{\theta} \widetilde{x}_{\theta}) \in \C $ since $\Theta \in \C$ and $C$ is a cone. 
		Notice that
		\[ 
		\tilde{x}_{\theta} = \frac{\P_{x} \theta}{\| \P_{x} \theta\|_2} = \frac{x}{\sqrt{NT}} \alpha_{*},
		\]
		where  $\alpha_{*} = \frac{ \left(\frac{x'x}{NT}\right)^{-1}\frac{x'}{\sqrt{NT}}\theta}
		{ \left( \theta' \frac{x}{\sqrt{NT}} \left(\frac{x'x}{NT}\right)^{-1}\frac{x'}{\sqrt{NT}}\theta \right)^{1/2}}$ and $\alpha_{*}'\left(\frac{x'x}{NT}\right)\alpha_{*} = 1$. This implies
		\[
		\widetilde{X}_{\theta}  = \alpha_{*} \cdot \frac{X}{\sqrt{NT}}
		\]	
		with $\alpha_{*}'\alpha_{*} = 1$.
		Therefore, we have
		\[
		(\ref{eq.lemma.RSC.1}) \geq \frac{1}{2NT} \|  \Theta \|_2^2  
		\left( \min_{ \alpha'\left(\frac{x'x}{NT}\right)\alpha = 1} 
		\mathcal{H}\left(\alpha \cdot \frac{X}{\sqrt{NT}}, \C \right)^2  \right).
		\]
		Then, the required result of the lemma follows by the assumptions in the lemma.
\end{proof}

\begin{lemma}\label{lemma:RSC.high.rank.sufficient}
     Consider $K=1$.
     	Let $s_1 \geq s_2  \geq s_3 \geq  \ldots \geq s_{\min(N,T)} \geq 0$ be the singular values of the $N \times T$
	matrix $\M_{\lambda_0}  X_1 \M_{f_0}$.  Assume that	there exists a sequence $q_{NT} \geq 2$ such that
	\begin{itemize}
		\item[(i)] $ \frac 1 {\sqrt{NT}} \|  X_1 \|_2 = \Op(1).$
		\item[(ii)] $ \frac 1 {NT}  \sum_{r=q_{NT}}^{\min(N,T)} s_r^2  \geq c > 0 $ wpa1.
		\item[(iii)] $\frac 1 {\sqrt{NT}} \sum_{r=1}^{q_{NT}-2}(s_r - s_{q_{NT}} ) \rightarrow_P  \infty$.
	\end{itemize}
	Then Assumption~\ref{ass:RSC} is satisfied with $\mu=c$.
\end{lemma}

This lemma could be generalized to $K>1$. We would then need to impose the conditions for $X_1$ in the lemma
for all linear combinations $\alpha \cdot X$, in an appropriate uniform sense 
over all $\alpha$
with  $\|\alpha\|=1$.

\begin{proof} [ \bf Proof of Lemma \ref{lemma:RSC.high.rank.sufficient}]~
	For given $N \times T$ matrix $X$, and $N \times R_0$ matrix $\lambda_0$, and $T \times R_0$ matrix $f_0$, we want to find a lower bound
	on
	\begin{align*}
	\nu_{NT} &:= NT \; \; \mathcal{H}\left( \frac{X_1}{\sqrt{NT}}, \C \right)^2
	=  NT \;  \min_{\Theta \in \C}   \left\|  X_{1} / \sqrt{NT} - \Theta \right\|_2^2
	\\
	&=  \min_{\Theta \in \mathbb{R}^{N \times T}}   \left\|  X_{1} - \Theta \right\|_2^2
	\qquad
	\text{s.t.}
	\qquad
	\left\| \M_{\lambda_0} \Theta \M_{f_0} \right\|_1 \leq 3 \left\| \Theta - \M_{\lambda_0} \Theta \M_{f_0} \right\|_1.
	\end{align*}
       By definition, we have
	\begin{align*}
	  \left\| X_{1} - \Theta \right\|_2^2 
	& = 
	\left\|  \M_{\lambda_0} X_{1} \M_{f_0} -\M_{\lambda_0} \Theta \M_{f_0} \right\|_2^2
	+   \left\|  \left(X_{1} - \M_{\lambda_0} X_{1} \M_{f_0} \right) - (\Theta - \M_{\lambda_0} \Theta \M_{f_0} )  \right\|_2^2.
	\end{align*}
	Also, ${\rm rank}( \Theta - \M_{\lambda_0} \Theta \M_{f_0}) \leq 2R_0$ (e.g., see Lemma 3.4 of \cite{RechtFazelParrilo2010}),
	and therefore
	$ \left\| \Theta - \M_{\lambda_0} \Theta \M_{f_0} \right\|_1 \leq \sqrt{2R_0}  \left\| \Theta - \M_{\lambda_0} \Theta \M_{f_0} \right\|_2$. 
	Using this we find
	\begin{align*}
	  \nu_{NT}  
	& \geq  \min_{\Theta \in \mathbb{R}^{N \times T}}   
	\left\{   \left\|  \M_{\lambda_0} X_{1} \M_{f_0} -\M_{\lambda_0} \Theta \M_{f_0} \right\|_2^2
	+   \left\|  \left(X_{1} - \M_{\lambda_0} X_{1} \M_{f_0} \right) - (\Theta - \M_{\lambda_0} \Theta \M_{f_0} )  \right\|_2^2 \right\}
	\\
	&  \qquad \qquad  \text{s.t.}
	\left\| \M_{\lambda_0} \Theta \M_{f_0} \right\|_1 \leq 3 \, \sqrt{2R_0} \,  \left\| \Theta - \M_{\lambda_0} \Theta \M_{f_0} \right\|_2 .
	\end{align*}
	Here, we have weakened the constraint (allowing more values for $\Theta$), and the minimizing value therefore weakly decreases.
	It is easy to see that for $\omega \geq 0$ we have 
	\begin{align*}
	\left(  \left\| X_{1} - \M_{\lambda_0} X_{1} \M_{f_0}    \right\|_2  -  \omega \right)^2
	&= \min_{\Theta \in \mathbb{R}^{N \times T}} 
	\left\|  (X_{1} - \M_{\lambda_0} X_{1} \M_{f_0}) - (\Theta - \M_{\lambda_0} \Theta \M_{f_0} )  \right\|_2^2
	\\&
	\qquad  \qquad  \qquad  \qquad \text{s.t.}
	\qquad    
	\left\| \Theta - \M_{\lambda_0} \Theta \M_{f_0} \right\|_2  = \omega ,
	\end{align*}
	because the optimal $\Theta - \M_{\lambda_0} \Theta \M_{f_0}  $ here equals $X_{1} - \M_{\lambda_0} X_{1} \M_{f_0}$ rescaled by a non-negative number.
	We therefore have
	\begin{align*}
	\nu_{NT} \geq 
	\min_{\omega \geq 0}
	\min_{\Theta \in \mathbb{R}^{N \times T}}  
	&
	\left( \left\|  \M_{\lambda_0} X_{1} \M_{f_0} -\M_{\lambda_0} \Theta \M_{f_0} \right\|_2  \right)^2
	+    \left(   \left\| X_{1} - \M_{\lambda_0} X_{1} \M_{f_0}   \right\|_2  -  \omega \right)^2
	\\
	&   \text{s.t.}
	\left\| \M_{\lambda_0} \Theta \M_{f_0} \right\|_1 \leq 3 \, \sqrt{2R_0} \,  \omega .
	\end{align*}	
	Let 
	$$
	\M_{\lambda_0} X_{1} \M_{f_0}
	 = \sum_{r=1}^{\min(N,T)-R_0}
	s_r \,  v_r w_r',
	$$
	be the singular value decomposition
	of $\M_{\lambda_0} X_{1} \M_{f_0}$  with singular values $s_r \geq 0$ and normalized singular vectors $v_r \in \mathbb{R}^N$ and 
	$w_r \in \mathbb{R}^T$.
	The optimal $\M_{\lambda_0} \Theta \M_{f_0}$ in the last optimization problem has the form
	$$
	\sum_{r=1}^{\min(N,T)-R_0}
	\max(0,  s_r - \xi) \,  v_r w_r' ,
	$$
	for some $\xi \geq 0$ (see Lemma \ref{lemma:MinGamma}). 
	Here, $\xi=0$ occurs if the constraint is not binding, that is, if
	$  \left\| \M_{\lambda_0} X_{1} \M_{f_0} \right\|_1 \leq 3 \, \sqrt{2R_0} \,  \omega$.
	We therefore have
	\begin{align*}
	\nu_{NT} \geq  
	\min_{\omega \geq 0, \, \xi \geq 0}
	&
	\sum_{r=1}^{\min(N,T)-R_0}
	\left(  s_r - \max(0,  s_r - \xi) \right)^2
	+    \left(  \left\| X_{1} - \M_{\lambda_0} X_{1} \M_{f_0}    \right\|_2   -  \omega \right)^2
	\\
	&   \text{s.t.}
	\sum_{r=1}^{\min(N,T)-R_0}
	\max(0,  s_r - \xi) 
	\leq 3 \, \sqrt{2R_0} \,  \omega .
	\end{align*}
	Here, the optimal $\omega$ equals
	$\max\left\{ \left\| X_{1} - \M_{\lambda_0} X_{1} \M_{f_0})    \right\|_2  , \, \frac 1 { 3 \, \sqrt{2R_0}}  \sum_{r=1}^{\min(N,T)-R_0}
	\max(0,  s_r - \xi)  \right\}$,
	and we thus have
	\begin{align*}
	\nu_{NT} \geq  
	&\min_{\xi \geq 0} 
	\sum_{r=1}^{\min(N,T)-R_0} 
	\Bigg[  \min( s_r^2  ,  \xi^2) \\
	& +       
	\left(  \max\left\{0 ,     \frac 1 { 3 \, \sqrt{2R_0}}  \left( \sum_{r=1}^{\min(N,T)-R_0}
	\max(0,  s_r - \xi) \right) - \left\| X_{1} - \M_{\lambda_0} X_{1} \M_{f_0}    \right\|_2     \right\}  \right)^2 \Bigg].
	\end{align*}       
        Let $\infty = s_0 > s_1 \geq  \ldots  \geq s_{\min(N,T)-R_0} \geq s_{\min(N,T)-R_0+1} = 0$.
	For any $\xi \geq 0$ there exists
	$q$ such that $\xi \in [ s_{q+1} , s_{q} ]$. We can therefore write
	\begin{align*}
	\nu_{NT} 
	  &\geq 
	 \min_{q \in \{0,1,2,\ldots,\min(N,T)-R_0\} }
	\; \min_{\xi \in [  s_{q+1} , s_{q} ]}
	\Bigg[ q \, \xi^2
	+ \sum_{r=q+1}^{\min(N,T)-R_0}  s_r^2
	  \\  & \quad
	+     
	\left(     \max\left\{0 ,   \frac 1 { 3 \, \sqrt{2R_0}}  \left( \sum_{r=1}^{q}  ( s_r - \xi) \right) \mathbbm{1}\{ q \geq 1 \}  - \left\| X_{1} - \M_{\lambda_0} X_{1} \M_{f_0})    \right\|_2    \right\}
	\right)^2 \Bigg]
   \\
        &\geq 
	 \min_{q \in \{0,1,2,\ldots,\min(N,T)-R_0\} }
	\Bigg[  \left( \min_{\xi \in [  s_{q+1} , s_{q} ]}   q \, \xi^2 \right)
	+ \sum_{r=q+1}^{\min(N,T)-R_0}  s_r^2
	  \\  & \quad
	+     
	\left(     \max\left\{0 ,   \frac 1 { 3 \, \sqrt{2R_0}}
	  \left( \min_{\xi \in [  s_{q+1} , s_{q} ]}   \sum_{r=1}^{q}  ( s_r - \xi) \right) \mathbbm{1}\{ q \geq 1 \}  - \left\| X_{1} - \M_{\lambda_0} X_{1} \M_{f_0})    \right\|_2    \right\}
	\right)^2 \Bigg]
   \\
        &= 
	 \min_{q \in \{0,1,2,\ldots,\min(N,T)-R_0\} }
	\Bigg[    q \, s_{q+1}^2 
	+ \sum_{r=q+1}^{\min(N,T)-R_0}  s_r^2
	  \\  & \quad
	+     
	\left(     \max\left\{0 ,   \frac 1 { 3 \, \sqrt{2R_0}}
	  \left(     \sum_{r=1}^{q-1}  ( s_r - s_{q} ) \right) \mathbbm{1}\{ q \geq 2 \}  - \left\| X_{1} - \M_{\lambda_0} X_{1} \M_{f_0})    \right\|_2    \right\}
	\right)^2 \Bigg] .
	\end{align*}   
       Shifting  $q \mapsto q-1$
	we can rewrite this as
	\begin{align*}
	\frac{ \nu_{NT} } {NT} \geq 
	\min_{q \in \{1,2,\ldots,\min(N,T) - R_0 \} }
	\bigg( a(q) + \left[ \max \left\{ 0,   b(q)  \right\} \right]^2 \bigg) ,
	\end{align*}   
	where 
	\begin{align*}
	a(q) &= \frac 1 {NT} \left[ (q-1) \, s_{q}^2 + \sum_{r=q}^{\min(N,T)} s_r^2 \right] ,
	\\
	b(q) &= \frac 1 {\sqrt{NT}} \left[ \frac 1 {3 \sqrt{2 R_0}} \left( \sum_{r=1}^{q-2}(s_r - s_{q}) \right) \mathbbm{1}\{ q \geq 3\}   
	-  \left\| X_{1} - \M_{\lambda_0} X_{1} \M_{f_0}  \right\|_2 
	\right] .
	\end{align*}
	Notice that $a(q)$ is nonnegative and weakly decreasing and $b(q)$ is weakly increasing. Then, for any integer valued sequence $q_{NT} $ between 1 and $\min(N,T) -R_0$ such that $b(q_{NT}) > 0$, 
	\begin{align*}
	& \min_{q \in \{1,2,\ldots,\min(N,T) -R_0 \} }
	\bigg( a(q) + \left[ \max \left\{ 0,   b(q)  \right\} \right]^2 \bigg) \\
	& = \min \left\{ 
	\min_{q \in \{1,2,\ldots, q_{NT} \} }
	\bigg( a(q) + \left[ \max \left\{ 0,   b(q)  \right\} \right]^2 \bigg) 
	, \min_{q \in \{q_{NT} + 1,\ldots,\min(N,T) - R_0 \} } 
	\bigg( a(q) + \left[ \max \left\{ 0,   b(q)  \right\} \right]^2 \bigg) \right\} \\
	& \geq \min \left\{ 
	\min_{q \in \{1,2,\ldots, q_{NT} \} } a(q),  
	\min_{q \in \{q_{NT} + 1,\ldots,\min(N,T) - R_0 \} } 
	\left[ \max \left\{ 0,   b(q)  \right\} \right]^2  \right\} \\
	& \geq \min \left\{ a(q_{NT}),  b(q_{NT}+1)^2 \right\}.
	\end{align*}
	The assumptions of the lemma thus guarantee that
	$\nu_{NT} /(NT)  \geq c$. The definition of $\nu_{NT}$ together with Lemma~\ref{lemma:RSC}
	thus guarantees that  Assumption~\ref{ass:RSC}
	is satisfied with $\mu = c$.
\end{proof}

\noindent {\bf Remarks}
\begin{itemize}
	\item[(a)]
	When $X$ is a``high-rank" regressor and $s_q$'s are of an order $\Op(\sqrt{\max(N,T)})$, we can choose, for example, $q_{NT} = \lfloor \min(N,T) / 2 \rfloor$, for $N,T$ converging to infinity at the same rate, where $\lfloor a \rfloor$ is the
	integer part of $a$.
	Then, it is easy to verify those sufficient conditions (i), (ii) and (iii)
	for e.g. $X_{it} \sim i.i.d. {\cal N}(0,\sigma^2)$ from well-known random matrix theory results.
	More generally, we can explicitly verify (i), (ii) and (iii) if $X$ has an approximate factor structure
	\[
	X =  \lambda_x f_x' + E_{x} ,
	\] 
        where $ \lambda_x f_x' $ is an arbitrary low-rank factor structure, 
        and $E_{x} \sim i.i.d. {\cal N}(0,\sigma^2)$.

	\item[(b)] 	
	For a low-rank regressor with ${\rm rank}(X) = 1$, we have singular values $s_1 = \|\M_{\lambda_0} X \M_{f_0}\|_2$ and
	$s_r = 0$ for all $r \geq 2$. In that case we find that $a(1) = \frac 1 {NT} s_1^2$
	and  $a(q)=0$ for $q > 1$,
	and we have
	$b(1) = b(2) = 0$ 
	and $b(q)=b(3)= \frac 1 {\sqrt{NT}} \left[ \frac 1 {3\sqrt{2R_0}} s_1 - \| X - \M_{\lambda_0} X \M_{f_0} \|_2 \right]$ for all $q \geq 3$. Also, $a(1) \geq b(2)$. Therefore
	\begin{align*}
	\min_{q \in \{1,2,\ldots,\min(N,T) \}}
	\bigg[ a(q) + \left( \max\left\{ 0,   b(q)  \right\} \right)^2  \bigg] 
	= \min \left\{ a(1), \left( \max\{ 0 , b(3) \} \right)^2 \right\}   .
	\end{align*}   
	Thus, the assumptions of Lemma~\ref{lemma:RSC.high.rank.sufficient} are satisfied if wpa1 we have
	\begin{align*}
	\frac 1 {\sqrt{NT}}  \left[   \|\M_{\lambda_0} X \M_{f_0}\|_2
	-  3 \sqrt{2 R_0}  \left\| X - \M_{\lambda_0} X \M_{f_0}  \right\|_2  \right] \geq c_1 > 0
	\end{align*}
	for some constant $c_1$. 
	This last condition simply demands that the part of $X$ that cannot be explained by $\lambda_0$
	and $f_0$ needs to be sufficiently larger than the part of $X$ that can be explained by either $\lambda_0$ or $f_0$.
	This is a sufficient condition for Assumption~\ref{ass:RSC}. An analysis that is specialized towards low-rank regressors 
	will likely give a weaker condition for Assumption~\ref{ass:RSC} in this case.
	
\end{itemize}

\subsection{Proofs for Section \ref{subsec:nuc.minimization}}\label{sec:proofs-of-section-refsecnucminimization}

\begin{proof}[\bf Proof of Theorem~\ref{th:ConsistencyNormMinEst} ]
	Remember the following singular value decompositions: $\Gamma_0 = U S V'$, 
	$\M_{\lambda_0} E \M_{f_0} = \M_U E \M_V  = U_E S_E V_E'$,
	 and $\M_{\lambda_0} X \M_{f_0} = \M_U X \M_V = U_x S_x V_x'$.
	The proof consists of two steps. In the first step, we show that the local minimizer that minimizes the objective function $Q_{*}(\beta)$ in a convex neighbourhood of $\beta_0$ defined by 
	\[
	\mathbf{B} := \left\{ \beta : \frac{ c_x \,  c_{\rm low} } {c_{\rm up} } \,   | \Delta \beta |  \leq 1 \right\}
	\] 
	is $\sqrt{T}$- consistent.	
	In the second step, we show that the local minimizer is the global minimizer, for which we use convexity of the objective function $Q_{*}(\beta)$.
	
	\noindent {\bf Step 1.}
	By definition of the nuclear norm, we have
	\begin{align*}
	Q_{*}(\beta) &=
	\left\|    \Gamma_0 + E  - \Delta \beta \cdot X \right\|_1
	= \sup_{ \left\{ A \, : \, \| A \|_\infty \leq 1 \right\}}   
	{\rm Tr}\left[  \left(  \Gamma_0 + E  - \Delta \beta \cdot X \right)' A \right] .
	\end{align*}
	To obtain a lower bound on  $Q_{*}(\beta)$ we choose the following matrix $A$
	in the above minimization,
	\begin{align*}
	A_\beta &= U V' +  \sqrt{1-a_\beta^2} \; U_E V_E' 
	-  a_\beta  \,  \left( {\rm sgn} \, \Delta \beta\right) \, \M_{U_E}  U_x V_x' ,
	\end{align*}
	where 
	$\M_{U_E} = I_N - U_E U_E'$ and
	$a_\beta \in [0,1]$ is given by 
	\begin{align*}
	a_\beta &= \frac{ c_x \,  c_{\rm low} } {c_{\rm up} } \,  | \Delta \beta |.
	\end{align*}
	
	We have $ \|A_\beta\|_\infty \leq 1$, because
	\begin{align*}
	\|A_\beta\|_\infty^2
	&=  \max \left\{  \left\| U V'  \right\|_{\infty}^2 , 
	\left\| \sqrt{1-a_\beta^2} \; U_E V_E' -  a_\beta \,  \left( {\rm sgn} \, \Delta \beta\right)   \, \M_{U_E}  U_x V_x' \right\|_{\infty}^2  \right\}
	\\
	&\leq    \max \left\{   \left\| U V'  \right\|_{\infty}^2 ,  
	(1-a_\beta^2) \; \left\| U_E V_E'  \right\|_{\infty}^2
	+ a_\beta^2  \left\| \M_{U_E}  U_x V_x' \right\|_{\infty}^2 
	\right\}
	\\
	&=1.
	\end{align*}
	Here, for the first line, we used that $U V' $ is orthogonal to 
	$ \sqrt{1-a_\beta^2} \; U_E V_E' -  a_\beta \left( {\rm sgn} \, \Delta \beta \right)  \, \M_{U_E}  U_x V_x'$
	in both matrix dimensions (that is, $U'U_E, U'U_k, V'V_E, V'V_x = 0)$ and applied Lemma \ref{lemma:ap.useful.facts}(iv).  
	For the second line, we used that the columns of $U_E V_E' $ are orthogonal
	to the columns of $ \M_{U_E}  U_x V_x'  $ since $U_E'\M_{U_E} = 0$, and applied Lemma \ref{lemma:ap.useful.facts}(v).
	In the final line we used that 
	$ \left\| U V'  \right\|_{\infty} = \left\| U_E V_E'  \right\|_{\infty} =1$
	and  that $ \left\| \M_{U_E}  U_x V_x' \right\|_{\infty} \leq 1$.
	
	With this choice of $A=A_\beta$ we obtain the following lower bound for the objective function; for all $\beta \in \mathbf{B}$,
	\begin{align*}
	Q_{*}(\beta) 
	&\geq
	{\rm Tr}\left[  \left(  \Gamma_0 + E  - \Delta \beta \cdot X \right)' A_\beta \right]
	\\
	&=  \| \Gamma_0 \|_1 +   {\rm Tr}\left( E'  U V' \right)   +      {\rm Tr}\left[  \left(     - \Delta \beta \cdot X \right)' U V' \right] 
	\\ & \quad
	+  \sqrt{1-a_\beta^2}  \; \| \M_U E \M_V  \|_1
	+   \sqrt{1-a_\beta^2}  \; {\rm Tr}\left[  \left(     - \Delta \beta \cdot X \right)' U_E V_E' \right]  
	\\
	& \quad          
	+ a_\beta \left| \Delta \beta \right|   {\rm Tr}\left[  X' \M_{U_E}  U_x V_x' \right] ,
	\end{align*}
	where we used the following:
	\begin{align*}
	{\rm Tr}\left( \Gamma_0'  U V' \right) &= {\rm Tr}\left( V S U'  U V' \right) = \Tr(S) = \| \Gamma_0 \|_1,
	\\
	{\rm Tr}\left(   E' U_E V_E' \right) &= {\rm Tr}\left(  (E - \M_UE \M_V + \M_U E \M_V)' U_E V_E' \right)\\
	&= \Tr((\M_U E \M_V)'U_E V_E')) = \Tr(S_E) = \| \M_U E \M_V  \|_1,
	\\
	{\rm Tr}\left( \Gamma_0'  U_E V_E' \right) &= {\rm Tr}\left( VSU'  U_E V_E' \right) = 0,
	\\
	{\rm Tr}\left[  \Gamma_0'  \M_{U_E}  U_x V_x' \right] 
	&= {\rm Tr} \left[  V S U  \M_{U_E}  U_x V_x' \right] = 0, 
	\\
	{\rm Tr}\left[  E'  \M_{U_E}  U_x V_x' \right] &= {\rm Tr}\left[  \M_V E' \M_U \M_{U_E}  U_x V_x' \right] +  {\rm Tr}\left[ (E' - \M_V E' \M_U) \M_{U_E}  U_x V_x' \right] = 0.
	\end{align*}
	We furthermore have $Q(\beta_0) =  \left\| \Gamma_0 + E \right\|_1$. Thus, applying the assumptions of the theorem 
	 and also using
	$\sqrt{1-a_\beta^2} \geq 1 - \frac 1 2 a_\beta^2 - \frac 1 2 a_\beta^4$, we obtain for $\beta \in \mathbf{B}$,
	\begin{align}
	Q_{*}(\beta) - Q_{*}(\beta_0) 
	&\geq
	{\rm Tr}\left[  \left(  \Gamma_0 + E  - \Delta \beta \cdot X \right)' A_\beta \right]
	- \left\| \Gamma_0 + E \right\|_1 \nonumber 
	\\
	&\geq 
	a_\beta  \left| \Delta \beta \right|   {\rm Tr}\left[  X' \M_{U_E}  U_x V_x' \right] \nonumber
	\\
	& \quad -  \frac 1 2 a_\beta^2 \,   \| \M_U E \M_V  \|_1  
	-  \left( \left\| \Gamma_0 + E \right\|_1 - \| \Gamma_0 \|_1 - \| \M_U E \M_V  \|_1 \right)  \nonumber 
	\\
	& \quad + {\rm Tr}\left( E'  U V' \right)  \nonumber
	- \frac 1 2 a_\beta^4   \| \M_U E \M_V  \|_1 \nonumber
	\\
	& \quad + \sqrt{1-a_\beta^2}  \; {\rm Tr}\left[  \left(     - \Delta \beta \cdot X \right)' U_E V_E' \right]  
	+       {\rm Tr}\left[  \left(     - \Delta \beta \cdot X \right)' U V' \right] \nonumber \\ 
	& =: B_1 - B_2 -   B_3 + B_4 - B_5 + B_6. \label{eq.ap.lowbound.Q*}
	\end{align}
	Here we bound $B_1$ from below by
	\begin{align*}
	B_1 &=  a_\beta \left| \Delta \beta \right|   {\rm Tr}\left(  X' \M_{U_E}  U_x V_x' \right) 
	\\
	&=   a_\beta \left| \Delta \beta \right|  
	\left[ {\rm Tr}\left( \M_V X' \M_U \M_{U_E}  U_x V_x' \right) - {\rm Tr}\left( (X' - \M_V X' \M_U ) \M_{U_E}  U_x V_x' \right) \right]
	\\
	&=  a_\beta \left| \Delta \beta \right|   
	 {\rm Tr}\left( V_x S_x U_x' \M_{U_E}  U_x V_x' \right)  \\
	&=   a_\beta \left| \Delta \beta \right|  \, \left| \Delta \beta \right|  \, \left[  {\rm Tr}\left(    S_x      \right) 
	- {\rm Tr}\left(  U_E' U_x   S_x  U_x'   U_E  \right)  \right]
	\\
	&\geq     a_\beta \, c_x  \, \left| \Delta \beta \right|  \,    {\rm Tr}\left(    S_x      \right)   
	=     a_\beta \, c_x   \, \left| \Delta \beta \right|  \,  
	\|   \M_{U} X  \M_{V}  \|_1  
	\\
	&\geq  a_\beta \, c_x \, c_{\rm low} \, T \, \sqrt{N}  \, \left| \Delta \beta \right|.
	\end{align*}
	Here the first inequality holds by assumption (vi),
	and the second inequality holds by assumption (v). 
	
	We bound $B_2$ from above by
	\begin{align*}
	B_2 &=  \frac 1 2 a_\beta^2 \,   \| \M_U E \M_V  \|_1
	\\
	&\leq \frac 1 2 a_\beta^2 \, \left( 
	\| E \|_1 + \| \P_{U} E \|_1 + \| E \P_{V} \|_1 + \| \P_U E \P_V \|_1 
	\right)
	\\
	&\leq \frac 1 2 a_\beta^2 \, \left( 
	\| E \|_1 + 3 R_0 \| E \|_{\infty} 
	\right) 
	\\
	&\leq  \frac 1 2 a_\beta^2 \,  T \, \sqrt{N} \, \left( \frac{c_{\rm up}}{2} + \frac{1}{T} \Op(1) \right) \quad {\rm wpa1}
	\\
	&\leq \frac 1 2 a_\beta^2 \,  T \, \sqrt{N} \, c_{\rm up} \quad {\rm wpa1},
	\end{align*}
	where the first inequality holds by the triangle inequality, the second inequality holds by Lemma~\ref{lemma:ap.useful.facts}(i) and the third and the fourth inequalities follow by assumption (i) and (ii).

	We bound term $B_3$ from above by 
	\begin{align*}
	B_3 &=  \left\| \Gamma_0 + E \right\|_1 - \| \Gamma_0 \|_1 - \| \M_U E \M_V  \|_1  
	\\ 
	&\leq     \left\|  E - \M_U E \M_V \right\|_1  = \| \P_U E + E \P_V - \P_U E \P_V \|_1  	
	\\
	&\leq \| \P_U E \|_1 + \| E \P_V \|_1 + \| \P_U E \P_V \|_1 
	\\
	&\leq 3R_0 \| E \|_{\infty} 
	\\
	&\leq \Op(\sqrt{N})
	\end{align*}
	where the second inequality holds by the triangle inequality and the third inequality holds by Lemma~\ref{lemma:ap.useful.facts}(i).

	For $B_4$, by H\"{o}lder's inequality we have
	\begin{align*}
	B_4 &=  \| {\rm Tr}\left( E'  U V' \right) \| \leq \| E \|_{\infty} \| UV'\|_1 =    \Op(\sqrt{N}).
	\end{align*}
	For $B_5$, denoting $O_{P+}(\cdot)$ as a stochastically strictly positive and bounded term and using similar arguments for the bound of term $B_2$, we obtain
	\begin{align*}
	B_5 &=  \frac 1 2 a_\beta^4   \| \M_U E \M_V  \|_1   =  O_{P+}(1) a_\beta^4 \,   T \, \sqrt{N} 
	=  O_{P+}(1)   (\Delta \beta )^4  \,   T \, \sqrt{N}.
	\end{align*}
	For $B_6$, we have
	\begin{align*}
	B_6 &=  \sqrt{1-a_\beta^2}  \; {\rm Tr}\left[  \left(     - \Delta \beta \cdot X \right)' U_E V_E' \right]  
	+       {\rm Tr}\left[  \left(     - \Delta \beta \cdot X \right)' U V' \right] 
	=    \Op \left( \sqrt{NT} | \Delta \beta | \right),   
	\end{align*}
	where the last equality holds since ${\rm Tr}(X_k U_E V_E') = \Op(\sqrt{NT})$ by assumption  (vi), and ${\rm Tr}(X_k U V') \leq \| X \|_{\infty} \| U V' \|_1 = \Op(\sqrt{NT})$ under assumption (iii).
	
	Notice that our choice for $a_\beta$ above is such that 
	$  a_\beta \, c_x \, c_{\rm low} |    \Delta \beta   |  -  \frac 1 2 a_\beta^2  \, c_{\rm up}$
	is maximized, which guarantees that $B_1 - B_2$ is positive, namely
	\begin{align*}
	\frac{B_1 - B_2} {T \sqrt{N}} \geq  \frac{c_x^2 c_{\rm low}^2} {2 \, c_{\rm up} }  | \Delta \beta |^2 .
	\end{align*}
	Combining the above, for any $\beta \in \mathbf{B}$, we have 
	\begin{align*}
	\frac 1 {T \, \sqrt{N}} 
	\left\{  Q_{*}(\beta) - Q_{*}(\beta_0)    \right\}              
	&\geq    \frac{c_x^2 c_{\rm low}^2} {2 \, c_{\rm up}}  | \Delta \beta |^2
	+    \Op \left(  \frac{ 1 } {\sqrt{T}}  |    \Delta \beta   | \right)
	+ \Op( T^{-1} )         
	+ O_{P+}(1) | \Delta\beta |^4,               
	\end{align*}
	which holds uniformly over $\beta \in \mathbf{B}$ (i.e. none of the 
	constants hidden in the $\Op(.)$ notation depend on $\beta$).
	
	Let 
	\[ 
	\widetilde{\beta}_{*} := \argmin_{\beta \in \mathbf{B}} Q_{*}(\beta)
	\]
	be the local minimizer in a convex neighbourhood $\mathbf{B}$ of $\beta_0$.
	Notice that since $\beta_0 \in \mathbf{B}$, $Q_{*}(\widehat{\beta}_*) \leq Q_{*}(\beta_0)$ by definition. Therefore, we have
	\begin{align*}
	0 
	& \geq  \frac{1}{T \sqrt{N}}(Q_{*}(\widetilde{\beta}_*) - Q_{*}(\beta_0)) \\
	& \geq \frac{c_x^2 c_{\rm low}^2} {2 \, c_{\rm up}} |\widetilde{\beta}_* - \beta_0 |^2 
	+ \Op \left( \frac{1}{\sqrt{T}} | \widetilde{\beta}_* - \beta_0  | \right)
	+ \Op\left( \frac{1}{T} \right) 
	+ O_{P+}\left( | \widetilde{\beta}_* - \beta_0  |^4 \right). 
	\end{align*}
	This implies  
	\begin{align*}
	O_{P+} \left( \frac{1}{T} \right)  
	& \geq \left( \frac{c_x^2 c_{\rm low}^2} {2 \, c_{\rm up}}
	+ O_{P+}(1) |\widetilde{\beta}_* - \beta_0 |^2  \right) | \widetilde{\beta}_* - \beta_0 |^2 
	+ \Op \left( \frac{1}{\sqrt{T}}  \right) | \widetilde{\beta}_* - \beta_0 | \\
	& \geq \frac{c_x^2 c_{\rm low}^2} {2 \, c_{\rm up}} | \widetilde{\beta}_* - \beta_0 |^2 
	+ \Op \left( \frac{1}{\sqrt{T}}  \right) | \widetilde{\beta}_* - \beta_0 |.
	\end{align*}
	From this we deduce 
	\begin{equation}
		| \widetilde{\beta}_* - \beta_0 | = \Op \left( \frac{1}{\sqrt{T}} \right). \label{eq.ap.sqrtT.consistency.local}
	\end{equation}  
	
	\noindent {\bf Step 2.}
	Let $\bar{\beta} \in \partial \mathbf{B}$, that is, $\alpha_{\bar{\beta}} = 1$. 
	Write $\Delta \bar{\beta} := \bar{\beta} - \beta_0$. 
	From (\ref{eq.ap.lowbound.Q*}) with $a_{\bar{\beta}} = 1$, we can bound $Q_{*}(\bar{\beta}) - Q_{*}(\beta_0)$ from below by 
	\begin{align*}
	& \frac{1}{T \, \sqrt{N}} ( Q_{*}(\bar{\beta}) - Q_{*}(\beta_0) ) 
	\\
	&\geq  c_x \, c_{\rm low} \, |    \Delta \bar{\beta}  | 
	- \frac{1}{2} \, c_{\rm up}  
	+  \Op \left(  \frac{ 1 } {\sqrt{T}}  |    \Delta \bar{\beta}   | \right)
	+ \Op \left( \frac{1}{T} \right) + O_{P+}(1) | \Delta \bar{\beta} |^4 
	\\
	&= \frac{1}{2} \, c_{\rm up} + \Op \left( \frac{1}{T} \right) + \Op \left(  \frac{ 1 } {\sqrt{T}} \right) \frac {c_{\rm up} } { c_x \,  c_{\rm low} } 
	+ O_{P+}(1) \left( \frac {c_{\rm up} } { c_x \,  c_{\rm low} } \right)^4 
	\\
	& > 0 \quad {\rm wpa1},
	\end{align*}
	where the equality holds since $ |    \Delta \bar{\beta}   | = \frac {c_{\rm up} } { c_x \,  c_{\rm low} }. $
	
	Since $Q_*(\beta)$ is convex and has unique minimum, the local minimum at $\widetilde{\beta}_*$ is also the global minimum asymptotically. Therefore,  asymptotically
	\[
	\widetilde{\beta}_* = \widehat{\beta}_* \quad {\rm wpa1}.
	\]
	Combining this with the $\sqrt{T}-$ consistency result of the local minimizer in (\ref{eq.ap.sqrtT.consistency.local}) gives the statement of the theorem.	
\end{proof}

\subsubsection{Extension of Theorem~\ref{th:ConsistencyNormMinEst}}

Theorem~\ref{th:ConsistencyNormMinEst} is the special case of one regressor 
($K=1$). We can extend this to a more general case with $K$ regressors. The 
proof of the following general theorem is similar to that of Theorem 
\ref{th:ConsistencyNormMinEst}, and we skip it.

\begin{theorem}[\bf Generalization of Theorem~\ref{th:ConsistencyNormMinEst} to multiple regressors]
	\label{th:ConsistencyNormMinEst2}
	Let there exist symmetric idempotent $T \times T$ matrices $Q_k = Q_{k,NT}$
	such that $Q_k V = 0$, for all $k \in \{1,\ldots,K\}$,
	and           $Q_k Q_\ell =0$, for all $k,\ell \in \{1,\ldots,K\}$.
	Suppose that $N>T$. As $N,T \rightarrow \infty$, we assume the following conditions hold.
	\begin{itemize} 
		\item[(i)]   $\| E \|_{\infty} = \Op(\sqrt{N})$.
		
		\item[(ii)] There exists a finite positive constant $c_{\rm up}$ such that 
		$\frac 1 {T \, \sqrt{N} } \| E \|_1 \leq \frac 1 2  c_{\rm up}$, wpa1.
		
		\item[(iii)] $ \| X_k \|_{\infty} = \Op(\sqrt{NT})$, for  $k\in \{1,\ldots,K\}$.
		
		\item[(vi)]  Let $U_E S_E V_E' $ be the singular value decomposition of $\M_{\lambda_0} E \M_{f_0}$. We assume ${\rm Tr}\left(  X_k' U_E V_E' \right)   = \Op ( \sqrt{NT} )$ for all  $k\in \{1,\ldots,K\}$.
		\item[(v)] We assume that there exists a constant $c_{\rm low}>0$ such that wpa1
		$$T^{-1} N^{-1/2} \| \M_{U} X_k  \M_{V}  Q_k  \|_1  \geq c_{\rm low}  ,
		$$
		for all $k \in \{1,\ldots,K\}$.
		\item[(vi)]  For $k=1,\ldots,K$ let $U_k S_k V_k' = \M_U X_k \M_V  Q_k (= \M_U X_k Q_k) $
		be the singular value decomposition of the matrix $\M_U X_k  \M_V Q_k$. We assume that 
		there exists $ c_x  \in (0,1)$ such that wpa1 
		$\| U_k'U_E\|_{\infty}^2 \leq (1-c_x)$ for all $k=1,\ldots,K$.         
	\end{itemize}
	We then have $\sqrt{T} \left(  \widehat{\beta}_{*} - \beta_0 \right) = \Op(1)$.
\end{theorem}

\noindent {\bf Remark} For $t \in \{1,2,\ldots,T\}$,  let $\mathfrak{e}_t$ be the $t$'th unit vector of dimension $T$. 
For $k \in \{1,\ldots,K\}$,
let $A_k =  ( \mathfrak{e}_{\lfloor (k-1) T/K \rfloor + 1},  \mathfrak{e}_{\lfloor (k-1) T/K \rfloor + 2}, \ldots \mathfrak{e}_{\lfloor k T/K \rfloor} )$
be a $T \times \lfloor  T/K \rfloor$ matrix,
and let $\mathbf{P}_{A_k}$ be the projector onto the column space of $A_k$.
Also define $f_{0,k} = \mathbf{P}_{A_k} f_0$
and $B_k = \M_{f_{0,k}} A_k$. 
Then, for $K>1$ one possible choice for $Q_k$ in assumption (vi) of Theorem~\ref{th:ConsistencyNormMinEst2} is given by
$$
Q_k = \mathbf{P}_{B_k}  = \M_{f_{0,k}} \mathbf{P}_{A_k} .
$$
The discussion of assumption (vi) of  Theorem~\ref{th:ConsistencyNormMinEst2} is then analogous to the $K=1$ case, except that
for the $k$'th regressor only the time periods $\lfloor (k-1) T/K \rfloor + 1$ to $\lfloor k T/K \rfloor$ are used
in the assumption, that is, we need enough variation in the $k$'th regressor within those time periods.
Other choices of $Q_k$ are also conceivable.

\subsection{Proofs for Section \ref{sec:post estimation}}

For $\beta \in \mathbb{R}^K$ we define 
\[
\left\{ \widehat{\lambda}(\beta),   \widehat{f}(\beta) \right\} := \argmin_{\lambda \in \R^{N \times R_0}, f \in \R^{T \times R_0}} \| Y - \beta \cdot X \|_2^2 \, ,
\] 
and the corresponding projection matrices
\begin{equation*}
\M_{\widehat \lambda}(\beta) := \mathbb{I}_N -  \widehat \lambda(\beta) \left( \widehat \lambda(\beta)'\widehat \lambda(\beta) \right)^{-1} \widehat \lambda(\beta)',
\qquad \M_{\widehat f}(\beta) :=  \mathbb{I}_T -  \widehat f(\beta) \left( \widehat f(\beta)'\widehat f(\beta) \right)^{-1} \widehat f(\beta)'.
\end{equation*}

\begin{lemma}
 \label{lemma:expansion.operators}
    Under the assumptions (i) and (ii) of Theorem~\ref{thm:post.matrix.lasso} we have
	\begin{align*}
	\M_{\widehat \lambda}(\beta) &= \M_{\lambda_0} + \M_{\widehat \lambda,E}^{(1)}
	+ \M_{\widehat \lambda,E}^{(2)}
	- \sum_{k=1}^K \left( \beta_k - \beta_{0,k} \right)  \M_{\widehat \lambda,k}^{(1)}
	+ \M_{\widehat \lambda,E}^{({\rm rem})}  + \M_{\widehat \lambda}^{({\rm rem})}(\beta) \; ,
	\nonumber \\
	\M_{\widehat f}(\beta) &= \M_{f_0} + \M_{\widehat f,E}^{(1)}
	+ \M_{\widehat f,E}^{(2)}
	- \sum_{k=1}^K \left( \beta_k - \beta_{0,k} \right)  \M_{\widehat f,k}^{(1)}
	+ \M_{\widehat f,E}^{({\rm rem})} +  \M_{\widehat f}^{({\rm rem})}(\beta) \; ,
	\end{align*}
	where the spectral norms of the remainders satisfy for any series $r_{NT} \rightarrow 0$,
	\begin{align*}
		\sup_{\beta \in {\cal B}(\beta_0, r_{NT})}
		\frac{\left\| \M_{\widehat \lambda}^{({\rm rem})}(\beta) \right\|_{\infty}}
		{ \|\beta - \beta_0\|^2 + (NT)^{-1/2}  \|E\|_{\infty}  \|\beta - \beta_0\|  }
		&= \Op\left(1\right)   ,
		&
		\sup_{\beta \in {\cal B}(\beta_0, r_{NT})}
		\frac{\left\| \M_{\widehat \lambda,E}^{({\rm rem})}  \right\|_{\infty}}
		{   (NT)^{-3/2}  \|E\|_{\infty}^3}
		&= \Op\left(1\right) ,
	 \\
		\sup_{\beta \in {\cal B}(\beta_0, r_{NT})}
		\frac{\left\| \M_{\widehat f}^{({\rm rem})}(\beta) \right\|_\infty}
		{ \|\beta - \beta_0\|^2 + (NT)^{-1/2}  \|E\|_\infty  \|\beta - \beta_0\|  }
		&= \Op\left(1\right)   ,
		&
		\sup_{\beta \in {\cal B}(\beta_0, r_{NT})}
		\frac{\left\| \M_{\widehat f,E}^{({\rm rem})}  \right\|_{\infty}}
		{   (NT)^{-3/2}  \|E\|_{\infty}^3}
		&= \Op\left(1\right) ,	
	\end{align*}
	and the expansion coefficients are given by
	\begin{align*}
	\M^{(1)}_{\widehat \lambda,E} &= - \, \M_{\lambda_0} \, E \,
	f_0 \, (f_{0}^{\prime}f_0)^{-1} \, (\lambda_{0}^{\prime}\lambda_0)^{-1} \lambda_{0}^{\prime}
	\, - \, \lambda_0 \, (\lambda_{0}^{\prime}\lambda_0)^{-1} \, (f_{0}^{\prime}f_0)^{-1} \, f_{0}^{\prime}
	\, E' \, \M_{\lambda_0}  \; ,
	\nonumber \\
	\M^{(1)}_{\widehat \lambda,k} &= - \, \M_{\lambda_0} \, X_k  \,
	f_0 \, (f_{0}^{\prime}f_0)^{-1} \, (\lambda_{0}^{\prime}\lambda_0)^{-1} \lambda_{0}^{\prime}
	\, - \, \lambda_0 \, (\lambda_{0}^{\prime}\lambda_0)^{-1} \, (f_{0}^{\prime}f_0)^{-1} \, f_{0}^{\prime}
	\,  X'_k  \, \M_{\lambda_0}  \; ,
	\nonumber \\
	\M^{(2)}_{\widehat \lambda,E} &=
	\M_{\lambda_0} \, E \, f_0 \, (f_{0}^{\prime}f_0)^{-1} \, (\lambda_{0}^{\prime}\lambda_0)^{-1} \lambda_{0}^{\prime}
	\, E \, f_0 \, (f_{0}^{\prime}f_0)^{-1} \, (\lambda_{0}^{\prime}\lambda_0)^{-1} \lambda_{0}^{\prime}
	\nonumber \\ & \qquad
	+\lambda_0 \, (\lambda_{0}^{\prime}\lambda_0)^{-1} \, (f_{0}^{\prime}f_0)^{-1} \, f_{0}^{\prime}
	\, E' \, \lambda_0 \, (\lambda_{0}^{\prime}\lambda_0)^{-1} \, (f_{0}^{\prime}f_0)^{-1} \, f_{0}^{\prime}
	\, E' \, \M_{\lambda_0}
	\nonumber \\ & \qquad
	- \M_{\lambda_0} \, E \, \M_{f_0} \, E' \,
	\lambda_0\,(\lambda_{0}^{\prime}\lambda_0)^{-1}\,(f_{0}^{\prime}f_0)^{-1}\,(\lambda_{0}^{\prime}\lambda_0)^{-1}\,\lambda_{0}^{\prime}
	\nonumber \\ & \qquad
	-\lambda_0\,(\lambda_{0}^{\prime}\lambda_0)^{-1}\,(f_{0}^{\prime}f_0)^{-1}\,(\lambda_{0}^{\prime}\lambda_0)^{-1}\,\lambda_{0}^{\prime}
	\, E \, \M_{f_0} \, E' \, \M_{\lambda_0}
	\nonumber \\ & \qquad
	- \M_{\lambda_0} \, E \,
	f_0 \, (f_{0}^{\prime}f_0)^{-1} \, (\lambda_{0}^{\prime}\lambda_0)^{-1} \, (f_{0}^{\prime}f_0)^{-1} \, f_{0}^{\prime}
	\, E' \, \M_{\lambda_0}
	\nonumber \\ & \qquad
	+ \lambda_0 \, (\lambda_{0}^{\prime}\lambda_0)^{-1} \, (f_{0}^{\prime}f_0)^{-1} \, f_{0}^{\prime}
	\, E' \, \M_{\lambda_0} \, E \,
	f_0 \, (f_{0}^{\prime}f_0)^{-1} \, (\lambda_{0}^{\prime}\lambda_0)^{-1} \lambda_{0}^{\prime} \, ,
	\end{align*}
	analogously
	\begin{align*}
	\M^{(1)}_{\widehat f,E} &= \, - \, \M_{f_0} \, E' \,
	\lambda_0 \, (\lambda_{0}^{\prime}\lambda_0)^{-1} \, (f_{0}^{\prime}f_0)^{-1} f_{0}^{\prime}
	\, - \, f_0 \, (f_{0}^{\prime}f_0)^{-1} \, (\lambda_{0}^{\prime}\lambda_0)^{-1} \, \lambda_{0}^{\prime}
	\, E \, \M_{f_0}   \; ,
	\nonumber \\
	\M^{(1)}_{\widehat f,k} &=
	\, - \, \M_{f_0} \,  X'_k  \,
	\lambda_0 \, (\lambda_{0}^{\prime}\lambda_0)^{-1} \, (f_{0}^{\prime}f_0)^{-1} f_{0}^{\prime}
	\, - \, f_0 \, (f_{0}^{\prime}f_0)^{-1} \, (\lambda_{0}^{\prime}\lambda_0)^{-1} \, \lambda_{0}^{\prime}
	\,  X_k \,\M_{f_0}  \; ,
	\nonumber \\
	\M^{(2)}_{\widehat f,E} &=
	\M_{f_0} \, E' \, \lambda_0 \, (\lambda_{0}^{\prime}\lambda_0)^{-1} \, (f_{0}^{\prime}f_0)^{-1} f_{0}^{\prime}
	\, E' \, \lambda_0 \, (\lambda_{0}^{\prime}\lambda_0)^{-1} \, (f_{0}^{\prime}f_0)^{-1} f_{0}^{\prime}
	\nonumber \\ & \qquad
	+f_0 \, (f_{0}^{\prime}f_0)^{-1} \, (\lambda_{0}^{\prime}\lambda_0)^{-1} \, \lambda_{0}^{\prime}
	\, E \, f_0 \, (f_{0}^{\prime}f_0)^{-1} \, (\lambda_{0}^{\prime}\lambda_0)^{-1} \, \lambda_{0}^{\prime}
	\, E \, \M_{f_0}
	\nonumber \\ & \qquad
	- \M_{f_0} \, E' \, \M_{\lambda_0} \, E \,
	f_0\,(f_{0}^{\prime}f_0)^{-1}\,(\lambda_{0}^{\prime}\lambda_0)^{-1}\,(f_{0}^{\prime}f_0)^{-1}\,f_{0}^{\prime}
	\nonumber \\ & \qquad
	-f_0\,(f_{0}^{\prime}f_0)^{-1}\,(\lambda_{0}^{\prime}\lambda_0)^{-1}\,(f_{0}^{\prime}f_0)^{-1}\,f_{0}^{\prime}
	\, E' \, \M_{\lambda_0} \, E \, \M_{f_0}
	\nonumber \\ & \qquad
	- \M_{f_0} \, E' \,
	\lambda_0 \, (\lambda_{0}^{\prime}\lambda_0)^{-1} \, (f_{0}^{\prime}f_0)^{-1} \, (\lambda_{0}^{\prime}\lambda_0)^{-1} \, \lambda_{0}^{\prime}
	\, E \, \M_{f_0}
	\nonumber \\ & \qquad
	+ f_0 \, (f_{0}^{\prime}f_0)^{-1} \, (\lambda_{0}^{\prime}\lambda_0)^{-1} \, \lambda_{0}^{\prime}
	\, E \, \M_{f_0} \, E' \,
	\lambda_0 \, (\lambda_{0}^{\prime}\lambda_0)^{-1} \, (f_{0}^{\prime}f_0)^{-1} f_{0}^{\prime} \, .
	\end{align*}
\end{lemma}

\begin{proof}[\bf Proof]
This lemma is a restatement of Theorem S.9.1 in the supplementary appendix of \cite{MoonWeidner2017},
and the proof is given there. However, in the presentation here we split the remainder terms of the expansions
into two components, e.g.\ $ \M_{\widehat \lambda,E}^{({\rm rem})}  + \M_{\widehat \lambda}^{({\rm rem})}(\beta)$,
where $ \M_{\widehat \lambda,E}^{({\rm rem})}$ summarizes all higher order expansion terms depending on $E$
only, and $ \M_{\widehat \lambda}^{({\rm rem})}(\beta)$ summarizes all higher order terms also involving $\beta-\beta_0$.
The reason for this change in presentation is that we will   consider differences of the form
$\M_{\widehat \lambda}(\beta_1) - \M_{\widehat \lambda}(\beta_2)$ below, and the remainder
terms $ \M_{\widehat \lambda,E}^{({\rm rem})}$  cancel in those differences.
\end{proof}

\begin{proof}[\bf Proof of Theorem~\ref{thm:post.matrix.lasso}]
    \# The first statement of the theorem is an almost immediate consequence of 
     Theorem 4.1 in \cite{MoonWeidner2017}. That theorem shows that, under the assumptions we impose here,
  we have the following approximate quadratic expansion of the profile LS objective function,
  \begin{align*}
  L_{R_0}(\beta)   &= L_{R_0}(\beta_0)
                       \, - \, \frac 1 {\sqrt{NT}}  \, (\beta-\beta_0)' \, C_{NT}
                        + \,\frac 1 2 \, (\beta-\beta_0)' \, W_{NT} \, (\beta-\beta_0) + \frac 1 {NT} \, R_{NT}(\beta) \; ,
\end{align*}
where the remainder $R_{NT}(\beta)$ is such that for any sequence
$r_{NT}\rightarrow 0$ we have
\begin{align*}
  \sup_{\beta \in {\cal B}(\beta_0, r_{NT})}   \frac{ \left|
  R_{NT}(\beta) \right| } { \left( 1 + \sqrt{NT} \, \left\| \beta -\beta_0 \right\| \right)^2 } = o_{p}\left( 1 \right) ,
\end{align*}
and  $W_{NT} =  \frac 1 {NT} \, x' \left( \M_{f_0} \otimes 
\M_{\lambda_0} \right) x$,
and $C_{NT} = C^{(1)}_{NT}  + C^{(2)}_{NT}$,
with $C^{(1)}_{NT}  = \frac 1 {NT} \, x' \left( \M_{f_0} \otimes 
\M_{\lambda_0} \right) x$,
and the $K$-vector $C^{(2)}_{NT}$ has entries, $k=1,\ldots,K$,
\begin{align*}
  C^{(2)}_{NT,k} &= - \, \frac 1 {\sqrt{NT}} \, \bigg[
       {\rm Tr}\left(E M_{f_0} \, E' \, M_{\lambda_0} \, X_k \,
              f_0 \, (f_0^{\prime}f_0)^{-1} \, (\lambda_0^{\prime}\lambda_0)^{-1} \, \lambda_0^{\prime} \right)
    \nonumber \\ & \qquad \qquad \quad
       +{\rm Tr}\left(E^{\prime}M_{\lambda_0} \, E \, M_{f_0} \, X^{\prime}_k \,
              \lambda_0 \, (\lambda_0^{\prime}\lambda_0)^{-1} \, (f_0^{\prime}f_0)^{-1} \, f_0^{\prime} \right)
    \nonumber \\ & \qquad \qquad \quad
       +{\rm Tr}\left(E^{\prime}M_{\lambda_0} \, X_k \, M_{f_0} \, E^{\prime}
                \, \lambda_0 \, (\lambda_0^{\prime}\lambda_0)^{-1} \, (f_0^{\prime}f_0)^{-1} \, f_0^{\prime} \right)
                        \bigg]  \; .
\end{align*}
We have assumed that $\plim_{N,T \rightarrow \infty} W_{NT}>0$ and $ C^{(1)}_{NT}   = \Op(1)$,
and using our assumptions (i) and (ii) we also find that
\begin{align*}
    \left| C^{(2)}_{NT,k}   \right|
    \leq  \frac {3 \, R_0} {\sqrt{NT}}  \left\| E \right\|_\infty^2    \left\| X_k \right\|_\infty
    \left\|  \lambda_0 \right\|_\infty  \left\| f_0 \right\|_\infty
      \left\|  (\lambda_0^{\prime}\lambda_0)^{-1} \right\|_\infty    \left\|  (f_0^{\prime}f_0)^{-1} \right\|_\infty  = \Op(1) ,
\end{align*}
and therefore $C_{NT} = 0$. From this approximate quadratic expansion we conclude that $  L_{R_0}(\beta) $
has indeed at least one local minimizer  within ${\cal B}(\beta_0, r_{NT})$, and that any such local minimizer 
within ${\cal B}(\beta_0, r_{NT})$ satisfies
\begin{align*}
   \sqrt{NT} \left( \widehat{\beta}^{\rm \, local}_{{\rm LS},R_0} - \beta_0 \right) =  W_{NT}^{-1} \, C_{NT}  = \Op(1) .
\end{align*}

\bigskip

\# Next, we want to show the second statement of the theorem.
Let  $\widehat{\lambda} := \widehat{\lambda} \left( \widehat{\beta}^{\rm \, local}_{{\rm LS},R_0} \right)$
and $\widehat{f} := \widehat{f} \left( \widehat{\beta}^{\rm \, local}_{{\rm LS},R_0} \right)$.
By definition we have
$\widehat{\lambda}^{(s+1)} =   \widehat{\lambda} \left( \widehat{\beta}^{(s)}  \right)$
and  $\widehat f^{(s+1)} =   \widehat f\left( \widehat{\beta}^{(s)}  \right)$, and
\begin{align*}
\left( x' \left( \M_{\widehat{f}^{(s+1)}} \otimes 
\M_{\widehat{\lambda}^{(s+1)}} \right) x \right) 
\widehat{\beta}^{(s+1)} &=  
   x' \left( \M_{\widehat{f}^{(s+1)}} \otimes 
\M_{\widehat{\lambda}^{(s+1)}} \right) y ,
  \\
  \left( x' \left( \M_{\widehat{f}} \otimes 
\M_{\widehat{\lambda}} \right) x \right) 
 \widehat{\beta}^{\rm \, local}_{{\rm LS},R_0}   &=  
   x' \left( \M_{\widehat{f}} \otimes 
\M_{\widehat{\lambda}} \right) y .
\end{align*}
By taking the difference of those last equations we obtain
\begin{align*}
  &     \left( x' \left( \M_{\widehat{f}} \otimes 
\M_{\widehat{\lambda}} \right) x \right) 
 \left( \widehat{\beta}^{(s+1)}  - \widehat{\beta}^{\rm \, local}_{{\rm LS},R_0}   \right) 
  = 
  x' \left( 
  \M_{\widehat{f}^{(s+1)}} \otimes 
\M_{\widehat{\lambda}^{(s+1)}}
-
  \M_{\widehat{f}} \otimes 
\M_{\widehat{\lambda}} \right) \left( y - x \, \widehat{\beta}^{(s+1)}   \right)
\\
 & \qquad  \qquad \qquad \qquad =   x' \left( 
  \M_{\widehat{f}^{(s+1)}} \otimes 
\M_{\widehat{\lambda}^{(s+1)}}
-
  \M_{\widehat{f}} \otimes 
\M_{\widehat{\lambda}} \right) \left[ e - x \left( \widehat{\beta}^{(s+1)} - \beta_0 \right) + ( f_0 \otimes \lambda_0) {\rm vec}(\mathbf{I}_{R})  \right] ,
\end{align*}
where in the last step we plugged in the model for $y$. Applying Lemma~\ref{lemma:expansion.operators},
the result from the first part of the theorem, and our assumptions we find that
\begin{align*}
    \frac 1 {NT} \,  x' \left( \M_{\widehat{f}} \otimes 
\M_{\widehat{\lambda}} \right) x
    &=
     \frac 1 {NT} \, x' \left( \M_{f_0} \otimes 
\M_{\lambda_0} \right) x
   + o_P(1) ,
\end{align*}
and since the probability limit of  $ \frac 1 {NT} \, x' \left( \M_{f_0} \otimes 
\M_{\lambda_0} \right) x$ is assumed to be invertible we obtain
\begin{align*}
  \widehat{\beta}^{(s+1)}  - \widehat{\beta}^{\rm \, local}_{{\rm LS},R_0}   
  &=
    \left[   \frac 1 {NT} \, x' \left( \M_{f_0} \otimes  \M_{\lambda_0} \right) x  \right]^{-1} 
  \frac 1 {NT}    \,  x' \left( 
  \M_{\widehat{f}^{(s+1)}} \otimes  \M_{\widehat{\lambda}^{(s+1)}} - \M_{\widehat{f}} \otimes  \M_{\widehat{\lambda}} \right) 
  \\
  & \qquad   \qquad  \qquad   \qquad    \qquad    \qquad 
  \times \left[ e - x \left( \widehat{\beta}^{(s+1)} - \beta_0 \right) + ( f_0 \otimes \lambda_0) {\rm vec}(\mathbf{I}_{R})  \right]
  [ 1 + o_P(1) ] .
\end{align*}
Again applying  Lemma~\ref{lemma:expansion.operators} and our assumptions one can show that
\begin{align*}
    \left\| 
         \M_{\widehat{f}^{(s+1)}} \otimes  \M_{\widehat{\lambda}^{(s+1)}} - \M_{\widehat{f}} \otimes  \M_{\widehat{\lambda}}
    \right\|_\infty
    &= \Op\left(  \left\| \widehat{\beta}^{(s)}  - \widehat{\beta}^{\rm \, local}_{{\rm LS},R_0}    \right\|   \right) ,
\end{align*}
and therefore
\begin{align*}
      \frac 1 {NT}    \,  x' \left( 
  \M_{\widehat{f}^{(s+1)}} \otimes  \M_{\widehat{\lambda}^{(s+1)}} - \M_{\widehat{f}} \otimes  \M_{\widehat{\lambda}} \right)  e
  &=  \Op\left(  \frac{ \|E\|_\infty \max_k \|X_k\|_\infty } {NT} 
   \left\| \widehat{\beta}^{(s)}  - \widehat{\beta}^{\rm \, local}_{{\rm LS},R_0}    \right\|   \right) 
  \\
  &=  \Op\left( \frac{    \left\| \widehat{\beta}^{(s)}  - \widehat{\beta}^{\rm \, local}_{{\rm LS},R_0}    \right\|  } {\sqrt{\min(N,T)}}  \right)  ,
\end{align*}
and
\begin{align*}  
  &    \frac 1 {NT}    \,  x' \left( 
  \M_{\widehat{f}^{(s+1)}} \otimes  \M_{\widehat{\lambda}^{(s+1)}} - \M_{\widehat{f}} \otimes  \M_{\widehat{\lambda}} \right) 
   x \left( \widehat{\beta}^{(s+1)} - \beta_0 \right) 
   = \Op\left(  \left\| \widehat{\beta}^{(s)}  - \widehat{\beta}^{\rm \, local}_{{\rm LS},R_0}    \right\| 
   \left\| \widehat{\beta}^{(s+1)}  - \beta_0  \right\|    \right)  
   \\ 
   & \qquad = \Op\left(  \left\| \widehat{\beta}^{(s)}  - \widehat{\beta}^{\rm \, local}_{{\rm LS},R_0}    \right\| 
   \left\| \widehat{\beta}^{(s+1)}  -  \widehat{\beta}^{\rm \, local}_{{\rm LS},R_0}   \right\|  
   +  \frac{  \left\| \widehat{\beta}^{(s)}  - \widehat{\beta}^{\rm \, local}_{{\rm LS},R_0}    \right\| } {\sqrt{NT}}   \right) , 
\end{align*}
where in the last step we used that part of the theorem implies that
$\widehat{\beta}^{(s+1)}  - \beta_0 =  \widehat{\beta}^{(s+1)}  -  \widehat{\beta}^{\rm \, local}_{{\rm LS},R_0} 
+ \Op(1/\sqrt{NT})$.
Finally, using one more time Lemma~\ref{lemma:expansion.operators} and our assumptions we can also show that
\begin{align*}
&  \frac 1 {NT}    \,  x' \left( 
  \M_{\widehat{f}^{(s+1)}} \otimes  \M_{\widehat{\lambda}^{(s+1)}} - \M_{\widehat{f}} \otimes  \M_{\widehat{\lambda}} \right) 
  ( f_0 \otimes \lambda_0) {\rm vec}(\mathbf{I}_{R})
\\
& \qquad  \qquad \qquad \qquad \qquad \qquad \qquad =  \Op\left(
  \left\| \widehat{\beta}^{(s)}  - \widehat{\beta}^{\rm \, local}_{{\rm LS},R_0}    \right\|^2 +
 \frac{    \left\| \widehat{\beta}^{(s)}  - \widehat{\beta}^{\rm \, local}_{{\rm LS},R_0}    \right\|  } {\sqrt{\min(N,T)}}  \right).
\end{align*}
Combining the above gives
\begin{align*}
     & \widehat{\beta}^{(s+1)}  - \widehat{\beta}^{\rm \, local}_{{\rm LS},R_0}   
      \\
      &= \Op\left\{
      \left\| \widehat{\beta}^{(s)}  - \widehat{\beta}^{\rm \, local}_{{\rm LS},R_0}    \right\| 
      \left[
   \left\| \widehat{\beta}^{(s+1)}  -  \widehat{\beta}^{\rm \, local}_{{\rm LS},R_0}   \right\|  
   +
  \left\| \widehat{\beta}^{(s)}  - \widehat{\beta}^{\rm \, local}_{{\rm LS},R_0}    \right\|  +
 \frac{  1  } {\sqrt{\min(N,T)}} \right]  \right\}  [ 1 + o_P(1) ] .
\end{align*}
Starting from the assumptions 
 $\| \widehat{\beta}^{(0)} - \beta_0 \| = \Op(c_{NT})$, for  $c_{NT}  \rightarrow 0$,
we thus conclude that 
$$
         \left\| \widehat{\beta}^{(1)} -  \widehat{\beta}^{\rm \, local}_{{\rm LS},R_0} \right\| = \Op \left\{ c_{NT} \left( c_{NT} +  \frac 1 {\sqrt{\min (N,T )}} \right) \, \right\} ,
$$
and then also
$$
         \left\| \widehat{\beta}^{(2)} -  \widehat{\beta}^{\rm \, local}_{{\rm LS},R_0} \right\| = \Op \left\{ c_{NT} \left( c_{NT} +  \frac 1 {\sqrt{\min (N,T )}} \right)^2 \, \right\} ,
$$
and by induction over $s$ we conclude in this way that
$$
         \left\| \widehat{\beta}^{(s)} -  \widehat{\beta}^{\rm \, local}_{{\rm LS},R_0} \right\| = \Op \left\{ c_{NT} \left( c_{NT} +  \frac 1 {\sqrt{\min (N,T )}} \right)^{s} \, \right\} .
$$
\end{proof}

\subsection{Proofs for Appendix~\ref{sec:Nonlinear:appendix}}

\begin{proof}[\bf Proof of Theorem~\ref{th:Nonlinear}]
	Like in the previous section, let $Q_{\psi}(\beta):= \min_{\Gamma} Q_{\psi}(\beta,\Gamma)$. Let $\mathbf{B}_{\psi}(M) := \{ \beta \in \mathbf{B} : \| \beta - \beta_0 \| = M \psi^{1/2} \}$ be the restricted parameter set consisting of $\beta$'s whose distance to $\beta_0$ is less than or equal to $M \psi^{1/2}$.
	In the special case where $\beta$ is a scalar (i.e., $K=1$) which is assumed in the theorem, $\mathbf{B}_{\psi}(M)$ is a finite discrete set consisting of two points,
	\begin{equation}
	\mathbf{B}_{\psi}(M) =  \{ \beta_0 - M \psi^{1/2}, \, \beta_0 + M \psi^{1/2}  \} .
	 \label{eq.ap.nolinear.parameter.set}
	\end{equation}		
	Since $Q_{\psi}(\beta)$ is convex, if we show that there exists a finite constant $M$ such that  
	\begin{equation} 
	\min_{\beta \in \mathbf{B}_{\psi}(M)} Q_{\psi}(\beta) - Q_{\psi}(\beta_0,\Gamma_0) > 0 \quad {\rm wp1}, \label{eq.ap.required.thm.nonlinear}
	\end{equation}
	then we can deduce 
	\[ 
	\| \widehat{\beta}_{\psi} - \beta_0 \| \leq M \psi^{1/2} \quad {\rm wp1},
	\]
	which is required for the theorem.
	
	For (\ref{eq.ap.required.thm.nonlinear}), we find a function $Q^*_\psi(\beta,\Gamma)$ 
	such that $ Q_\psi(\beta,\Gamma) \geq Q^*_\psi(\beta,\Gamma)$ for all $\beta,\Gamma$. With $Q^*_\psi(\beta) := \min_{\Gamma} Q^*_\psi(\beta,\Gamma)$, we show that there exists a finite constant $M$ such that  
	\begin{equation} 
	\min_{\beta \in \mathbf{B}_{\psi}(M)} Q^*_{\psi}(\beta) - Q^*_{\psi}(\beta_0,\Gamma_0) > 0 \quad {\rm wp1}. \label{eq.ap.required.2.thm.nonlinear}
	\end{equation}

	\# \underline{A lower bound objective function, $Q^*_\psi(\beta,\Gamma)$:} 
	For every pair $i,t$ we define the function $m_{it}^* : \mathbb{R} \rightarrow \mathbb{R}$ as 
	the function that satisfies
	\begin{align*}
	m_{it}^*(z_{it}^0) &=  m_{it}(z_{it}^0) ,
	&
	\partial_z  m_{it}^*(z_{it}^0) &=     \partial_z  m_{it}(z_{it}^0)  ,
	&
	\forall z \in \mathbb{R}: 
	\;
	\frac{ \partial_{z^2}  m_{it}^*(z) } b &=  \frac{\partial_{z^2} m_{it}(z) } {  \partial_{z^2} \overline m_{it}(z) } 
	.
	\end{align*}
	Here, the last condition on the second derivative should be interpreted in terms of ``generalized functions''
	in cases where $m_{it}(z)$ is not twice differentiable. For example,
	in the quantile regression example we have  $ m_{it}( z ) = \rho_{\tau}( Y_{it} - z )$,
	and therefore $\partial_{z^2} m_{it}(z) = \delta(Y_{it} - z )$, where $\delta(.)$ denotes the Dirac delta function.
	In general, solving for $m_{it}^*(z)$ we find that
	\begin{align*}
	m_{it}^*(z) &:= m_{it}(z_{it}^0) +   
	\left(z - z_{it}^0 \right) \, \partial_z m_{it}(z_{it}^0)
	+   b \, \int_{z_{it}^0}^z  \int_{z_{it}^0}^\zeta \frac{ \partial_{z^2} m_{it}(\xi) } {  \partial_{z^2} \overline m_{it}(\xi) } d \xi   d\zeta ,
	\end{align*} 
	where for $z< z_{it}^0$ the integral should be interpreted as $\int_{z_{it}^0}^z q(\zeta) d\zeta = - \int_z^{z_{it}^0} q(\zeta) d\zeta$,
	and analogously for the integral over $\xi$.\footnote{In the quantile regression case we have
		\begin{equation}
		\int_{z_{it}^0}^\zeta \frac{ \partial_{z^2} m_{it}(\xi) } {  \partial_{z^2} \overline m_{it}(\xi) } d \xi  
		= \frac{ \mathbbm{1}(z_{it}^0 < Y_{it} < \zeta) - \mathbbm{1}(\zeta  < Y_{it} < z_{it}^0) } {  \partial_{z^2} \overline m_{it}(Y_{it}) } .  \label{eq.ap.m*bar_it}
		\end{equation}  
	}
	Let $\overline m^*_{it}(z) = \mathbb{E}( m^*_{it}(z)  | X)$. 
	Our definition of $m_{it}^*(z)$
	together with $ \mathbb{E}[ \partial_z   m_{it}(z_{it}^0)  | X] = 0$ 
	imply that
	\begin{align*}
	\overline m^*_{it}(z) &= \overline m_{it}(z_{it}^0)    
	+ \frac b 2  \left(z - z_{it}^0 \right)^2 ,
	\end{align*}
	that is, $\overline m^*_{it}(z) $ is a quadratic function with second derivative equal to $b$. 
	Our assumption $ \partial_{z^2} \overline  m_{it}(z)  \geq b$ for all $z \in \mathcal{Z}$ (Assumption \ref{ass:NonlinearObj}(iii)) together with
	convexity of $m_{it}(z)$ (Assumption \ref{ass:NonlinearObj}(ii)) implies 
	furthermore that
	$0 \leq \partial_{z^2}  m_{it}^*(z) \leq \partial_{z^2} m_{it}(z)$.
	Therefore, $m_{it}^*(z)$ is a convex function and satisfies
	\begin{align}
	m_{it}(z) - m_{it}(z_{it}^0) \geq   m_{it}^*(z) -  m_{it}^*(z_{it}^0),
	\label{BoundMstar}
	\end{align}
	because $m_{it}^*(z_{it}^0) = m_{it}(z_{it}^0)$ and the convex function 
	$m_{it}(z)$ have a steeper curvature than the convex function $m_{it}^*(z)$ 
	everywhere.

	Next, we define
	\begin{align}
	Q^*_\psi(\beta,\Gamma) 
	&:=   \frac{1}{NT} \sum_{i=1}^N \sum_{t=1}^T \, m^*_{it} \left( X'_{it} \, \beta + \Gamma_{it}  \right)
	+   \frac{ \psi } {\sqrt{NT}} 
	\underbrace{
		\max_{\left\{ A \in \mathbb{R}^{N \times T} \big|  \|A\|_{\infty} \leq 1 \right\}}
		{\rm Tr}( \Gamma' A)
	}_{ = \| \Gamma \|_1}  ,
	\label{QpsiMaxA}
	\end{align}
	and $Q^*_\psi(\beta)  := \min_{\Gamma  \in \mathbb{R}^{N \times T}} Q^*_\psi(\beta,\Gamma) $.
	From \eqref{BoundMstar} we obtain that
	\begin{align*}
	Q_\psi(\beta,\Gamma) - Q_\psi(\beta_0,\Gamma_0)
	\geq  Q^*_\psi(\beta,\Gamma) - Q^*_\psi(\beta_0,\Gamma_0) .
	\end{align*}
	
	\noindent
	\# \underline{Additional definitions:}
	We already defined the expected objective function $\overline m_{it}(z) $ in the main text.
	We now also define the deviation from the expectation $\widetilde m_{it}(z) := m_{it}(z) - \overline m_{it}(z) $.
	We drop the argument $z$ whenever those functions and their derivatives are 
	evaluated at the true
	values $z^0_{it}$, for example,
	$m_{it} = m_{it}(z^0_{it}) $,
	$\partial_z \widetilde m_{it} = \partial_z  \widetilde m_{it}(z^0_{it}) $,
	$\partial_{z^2} \overline m_{it} = \partial_{z^2} \overline   m_{it}(z^0_{it}) $.
	We use the same notation for  $m^*_{it}(z) $, for example, 
	$\partial_{z^2} \overline m^*_{it} = \partial_{z^2} \overline   m^*_{it}(z^0_{it}) $.
	In addition, we define the $N \times T$ matrix $\Gamma^* := \Gamma_0 +  X^{(1)} \cdot \beta_{0}$,
	and we let $z_{it}(\beta) := X^{(2) \prime}_{it} \, \beta + \Gamma^*_{it}$.

	\medskip
	\noindent
	\# \underline{Deriving a lower bound on $Q^*_\psi(\beta) $ within the shrinking neighbourhood:}
	Our goal here  is to find a lower bound on $Q^*_\psi(\beta)$ that is valid within the shrinking neighbourhood of $\beta_0, \mathbf{B}_{\psi}(M)$.
	To obtain such a lower bound we choose the matrix $A$ in equation \eqref{QpsiMaxA} to be  the $N \times T$ matrix $A(\beta)$ 
	with elements    
	\begin{align*}
	A_{it}(\beta) &:= - \frac 1 {\sqrt{NT} \psi} \;  \partial_z m^*_{it} \left( z_{it}(\beta)  \right) 
	\\
	&=      - \frac 1 {\sqrt{NT} \psi} 
	\left[ \partial_z \widetilde m^*_{it} \left( z_{it}(\beta) \right) 
	+ \partial_z \overline m^*_{it} \left( z_{it}(\beta) \right) 
	\right]   
	\\
	&=      - \frac 1 {\sqrt{NT} \psi} 
	\left[ \partial_z \widetilde m^*_{it} \left( z_{it}(\beta) \right)  
	+ b \, X^{(2) \prime}_{it}  (\beta - \beta_0)  
	\right]   ,
	\end{align*}
	where in the final step we used that 
	$\partial_z \overline m^*_{it} \left( z \right) = \partial_z \overline m^*_{it}  + b (z - z^0_{it})$,
	and $\partial_z \overline m^*_{it}  = 0$,
	and $z_{it}(\beta)  -  z^0_{it}  = X^{(2) \prime}_{it}  (\beta - \beta_0) $.
	For the mean zero $N \times T$ matrix $\partial_z \widetilde m^*  \left( z (\beta) \right) := \left[ \partial_z \widetilde m^*_{it} \left( z_{it}(\beta) \right)   \right]  $ 
	we have 
	\begin{align*}
	\sup_{\beta \in \mathbf{B}_{\psi}(M)}  \left\| \partial_z \widetilde m^*  \left( z (\beta) \right) \right\|_{\infty}
	&\leq  \left\| \partial_z \widetilde m^*  \left( z (\beta_0 - M \psi^{1/2}) \right) \right\|_{\infty}
	+  \left\| \partial_z \widetilde m^*  \left( z (\beta_0 + M \psi^{1/2}) \right) \right\|_{\infty}
	\\
	&= \Op\left( \sqrt{\max(N,T)} \right).
	\end{align*}
	We thus find that
	\begin{align*}
	\sup_{\beta \in \mathbf{B}_{\psi}(M)}   \|A(\beta)\|_\infty
	&\leq
	o_P(1)  
	+ \sup_{\beta \in \mathbf{B}_{\psi}(M)}   o_P\left( \frac{\|\beta - \beta_0\|} {\psi^{1/2}} \right) \leq o_P(M).
	\end{align*}
	A sufficient condition for $  \|A(\beta)\|_\infty \leq 1$ wp1 uniformly in $\beta \in \mathbf{B}_{\psi}(M)$ is therefore satisfied. 
	From now on, we use $\leq_{\rm u.p.}$ to denote that the inequality holds wp1 uniformly in $\beta \in \mathbf{B}_{\psi}(M)$.
	
	Under that condition we thus have
	\begin{align}
	Q^*_\psi(\beta) 
	&\geq 
	\min_{\Gamma  \in \mathbb{R}^{N \times T}} 
	\left\{  \frac{1}{NT} \sum_{i=1}^N \sum_{t=1}^T \, m^*_{it} \left( X'_{it} \, \beta + \Gamma_{it}  \right)
	+   \frac{ \psi } {\sqrt{NT}}  
	{\rm Tr}[ \Gamma' A(\beta)]
	\right\}
	\nonumber \\
	&=
	\frac{1}{NT} \sum_{i=1}^N \sum_{t=1}^T \, m^*_{it} \left( X^{(2) \prime}_{it} \, \beta + \Gamma^*_{it}   \right)
	+   \frac{ \psi } {\sqrt{NT}}  
	{\rm Tr}[ \left(  \Gamma_0 - X^{(1)} \cdot (\beta - \beta_0) \right)' A(\beta)]
	\nonumber \\
	&\geq    \frac{1}{NT} \sum_{i=1}^N \sum_{t=1}^T \, m^*_{it} \left( X^{(2) \prime}_{it} \, \beta + \Gamma^*_{it}   \right)
	-  \frac{ \psi } {\sqrt{NT}} 
	\left\|  \Gamma_0 - X^{(1)} \cdot (\beta - \beta_0) \right\|_1
	\left\| A(\beta) \right\|_\infty
	\nonumber \\
	&\geq_{\rm u.p.}   \frac{1}{NT} \sum_{i=1}^N \sum_{t=1}^T \, m^*_{it} \left( X^{(2) \prime}_{it} \, \beta + \Gamma^*_{it}   \right)
	-  \psi \frac{ \left\|  \Gamma_0 \right\|_1 } {\sqrt{NT}}   
	-  \psi \frac{ \left\| X^{(1)} \cdot (\beta - \beta_0) \right\|_1 } {\sqrt{NT}}    
	\label{eq.ap.lowbound.Q*(beta)}
	\end{align}
	where the second line (the equality part) holds because we used that our choice of $A(\beta)$ implies that the FOC for the minimization over $\Gamma$
	are satisfied for $\Gamma = \Gamma^* -  X^{(1)} \cdot \beta = \Gamma_0 - 
	X^{(1)} \cdot (\beta - \beta_0)$. The third line holds by the H\"{o}lder 
	inequality $ |\Tr{A'B} | \leq \| A \|_{\infty} \| B \|_1$, and the last 
	line holds by the triangle inequality and  $\| A(\beta) \|_1 \leq_{\rm 
	u.p.} 1$.
	
	Next,  by expanding $X^{(2) \prime}_{it} \, \beta + \Gamma^*_{it} $ around
	$z^0_{it} = X^{(2) \prime}_{it} \, \beta_0 + \Gamma^*_{it}$ and by definition of $\overline m^*_{it}$, we obtain
	\begin{align}
	&    \frac{1}{NT} \sum_{i=1}^N \sum_{t=1}^T \, m^*_{it} \left( X^{(2) \prime}_{it} \, \beta + \Gamma^*_{it}   \right)
	\nonumber \\
	&=
	\frac{1}{NT} \sum_{i=1}^N \sum_{t=1}^T \, \overline m^*_{it} \left( X^{(2) \prime}_{it} \, \beta + \Gamma^*_{it}   \right)  
	+   
	\frac{1}{NT} \sum_{i=1}^N \sum_{t=1}^T \, \widetilde  m^*_{it} \left( X^{(2) \prime}_{it} \, \beta + \Gamma^*_{it}   \right)       
	\nonumber \\  
	&=  \frac{1}{NT} \sum_{i=1}^N \sum_{t=1}^T  \overline m^*_{it}
	+ b (\beta-\beta_0)' W (\beta-\beta_0) 
	+ O_P(1/\sqrt{NT}) , \label{eq.ap.lowbound.sum.m*it}
	\end{align}
	where the $O_P(1/\sqrt{NT})$ holds uniformly over $\beta$ in $\mathbf{B}_{\psi}(M)$.

	\medskip
	\noindent
	\# \underline{Consistency of $\widehat \beta_\psi$:}  
	
	Using the lower bounds of (\ref{eq.ap.lowbound.Q*(beta)}) and (\ref{eq.ap.lowbound.sum.m*it}), and the definition of $Q_{\psi}^*(\beta_0,\Gamma_0)$, we have
	\begin{align*}
	&\min_{\beta \in \mathbf{B}_{\psi}(M)} 
	Q_{\psi}^*(\beta) - Q_{\psi}^*(\beta_0,\Gamma_0) 
	\\
	&\geq \min_{\beta \in \mathbf{B}_{\psi}(M)}  \left[ \frac{1}{NT} \sum_{i=1}^N \sum_{t=1}^T \, \left(  m^*_{it} \left( X^{(2) \prime}_{it} \, \beta + \Gamma^*_{it}   \right) - m^*_{it} \left( X^{(2) \prime}_{it} \, \beta_0 + \Gamma^*_{it}   \right) \right) \right] 
	\\
	&\qquad \qquad 
	-  \psi  \frac{ 2 \left\|  \Gamma_0 \right\|_1 } {\sqrt{NT}}     
	-  \psi \max_{\beta \in \mathbf{B}_{\psi}(M)}
	\frac{ \left\| X^{(1)} \cdot (\beta - \beta_0) \right\|_1 } {\sqrt{NT}}   
	\\
	&= \min_{\beta \in \mathbf{B}_{\psi}(M)}  
	\left[ \frac{1}{NT} \sum_{i=1}^N \sum_{t=1}^T \, \left(  m^*_{it} \left( X^{(2) \prime}_{it} \, \beta + \Gamma^*_{it}   \right) - \overline{m}^*_{it} \right) \right]
	- O_{+}(1) \psi  - o_{P+}(1) M \psi
	\\
	&\geq b \min_{\beta \in \mathbf{B}_{\psi}(M)}   (\beta-\beta_0)' W (\beta-\beta_0) 
	- \frac{O_{P+}(1)}{\sqrt{NT}}  - O_{+}(1) \psi  - o_{P+}(1) M \psi
	\\
	&\geq  b \lambda_{min}(W) M^2 \psi 
	- \frac{O_{P+}(1)}{\sqrt{NT}}  - O_{+}(1) \psi  - o_{P+}(1) M \psi
	\\
	&\geq M \psi \left( b \lambda_{min}(W) M - \frac{O_{P+}(1)}{M \psi \sqrt{NT}} - \frac{O_{+}(1)}{M} - o_{P+}(1) \right). 
	\end{align*} 
	Since $\lambda_{min}(W) \rightarrow_p \lambda_{min}(W_{\infty}) > 0$ and $\psi \sqrt{NT} \rightarrow 0$, 
	we can choose a large constant $M$ such that 
	\[
	b \lambda_{min}(W) M - \frac{O_{P+}(1)}{M \psi \sqrt{NT}} - \frac{O_{+}(1)}{M} - o_{P+}(1) > 0 \quad {\rm wp1}.
	\]
	Then, we have the required result for the theorem. 	
\end{proof}

To establish the consistency result in the theorem in a more general case where $K > 1$, the proof requires some additional technical restrictions. The first technical requirement is the uniform bound,  $ \sup_{\beta \in \mathbf{B}_{\psi}(M)} \| A (\beta) \|_{\infty}$. For this, we may use a recent random matrix theory result in \cite{franguridi2022uniform} which requires further regularity conditions such as the tail condition of the distribution of $A_{it}(\beta)$ and a restriction of the entropy of the parameter set $\mathbf{B}_{\psi}(M)$. Secondly, we need additional technical restrictions for a uniform stochastic bound of $\sup_{\beta \in \mathbf{B}_{\psi}(M)} \frac{1}{NT} \sum_{i=1}^N \sum_{t=1}^T \, \widetilde  m^*_{it} \left( X^{(2) \prime}_{it} \, \beta + \Gamma^*_{it}   \right) = \Op(1/\sqrt{NT})$.

\end{document}